\documentclass[11pt]{article}
\usepackage{preamble}
\usepackage[para]{footmisc}

\newcommand*{\email}[1]{\href{mailto:#1}{\nolinkurl{#1}} }

\title{Quantum algorithm for stochastic optimal stopping problems with applications in finance}

\author[1]{Jo{\~a}o F. Doriguello\thanks{joaofd@nus.edu.sg}}
\author[1]{Alessandro Luongo\thanks{ale@nus.edu.sg}}
\author[1]{Jinge Bao\thanks{jbao@u.nus.edu.sg}}
\author[1]{\\Patrick Rebentrost\thanks{cqtfpr@nus.edu.sg}}
\author[1]{Miklos Santha\thanks{cqtms@nus.edu.sg}}
\affil[1]{Centre for Quantum Technologies, National University of Singapore, Singapore}

\date{\today}

\begin{document}

\maketitle

\begin{abstract}
The famous least squares Monte Carlo ({\rm LSM}) algorithm combines linear least square regression with Monte Carlo simulation to approximately solve problems in stochastic optimal stopping theory.
In this work, we propose a quantum {\rm LSM} based on quantum access to a stochastic process, on quantum circuits for computing the optimal stopping times, and on quantum techniques for Monte Carlo.
For this algorithm, we elucidate the intricate interplay of function approximation and quantum algorithms for Monte Carlo. 
Our algorithm achieves a nearly \textit{quadratic} speedup in the runtime compared to the {\rm LSM} algorithm under some mild assumptions.
Specifically, our quantum algorithm can be applied to American option pricing and we analyze a case study for the common situation of Brownian motion and geometric Brownian motion processes.
\end{abstract}


\section{Introduction}
Within stochastic optimization, optimal stopping theory is a broad area of applied mathematics that started in the 1940s and 1950s mainly with A.\ Wald~\cite{Wald:1947} and is concerned with the problem of deciding the best time to ``stop'' or take an action in order to maximize an expected reward~\cite{shiryaev2007optimal}. The time at which the observations are terminated, called \emph{stopping time}, is a random variable depending on the data observed so far in the process. A simple example of an optimal stopping problem is the following: consider a game in which $100$ numbers are written on $100$ pieces of paper without restrictions on the numbers, except that no number appears more than once. The pieces of paper are shuffled faced down and you are asked to look at the numbers, without having seen them before and \emph{once} at a time, and to stop when you think that you have found the biggest number. It turns out that there is a stopping rule that allows you to stop at the biggest number for $1/e$ fraction of the inputs. 

Since its conception, optimal stopping theory collected problems from many disparate areas under a unique umbrella~\cite{powell2019unified}, e.g.\ quickest detection~\cite{shiryaev2010quickest}, sequential parameter estimation~\cite{muravlev2020bayesian} and sequential hypothesis testing~\cite{daskalakis2017optimal}. Probably the most famous optimal stopping problem is the one of option pricing in finance, especially American options~\cite{longstaff2001valuing}. A central problem in the world of finance is to assign a monetary value to a hitherto unvalued asset. 
In the capital markets, there exists a large variety of financial assets which are derivative to underlying assets such as stocks, bonds, or commodities. One of the most well-known examples is the European call option which allows the buyer to ``lock in" a price for buying a stock at some future time (or ``exercise'' time). A fair valuation of such an option was first discussed in the seminal works of Black and Scholes~\cite{Black1973} and Merton~\cite{Merton1973}. In the times since,
the methods proposed in these works have become standard practice in the financial sector and have been extended and generalized for many financial derivatives and market models.

American options allow the buyer to exercise the option at any point in time between the time of purchase and a fixed final time. 
In contrast to an European option, there are no known closed formulas for the price of an American option with finite maturity date even in simple models like the Black-Scholes-Merton one. Theoretically, an American option can be viewed as a stochastic optimal stopping problem for the buyer and a super-martingale hedging problem for the seller~\cite{follmer2016stochastic}. 
Practical algorithms have been developed for the pricing of American options~\cite{cox1979option,rendleman1979two,jaillet1990variational,huang2003option}, an important class being least squares Monte Carlo ({\rm LSM}) algorithms originally proposed independently by Tsitsiklis and Van Roy~\cite{tsitsiklis2001regression} and by Longstaff and Schwartz~\cite{longstaff2001valuing}.

Among the aforementioned classical algorithms for option pricing --- and other topics in finance --- is the sub-field of quantum computing of designing quantum algorithms in the context of financial problems~\cite{orus2019quantum,bouland2020prospects,egger2020quantum}, e.g.\ risk management~\cite{Woerner2018}, financial greeks~\cite{stamatopoulos2021towards,an2021quantum}, portfolio optimization~\cite{barkoutsos2020improving,dasgupta2019quantum,hodson2019portfolio,rebentrost2018quantum,han2022quantum,alcazar2020classical} and option pricing~\cite{Martin2019,vazquez2021efficient}. 
A common tool in obtaining a quantum advantage is amplitude estimation~\cite{brassard2002quantum} and its generalizations for Monte Carlo sampling~\cite{montanaro2015quantum,hamoudi2021quantum,cornelissen2021quantum,cornelissen2021bquantum}. A few different works devised quantum algorithms for derivative pricing based on quantum subroutines for Monte Carlo~\cite{Rebentrost2018finance,Stamatopoulos2019,chakrabarti2021threshold}, e.g.\ European~\cite{ramos2021quantum,fontanela2021quantum,radha2021quantum} and American/Bermudan~\cite{miyamoto2021bermudan} option pricing, and option pricing in the local volatility model~\cite{kaneko2020quantum,an2021quantum} (of which the Black-Scholes model is a subcase). Given its versatility and previous cases of success, it is only natural to explore the applicability of quantum methods for Monte Carlo to problems in optimal stopping theory. 
In this work we focus on tailoring these methods to {\rm LSM} algorithms.

\subsection{Problem statement and past results}

Optimal stopping theory is concerned with the problem of finding the best moment to stop a process in order to maximize an expected reward. More generally, assume a stochastic process (which corresponds to the market model in financial applications) modelled with a Markov chain $(X_t)_{t=0}^T$ with sample space $\Omega$ and state space $E \subseteq \mathbb{R}^d$, i.e., the random variables defined as $X_t : \Omega \to E$. 
Each element $X_t$ for $t\in\{0,\dots,T\}$ gives rise to an image probability measure $\rho_t$ in $E\subseteq\mathbb{R}^d$. Let $L^2(E,\rho_t)=L^2(\rho_t)$ be the set of squared integrable functions with norm $\|f\|_{L^2(\rho_t)} = \sqrt{\mathbb{E}_{\rho_t}[|f(X_t)|^2]}$.
Assume further a \emph{payoff process}: another stochastic process $(Z_t)_{t=0}^T$ obtained from $(X_t)_{t=0}^T$ as $Z_t = z_t(X_t)$ for square-integrable real functions $z_t\in L^2(E)$, for $t\in\{0,\dots,T\}$, over the probability distribution of the Markov chain. The main problem is when to stop the process and take the payoff so that the expected payoff is maximized. A stopping time is a random variable $\tau:\Omega\to\{0,1,\dots,T\}$ such that $\{\omega\in\Omega|\tau(\omega)=t\}$ does not depend on future information (for $t'>t$) about the stochastic process. The optimal stopping problem then consists in finding a stopping time $\tau$ that maximizes the expected value payoff.
\begin{problem}[Optimal stopping problem]
\label{prob:prob1}
For $\epsilon_{\rm final}>0$, approximate $\mathcal U_0$ to additive accuracy $\epsilon_{\rm final}$ with high probability, where 
\begin{equation*} 
    \mathcal U_0 := \operatorname*{sup}_{{\tau}~{\rm stopping~time}}\mathbb{E}[Z_{\tau}].
\end{equation*}
\end{problem}

A more precise formulation is given in Problem \ref{defProblem}.
The quantity $\mathcal{U}_0$ is related to a stochastic process called the \emph{Snell envelope}~\cite{neveu1975discrete,follmer2016stochastic} and can be expressed as the solution of the dynamic programming~\cite{clement2001analysis,follmer2016stochastic}
\begin{align*}
    \begin{cases} 
        \tau_T = T, \\
        \tau_t = t \mathbf{1}\{Z_t \geq \E[Z_{\tau_{t+1}} | X_t]\} +  \tau_{t+1}\mathbf{1}\{Z_t < \E[Z_{\tau_{t+1}} | X_t]\}, & 0 \leq t \leq T-1.
   \end{cases}
\end{align*}
The stopping time $\tau_0$ can be shown to maximize $\mathbb{E}[Z_\tau]$ in Problem~\ref{prob:prob1}. The quantities $\mathbb{E}[Z_{\tau_{t+1}}|X_t]$ are called \emph{continuation values}. In the past, many different approaches were developed to tackle the dynamic programming above~\cite{kim1990analytic,chance2007synthesis,cox1979option,rendleman1979two,sharpe1999investments,mckean1965free,peskir2006optimal,jaillet1990variational, bensoussan2011applications,huang2003option,van1976optimal, boyle1977options,fu2001pricing,kohler2010review,broadie1997monte}. A famous approach is the least squares Monte Carlo ({\rm LSM}) by Longstaff and Schwartz~\cite{longstaff2001valuing}. The {\rm LSM} algorithm combines linear least-square regression with Monte Carlo simulation to approximate the hard-to-compute continuation values arising in the dynamic programming.

\subsubsection{Least squares Monte Carlo algorithms}

The Longstaff-Schwartz algorithm combines linear least-square regression with Monte Carlo simulation to approximating certain quantities known as continuation values arising in a discrete-time optimal stopping problem. The slightly simpler Tsitsiklis-Van Roy algorithm also tackles the problem by parametric approximations based on stochastic approximation techniques~\cite{Benveniste1990AdaptiveAA}. More specifically, the {\rm LSM} algorithm from Longstaff and Schwartz~\cite{longstaff2001valuing} solves the dynamic programming problem by means of two approximations. The first one is based on approximating the continuation values with a finite set of basis functions.
For each time step $t\in[T-1]$ we choose a set of $m$ expansion functions $e_{t,k}(x)\in L^2(E)$, $k\in[m]$, that is rich enough to approximate sufficiently well the continuation values $\mathbb{E}[Z_{\tau_{t+1}}|X_t]$ (in the $L^2$ norm) but is simple enough to make the problem tractable. Often a few low-degree polynomials are sufficient to obtain a good approximation. 
Then we can linearly decompose the continuation value function as $\mathbb{E}[Z_{\tau_{t+1}}|X_t] \approx \sum_{k=1}^m \alpha_k e_{t,k}(X_t)$, where $\alpha_k$ are projection coefficients given by $\alpha_t = A_t^{-1}b_t$, where $A_t\in\mathbb{R}^{m\times m}$ is a matrix with entries $(A_t)_{k,l} = \mathbb{E}[e_{t,k}(X_t)e_{t,l}(X_t)]$ and $b_t\in\mathbb{R}^m$ is a vector with entries $(b_t)_k = \mathbb{E}[Z_{\tau_{t+1}}e_{t,k}(X_t)]$. 
The key advantage of the first approximation is that instead of Monte Carlo estimating the continuation values $f(x)$ we can Monte Carlo estimate quantities involving the expansion functions. In particular, we can use Monte Carlo estimation to approximate the coefficients $\alpha_t$.

The second approximation is to solve the dynamic programming for the stopping times by using Monte Carlo sampling. In other words, we sample $N$ independent paths of the Markov Chain, compute the associated payoff values and solve the dynamic programming for each of these paths (using the approximation to the continuation values, which are easy to evaluate on the sample paths). The end result is a set of optimal stopping times $\{\tau^{(n)}\}_{n=1}^N$, from which the quantity $\sup_{\tau}\mathbb{E}[Z_\tau]$ is approximated as $\frac{1}{N}\sum_{n=1}^N Z_{\tau^{(n)}}$.
The {\rm LSM} algorithm~\cite{longstaff2001valuing} takes as input sampling access to the Markov chain and access to the payoff and the expansion functions.

While Longstaff and Schwartz~\cite{longstaff2001valuing} offered some theoretical arguments for the convergence of their algorithm, corroborated by many numerical experiments with improved performances compared to other methods at the time, they did not present any explicit convergence result. Tsitsiklis and Van Roy~\cite{tsitsiklis2001regression}, on the other hand, proved that the estimator in their algorithm converges to an estimator of an ideal algorithm where sampling error from Monte Carlo simulation is ignored. However, no convergence rate was provided.

The first theoretical analysis of the {\rm LSM} algorithm was done by Cl\'ement et al.~\cite{clement2001analysis} and focused on finite-dimensional linear approximation spaces. They proved that the approximation error of the continuation values goes to $0$ as the number of basis function goes to infinity. Moreover, they also provided an almost sure convergence of the Monte Carlo procedure via a Central Limit Theorem, but neither rate of convergence nor sample complexity was analyzed. The generality of their results is also limited in the sense that nonlinear approximation spaces were not studied and some assumptions were placed on the underlying Markov process. The work of Stentoft~\cite{stentoft2004convergence} is placed in the framework as~\cite{clement2001analysis} and claims to have obtained convergence rates in the two-period case as both the number of basis functions and the number of simulated paths go to infinity together under some assumptions. However, as noted by~\cite{liu2019american}, the proof of~\cite{stentoft2004convergence} contains some errors. 

Glasserman and Yu~\cite{glasserman2004number} were the first to regenerate the sample paths at each time step in order to simplify the error analysis. Moreover, they also considered quasi-regression instead of regression. Explicit results were obtained when taking polynomial functions as basis functions and considering Brownian motion and geometric Brownian motion as underlying process. Their analysis was later extended by Gerhold~\cite{gerhold2011longstaff} for different L\'evy models. 

Egloff~\cite{egloff2005monte} unified the Longstaff-Schwartz and the Tsitsiklis-Van Roy algorithms by proposing a new class of algorithms named dynamic look-ahead, and managed to obtain major advances in the error analysis. By leveraging the empirical risk minimization framework of Vapnik and Chervonenkis~\cite{vapnik1971uniform}, Egloff provided error estimates and sample complexity results by upper-bounding the expected sample error by $O(\sqrt{\log(N)/N})$, where $N$ is the number of sampled paths, for a variety of approximation spaces, including nonlinear ones. His analysis requires that the payoff functions be bounded, and that the approximation space be uniformly bounded, closed, and convex in some $L^p$ space, $2\leq p \leq \infty$, and possess finite Vapnik–Chervonenkis (VC) dimension.

Following the results of Egloff~\cite{egloff2005monte}, a series of works from Zanger slowly dropped previous assumptions and generalized the error analysis. Starting in~\cite{zanger2009convergence}, Zanger removed the convexity and closeness conditions on the approximation space from Egloff~\cite{egloff2005monte}, but obtained a worse expected $L^2$ sample error of $o(1/N^{1/4})$. The uniform boundedness of the payoff and approximation functions were still presumed, together with finite VC dimension. In~\cite{zanger2013quantitative}, Zanger managed to drop the boundedness assumption on the approximation functions, extending previous results to arbitrary sets of $L^2$ functions with finite VC dimension, together with an improved $O(\sqrt{\log(N)/N})$ expected $L^2$ sample error. However, like~\cite{glasserman2004number,gerhold2011longstaff}, regenerating new sets of independent sample paths at each time step was needed. Such assumption was later dropped in~\cite{zanger2018convergence}, but at the expense of again requiring uniformly bounded approximation functions. Finally, in~\cite{zanger2020general} Zanger proved an $O(\sqrt{\log(N)/N})$ expected $L^2$ sample error by using a single set of independent sample paths. The approximation space was assumed to be an arbitrary set of $L^2$ functions with finite VC dimension. Further, the payoff functions were assumed to be bounded in $L^p$, $2<p<\infty$, but not necessarily under the uniform norm. Moreover, this work also considered the case in which the least-square regression step is solved only approximately. The results from~\cite{zanger2020general} are the most general ones that we are aware of which apply to the original version of the {\rm LSM} algorithm introduced in~\cite{longstaff2001valuing}.

Apart from the aforementioned works, Klimek and Pitera~\cite{klimek2018least} investigated in the same framework of Cl\'ement et al.~\cite{clement2001analysis} but considered the case when the underlying process is non-Markovian and the payoff is path-dependent. Liu et al.~\cite{liu2019american} revised the proof of Cl\'ement et al.~\cite{clement2001analysis} employing quasi-regression and proved that the mean sample error asymptotically scales as $1/\sqrt{N}$.

\subsubsection{Applications of {\rm LSM} algorithm}

Among the whole domain of optimal stopping problems, there are many that can be approached directly with {\rm LSM}, e.g.\ the secretary problem~\cite{crosby2017optimal}, modelling the optimal time to call an election based on data~\cite{tonkes2012longstaff}, estimating the solvency of governments with respect to their debt~\cite{sauli2013suitability}, and multi-armed bandit problems~\cite{gutin2018practical}. 
Another important application of {\rm LSM} is in the insurance sector. In fact, {\rm LSM} can be used to estimate the VaR (Value at Risk)~\cite{krah2018least} and life insurance contracts~\cite{bacinello2009pricing} (see also~\cite{pelsser2016difference} for a comparison of {\rm LSM} with other methods). The computational challenges of this domain were further highlighted by recent European regulatory requirements~\cite{solvency2009directive,dimitrakopoulos2013least}. {\rm LSM} is also often used for solving Backward Stochastic Differential Equations (BSDE). 
Some numerical algorithms for BSDE are two-steps stochastic procedures involving a discretisation step where the solutions obtained at time $t$ of the BSDE are projected onto a space obtained from the filtrations at time $t-1$. This step involves a conditional expectation that cannot be calculated analytically, but must be estimated using some approximation procedure. The idea of applying {\rm LSM} to BSDE was first introduced in~\cite{gobet2005regression} and further developed in~\cite{gobet2011approximation, lemor2006rate}. Recently, this method has been generalized to solve two-dimensional forward-backward stochastic differential equations~\cite{li2021regression,bender2012least}.

\subsubsection{Other algorithms for option pricing}\label{subsec:nonlsmc}

{\rm LSM} is not the only type of algorithm that can be used to price American options~\cite{egloff2005monte}. Besides a few attempts to give an analytical formula under certain conditions~\cite{kim1990analytic}, the vast majority of them has been directed towards giving numerical results, which we briefly discuss in this section. A simple and well-known way of pricing American options is through the use of binomial trees. While the origins of this technique are somewhat unclear~\cite{chance2007synthesis}, the first articles that proposed the idea of binomial trees for pricing options are considered to be~\cite{cox1979option,rendleman1979two}, with the first seminal ideas proposed in the first edition of~\cite{sharpe1999investments}. McKean~\cite{mckean1965free} realized that the price of an American option can be cast as a free boundary problem~\cite{peskir2006optimal}, which is a particular partial differential equation that can be solved numerically. There is a flurry of other methods to price American options based on partial differential equations. We name a few approaches such as variational inequalities~\cite{jaillet1990variational, bensoussan2011applications}, linear complementary~\cite{huang2003option}, and those related to free boundaries~\cite{van1976optimal}. However, as noted in~\cite{egloff2005monte}, these methods often suffer from the curse of dimensionality, as they require the computing time and the storage to grow exponentially with the dimension of the underlying state space.
{\rm LSM} is also not the only Monte Carlo approach for pricing American options. One of the first works using Monte Carlo for option pricing is~\cite{boyle1977options}. Reviews of Monte Carlo and other methods for the problem of American option pricing can be found in~\cite{fu2001pricing,kohler2010review,broadie1997monte}. In contrast to giving lower bounds for the true optimal stopping value --- as the {\rm LSM} algorithm --- Rogers~\cite{rogers2002monte} proposed a method which leverages a dual problem, resulting in an upper bound for the optimal stopping value. Last but not least, semi-analytical approaches for American option pricing and optimal stopping time are also used~\cite{barone1987efficient}. 

\subsection{Our results}

We propose a quantum version of the {\rm LSM} algorithm and perform the error analysis considering a general set of expansion functions $\{e_{t,k}\}_{t,k}$. We achieve a quadratic speedup in the approximation error compared to the classical counterpart under some assumptions of smoothness on the continuation values (which are explained below).

Our quantum algorithm assumes quantum sampling access to the Markov chain and quantum query access to the functions $\{z_t:E\to\mathbb{R}\}_{t=0}^T$ and $\{e_{t,k}:E\to\mathbb{R}\}_{t\in[T-1],k\in[m]}$. Given these, it follows the classical {\rm LSM} in approximating the continuation values with a suitable set of expansion functions. However, the dynamic programming is not solved separately along different sampled paths, but in superposition along all possible stochastic processes. More specifically, at any given time $t$, the dynamic programming is solved in a backward fashion from time $T$ to $t$ by constructing a unitary that prepares the stopping times $\tau_{t}$ in superposition via the mapping $|x\rangle \ket{\bf 0} \mapsto |x\rangle|\tau_t(x)\rangle$ for all $x\in E$. This resulting state is used in the quantum subroutines for Monte Carlo to extract expectation values $\mathbb{E}[Z_{\tau_{t+1}}|X_t]$ that make up the vector $b_t$ and which are required to continue solving the dynamic programming at the next time step $t-1$. Such procedure is repeated until $t=0$, when the optimal stopping time $\tau_0$ can be computed in superposition and thus the quantity $\sup_\tau \mathbb{E}[Z_\tau]$ be finally approximated. We note that, unlike the classical {\rm LSM}, our quantum algorithm requires redoing all previous dynamic programming steps before a given time $t$ in order to progress into the next time step $t-1$. The final procedure involves $O(T^2)$ time steps instead of $O(T)$.

We prove that the classical {\rm LSM} algorithm and our proposed quantum {\rm LSM} algorithm approximate the sought-after quantity $\mathcal{U}_0$ up to additive accuracy with high probability. In the following, define $\vec{e}_t(\cdot) := (e_{t,1}(\cdot),\dots,e_{t,m}(\cdot))^{\top}$.
\begin{theorem}[Informal version of Corollary~\ref{cor:comparison}]
    Consider a set of linearly independent functions $\{e_{t,k}:E\to\mathbb{R}\}_{k=1}^m$ for each $t\in[T-1]$ and payoff functions $\{z_t:E\to\mathbb{R}\}_{t=0}^T$. Then, for $\delta\in(0,1)$ and $\epsilon > 0$, the classical and quantum {\rm LSM} algorithms output $\widetilde{\mathcal{U}}_0$ such that
    \begin{align*}
        \operatorname{Pr}\left[|\widetilde{\mathcal{U}}_0 - \mathcal{U}_0| \geq 5^T\left(\epsilon + \operatorname*{\max}_{0< t < T} \operatorname*{\min}_{a\in\mathbb{R}^m}\|a\cdot \vec{e}_t(X_t) - \mathbb{E}[Z_{\tau_{t+1}}|X_t]\|_{L^2(\rho_t)}\right) \right] \leq \delta
    \end{align*}
    using time, respectively, $\widetilde{O}\left(\frac{Tm^4(T+m^2)}{\epsilon^2}\right)$ and $\widetilde{O}\left(\frac{T^2m^4}{\epsilon}\right)$, up to ${\rm polylog}$ terms.
\end{theorem}
The error $\epsilon$ arises from the Monte Carlo subroutines and can be made smaller by increasing the calls to the quantum inputs (or to the number of sampled paths in the classical counterpart). Compared to the classical algorithm, the number of oracle calls is quadratically less in the quantum algorithm. The quantity $\operatorname*{\min}_{a\in\mathbb{R}^m}\|a\cdot \vec{e}_t(X_t) - \mathbb{E}[Z_{\tau_{t+1}}|X_t]\|_{L^2(\rho_t)}$ appearing in the theorem above is known as \emph{approximation error}. This term arises from approximating the continuation values by the $m$ expansion functions and is a deterministic quantity implicitly dependent on $m$ and on smoothness properties of the continuation values.

In order to obtain a final additive accuracy $\epsilon_{\rm final}$ for $\widetilde{\mathcal{U}}_0$, we must resolve the implicit dependence of the approximation error on $m$. This is done by considering specific sets of expansion functions and assuming sufficiently good smoothness properties for the continuation values. More specifically, for each $t\in[T-1]$ we consider functions $\{e_{t,k}:E\to\mathbb{R}\}_{k=1}^m$ that generate the space $\mathcal{R}_q$ of all polynomials of degree at most $q$, so that $m = \binom{q+d}{d}$. We also assume that $\mathbb{E}[Z_{\tau_{t+1}}|X_t]\in C^n$, i.e., the continuation values are $n$-differentiable functions. Then it is possible to bound the approximation error $\operatorname*{\min}_{a\in\mathbb{R}^m}\|a\cdot \vec{e}_t(X_t) - \mathbb{E}[Z_{\tau_{t+1}}|X_t]\|_{L^2(\rho_t)}$ by using a Jackson-like inequality~\cite{jackson1924general} and obtain the following result.
\begin{theorem}[Informal version of Theorem~\ref{thr:quantum_result3}]
    For each $t\in[T-1]$, consider a set of linearly independent functions $\{e_{t,k}:E\to\mathbb{R}\}_{k=1}^m$ that spans the space $\mathcal{R}_q$ with $m = \binom{q+d}{d}$ and consider payoff functions $\{z_t:E\to\mathbb{R}\}_{t=0}^T$. Assume that $\mathbb{E}[Z_{\tau_{t+1}}|X_t]\in C^n$ for all $t\in\{0,\dots,T-1\}$, where $n\leq q$. Then, for $\delta\in(0,1)$ and $\epsilon>0$, if $q = \lceil(5^T/\epsilon)^{1/n}\rceil$, the classical and quantum {\rm LSM} algorithms output $\widetilde{\mathcal{U}}_0$ such that
    \begin{align*}
        \operatorname{Pr}\left[|\widetilde{\mathcal{U}}_0 - \mathcal{U}_0| \geq \epsilon \right] \leq \delta
    \end{align*}
    using time, respectively, $\widetilde{O}\left((5^T/\epsilon)^{2+6d/n}\right)$ and $\widetilde{O}\left((5^T/\epsilon)^{1+4d/n}\right)$, up to ${\rm polylog}$ terms.
\end{theorem}

We obtain similar results if the continuation values are Lipschitz continuous, in which case the time complexities are the ones from the previous theorem with $n=1$.

If the continuation values are $n$-times differentiable, for $n = \Theta(\log(5^T/\epsilon)/\log\log(5^T/\epsilon))$, then we get the sought-after quadratic improvement from $\widetilde{O}((5^T/\epsilon)^{2})$ classical runtime to $\widetilde{O}(5^T/\epsilon)$ quantum runtime, up to polylog terms. We briefly note that such smoothness conditions on the continuation values are not unusual in areas like finance. Indeed, the continuation values can even be in $C^\infty$ in some models, e.g.\ Black-Scholes~\cite{gerhold2011longstaff,tankov2004option}. 

Very recently, Miyamoto~\cite{miyamoto2021bermudan} proposed a quantum {\rm LSM} algorithm based on Chebyshev interpolation through Chebyshev nodes and obtained a complexity of $O(\epsilon^{-1}\log^d(1/\epsilon)\text{poly}\log\log(1/\epsilon))$. Our approach, in contrast, is to project $\mathbb{E}[Z_{\tau_{t+1}}|X_t]$ onto a set of polynomials and is, for this reason, much more general. Moreover, our final result is a \emph{time} complexity, while the result from~\cite{miyamoto2021bermudan} is a \emph{query} complexity on the number of unitaries called by all quantum routines for Monte Carlo. Finally, Miyamoto~\cite{miyamoto2021bermudan} assumes that the continuation values are analytical functions, i.e., are in $C^\infty$, while we only need to assume $\mathbb{E}[Z_{\tau_{t+1}}|X_t]\in C^n$ for $n = \Theta(\log(5^T/\epsilon)/\log\log(5^T/\epsilon))$ in order to recover $\widetilde{O}(\epsilon^{-1})$ up to polylog factors. One downside of our approach, though, is the presence of quantities that implicitly depend on the underlying Markov chain.

As just mentioned, the full results behind the informal theorems above involve parameters that depend on the underlying Markov chain such as the minimum singular value of the matrices $A_t$. In order to explicitly work these parameters out, we also study the case when the underlying Markov process follows Brownian motion or geometric Brownian motion and obtain a simplified version of our algorithm. In the case of Brownian motion, we choose Hermite polynomials as the functions $\{e_{t,k}:E\to\mathbb{R}\}_{k=1}^m$ for each $t\in[T-1]$, since they are orthogonal under the probability measure underlying a Brownian motion. This means that the matrices $A_t$ are just the identity. The final result is a mild reduction on the classical and quantum time complexities to $\widetilde{O}((5^T/\epsilon)^{2+4d/n})$ and $\widetilde{O}((5^T/\epsilon)^{1+7d/2n})$, respectively. For the geometric Brownian motion, we pick suitable monomials that reduce the matrices $A_t$ to Vandermonde matrices, whose minimum singular value can be bounded. We obtain the final classical and quantum complexities $e^{O((5^T/\epsilon)^{2/n})}(5^T/\epsilon)^{2+12d/n}$ and $e^{O((5^T/\epsilon)^{2/n})}(5^T/\epsilon)^{1+15d/2n}$, respectively. If the continuation values are again $n$-times differentiable for $n=\Theta(\log(5^T/\epsilon)/\log\log(5^T/\epsilon))$, then the classical and quantum complexities for the Brownian motion setting reduce to the usual $\widetilde{O}(\epsilon^{-2})$ and $\widetilde{O}(\epsilon^{-1})$, respectively, while, for the geometric Brownian motion, they reduce to $e^{O(\log^c(5^T/\epsilon))}(5^T/\epsilon)^2$ and $e^{O(\log^c(5^T/\epsilon))}(5^T/\epsilon)$ for any constant $0<c<1$. These results for the geometric Brownian motion are slightly weaker than the usual $\widetilde{O}((5^T/\epsilon)^{2})$ and $\widetilde{O}(5^T/\epsilon)$, since the bound on the minimum singular value of the matrix $A_t$ is very sensitive to the degree $q$ of the chosen monomials.

\section{Preliminaries}

Define $[n] := \{1, \dots, n\}$. Given $b\in\mathbb{R}^m$ and $A\in\mathbb{R}^{m\times m}$ for some $m\in\mathbb{N}$, let $\|b\|_2 := \sqrt{\sum_{i=1}^m b_i^2}$ and $\|A\|_2 := \sigma_{\max}(A)$ be the vector and matrix norms, respectively, where $\sigma_{\max}(A)$ is the maximum singular value of $A$. We shall denote by $\sigma_{\min}(A)$ the minimum singular value of $A$. Given an event $\mathscr{E}$, we denote by $\mathbf{1}\{\mathscr{E}\}$ the indicator function $\mathbf{1}\{\mathscr{E}\} = 1$ if $\mathscr{E}$ is true and $0$ if not. Let $\omega_\ast$ denote the matrix multiplication exponent.

Let $(\Omega, \mathcal{F}, \mathbb{P})$ be a probability space and $\{0,1,\dots, T\}$ an index set with $T \in \mathbb{N}$. For every $t\in\{0,\dots,T\}$, let $\F_t$ be a sub $\sigma$-algebra on $\Omega$. If $\F_k \subseteq \F_l$ for all $k\leq l$, then $(\Omega, \mathcal{F}, (\F_t)_{t=0}^T, \mathbb{P})$ is a filtered probability space with filtration $(\F_t)_{t=0}^T$. 
Consider a discrete-time stochastic process $\mathbf{X} = (X_t)_{t=0}^T$, assumed to be Markovian, defined on a filtered probability space $(\Omega,\mathcal{F},(\mathcal{F}_t)_{t=0}^T,\mathbb{P})$ and with state space $(E,\mathcal{E})$, where $E \subseteq \mathbb{R}^d$. The set $E$ is equipped with its natural Borel $\sigma$-algebra inherited from $\mathbb{R}^d$, which we denote as $\mathcal{E}$. For $i\in[d]$, let $X_{t,i}:\Omega \to \mathbb{R}$ be the $i$-th component of $X_t$. We shall assume that $\mathbf{X}$ is \emph{adapted} with respect to $(\mathcal{F}_t)_{t=0}^T$, meaning that each $X_t$ is $\mathcal{F}_t$-measurable. We shall also assume that $X_0 = x_0$ is deterministic and, therefore, sometimes write the Markov chain $(X_t)_{t=1}^T$ as starting from $t=1$. Each element $X_t$ for $t\in\{0,\dots,T\}$, called the underlying process at time $t$, gives rise to an image probability measure (also called pushforward measure) $\rho_t$ in $E\subseteq\mathbb{R}^d$, i.e., $\rho_t(Y) = \mathbb{P}[\omega\in\Omega:X_t(\omega)\in Y]$ for any $Y\in\mathcal{E}$ (note that $\rho_0$ is the probability measure that assigns measure $1$ to the singleton set containing $x_0$). We denote by $\mathbf{X}_t = (X_j)_{j=t}^T$ the last $T-t+1$ random variables in the stochastic process. Let $L^2(E,\rho_t)=L^2(\rho_t)$ be the set of squared integrable functions with norm $\|f\|_{L^2(\rho_t)} = \sqrt{\mathbb{E}_{\rho_t}[|f(X_t)|^2]}$. Moreover, define the uniform norm $\|f\|_u = \sup \{|f(s)| : s \in E \}$ for $f:E\to\mathbb{R}$. 

\subsection{Stochastic optimal stopping problems}

The stochastic process $(X_t)_{t=0}^T$ with state space $E\subseteq\mathbb{R}^d$ represents the randomness of the problem under consideration. 
In economics, this process can model $d$ factors driving the economy and in finance this process can represent the 
values of some $d$ underlying set of securities. For our optimal stopping problem, we are interested in evaluating some function (which we may call ``payoff function") of this stochastic process. 
More formally:
\begin{definition}[Payoff process~\cite{follmer2016stochastic}]
\label{defPayoff}
    A payoff process is a non-negative adapted process $(Z_t)_{t=0}^T$ on the filtered probability space $(\Omega, \mathcal{F}, (\mathcal{F}_t)_{t=0}^T, \mathbb{P})$.
    In this work, we consider an adapted payoff process $(Z_t)_{t=0}^T$ obtained from the Markov chain $(X_t)_{t=0}^T$ as $Z_t = z_t(X_t)$ for Borel functions $z_t(\cdot)\in L^2(E,\rho_t)$, $t\in\{0,\dots,T\}$.
\end{definition}
An example of a non-negative adapted process is the American contingent claim.
Common examples of American options are the \emph{put} and \emph{call} options on a single asset given by the $Z^{\rm put} = z^{\rm put}(X) \coloneqq \max\{0,K-X\}$ and $Z^{\rm call} = z^{\rm call}(X) \coloneqq \max\{0,X-K\}$, respectively, where the constant $K>0$ is called strike price of the option.

Next, we define the notion of a stopping time (also called Markov time). A stopping time is any random variable that selects one of the possible times $\{0,1,\dots,T\}\cup\{+\infty\}$ and satisfies a measureability condition. 
\begin{definition}[Stopping time]
    A \emph{stopping time} is a function $\tau:\Omega\to\{0,1,\dots,T\}\cup\{+\infty\}$ such that $\{\omega\in\Omega|\tau(\omega)=t\}\in\mathcal{F}_t$ for $t\in\{0,\dots,T\}$. 
    The payoff obtained by using $\tau$ is $Z_\tau(\omega) := Z_{\tau(\omega)}(\omega)$.
    Let $\mathbb{T}_{t} := \{\tau|\tau~\text{is a stopping time with}~t\leq \tau \leq T\}$ be the set of all stopping times taking values in $[t,T]$. A stopping time $\tau^\ast\in\mathbb{T}_{t}$ is called \emph{optimal} in the interval $[t,T]$ if
    \begin{align*}
        \mathbb{E}[Z_{\tau^\ast}|X_t] = \operatorname*{ess~sup}_{\tau\in\mathbb{T}_{t}}\mathbb{E}[Z_\tau|X_t].
    \end{align*}
\end{definition}
The maximization is expressed via an essential supremum such that the null sets of the probability measure do not affect the result. For more details on the essential supremum see~\cite[Appendix~A.5]{follmer2016stochastic}.
The optimal stopping problem then consists in finding a stopping time $\tau$ that maximizes the expected value payoff.

\begin{problem}[Precise formulation of Problem~\ref{prob:prob1}] 
\label{defProblem}
    Let $(Z_t)_{t=0}^T$ be a payoff process.
    For $\epsilon_{\rm final}>0$, approximate the exact value $\operatorname*{sup}_{\tau\in\mathbb{T}_{0}}\mathbb{E}[Z_\tau]$ to additive accuracy $\epsilon_{\rm final}$ with high probability.
\end{problem}
%
A well-studied solution strategy for the above problem statement is based on dynamic programming for a set of stopping times. 
A crucial concept is the \emph{Snell envelope}~\cite{neveu1975discrete,follmer2016stochastic}.
\begin{definition}[Snell envelope]
    \label{def:snell}
    Let $(\Omega, \mathcal{F},(\mathcal{F}_t)_{t=0}^T, \mathbb{P})$ be a filtered probability space and $(X_t)_{t=0}^T$ a Markov chain with state space $E\subseteq\mathbb{R}^d$. Let $(Z_t)_{t=0}^T$ be a payoff process where $Z_t = z_t(X_t)$ for some $z_t \in L^2(E,\rho_t)$, $t\in\{0,\dots,T\}$. The Snell envelope $\mathcal{U}_t: \Omega \to\mathbb{R}$ of $Z_t$ is defined as
    \[ \begin{cases} 
          \mathcal{U}_T = Z_T, \\
          \mathcal{U}_t =  \max \left\{ Z_t, \E[ \mathcal{U}_{t+1} | X_t] \right\}, & 0 \leq t \leq T-1.
       \end{cases}
    \]
\end{definition}
Define the stopping times $\tau_t := \min\{u\geq t~|~\mathcal{U}_u = Z_u\}$. The Snell envelope is related to the maximal expected payoff according to the next theorem: $\tau_t$ maximizes the expectation of $Z_\tau$ among all $\tau\in\mathbb{T}_{t}$, i.e., that $\tau_{t}$ are optimal stopping times (in their respective intervals).
\begin{theorem}[{\cite[Theorem~6.18]{follmer2016stochastic}}]
    \label{thr:snellproperties}
    The Snell envelope $\mathcal{U}_t$ of $Z_t$ satisfies
    \begin{align*}
        \mathcal{U}_t = \mathbb{E}[Z_{\tau_t}|X_t] = \operatorname*{ess~sup}_{\tau\in\mathbb{T}_{t}}\mathbb{E}[Z_\tau|X_t],
    \end{align*}
    where $\tau_t := \min\{u\geq t~|~\mathcal{U}_u = Z_u\}$. In particular, $\mathcal{U}_0 = \mathbb{E}[Z_{\tau_0}] = \operatorname*{sup}_{\tau\in\mathbb{T}_{0}}\mathbb{E}[Z_\tau] = \max\{Z_0,\mathbb{E}[Z_{\tau_1}]\}$.
\end{theorem}
Hence, finding an approximate $\mathcal U_0$ solves our Problem~\ref{defProblem}.
In order to solve the dynamic programming behind the Snell envelope, it is more convenient to recast the dynamic programming in terms of the optimal stopping times $\tau_t$ (rather than in terms of value functions) as follows.
\begin{theorem}
    \label{thr:newdynamic}
    The dynamic programming principle in Definition~{\rm \ref{def:snell}} can be recast in terms of the stopping times $\tau_t = \min\{u\geq t~|~ \mathcal{U}_u=Z_u\}$ as
    \begin{align*}
        \begin{cases} 
            \tau_T = T, \\
            \tau_t = t \mathbf{1}\{Z_t \geq \E[Z_{\tau_{t+1}} | X_t]\} +  \tau_{t+1}\mathbf{1}\{Z_t < \E[Z_{\tau_{t+1}} | X_t]\}, & 0 \leq t \leq T-1.
       \end{cases}
    \end{align*}
\end{theorem}
\begin{proof}
    The case $t=T$ is trivial. Assume then $t<T$. Note that $\mathbb{E}[Z_{\tau_{t+1}}|X_t] = 
    \mathbb{E}[\mathbb E[Z_{\tau_{t+1}}|X_{t+1}]|X_t]
    $ because of the tower property of the expectation value with the filtration generated by $X_t$. In addition, $\mathbb E[Z_{\tau_{t+1}}|X_{t+1}] = \mathcal U_{t+1}$ from Theorem~\ref{thr:snellproperties}.
    Hence, if $Z_t \geq \mathbb{E}[Z_{\tau_{t+1}}|X_t]$, then $Z_t \geq  \mathbb{E}[\mathcal{U}_{t+1}|X_t]$. This latter statement, by the definition of the Snell envelope, implies $\mathcal{U}_t = Z_t$ and then $\tau_t = t$. On the other hand, if $Z_t < \mathbb{E}[Z_{\tau_{t+1}}|X_t]$, then, $Z_t < \mathbb{E}[\mathcal{U}_{t+1}|X_t] \implies Z_t \neq \mathcal{U}_t$, and so $\tau_t = \min\{u\geq t~|~ Z_u=\mathcal{U}_u\} = \min\{u\geq t+1~|~ Z_u=\mathcal{U}_u\} = \tau_{t+1}$.
\end{proof}
The stopping time $\tau_0$ maximizes $\mathbb{E}[Z_\tau]$ in Problem~\ref{defProblem}. The conditional expectations $\mathbb{E}[Z_{\tau_{t+1}}|X_t]$ are called \emph{continuation values}.

\subsection{The least squares Monte Carlo algorithm}

The {\rm LSM} algorithm consists in solving the dynamic programming in Theorem~\ref{thr:newdynamic} by means of two approximations. The first one is to approximate the continuation values $\mathbb{E}[Z_{\tau_{t+1}}|X_t]$ using a set of measurable real-valued functions in $L^2(E,\rho_t)$, e.g.\ by projection onto a finite-dimensional set of linearly independent polynomials. Let $\mathscr{H}_0\subseteq\mathbb{R}$ and let, for $t\in[T-1]$, $\mathscr{H}_t\subseteq L^2(E,\rho_t)$ be a subset of real-valued functions on $E$, called \emph{approximation architecture} or hypothesis class, that will be used to approximate the continuation values. By approximating $\mathbb{E}[Z_{\tau_{t+1}}|X_t]$ by $f_t\in\mathscr{H}_t$ for each $t\in\{0,\dots,T-1\}$, we can write the approximate dynamic programming as
\begin{align}
    \label{eq:newdynamic}
    \begin{cases} 
           \widetilde{\tau}_T = T, \\
            \widetilde{\tau}_t = t \mathbf{1}\{Z_t \geq f_t\} +  \widetilde{\tau}_{t+1}\mathbf{1}\{Z_t < f_t\}, & 0 \leq t \leq T-1.
       \end{cases}
\end{align}
Note that $\widetilde{\tau}_t = \widetilde{\tau}_t(f_t,\dots,f_{T-1})$ depends on the approximation architecture. 

The second approximation of the algorithm is to numerically evaluate the approximations $f_t$ in $L^2(\rho_t)$ by a Monte Carlo procedure. We sample $N$ independent paths $(X_t^{(1)})_{t=0}^T,\dots,(X_t^{(N)})_{t=0}^T$ of the Markov chain $\mathbf{X} = (X_t)_{t=0}^T$ and denote by $Z^{(n)}_t = z_t(X^{(n)}_t)$ the associated payoffs conditioned on $X_t^{(n)}$, where $z_t\in L^2(E,\rho_t)$, $t\in\{0,\dots,T\}$. Write the random variables of the last $T-t+1$ elements of all the sampled Markov chains by $\mathbf{X}_t^{(N)} = (X_t^{(1)},\dots,X_t^{(N)},X_{t+1}^{(1)},\dots,X_{t+1}^{(N)},\dots,X_T^{(1)},\dots,X_T^{(N)})$. For each path, the dynamic programming in Eq.~\eqref{eq:newdynamic} is solved recursively by approximating the continuation values in $\mathscr{H}_t$ via a least-square estimator. The result is sampled stopping times $\widetilde{\tau}_{t}^{(n)}$ that $\widetilde{\tau}_{t}$ takes on each random path. We stress that, due to the recursive nature of Eq.~\eqref{eq:newdynamic}, the stopping times $\widetilde{\tau}_{t}^{(n)}$ will depend on $\mathbf{X}_t^{(N)}$, and consequently also the payoffs $Z^{(n)}_{\widetilde{\tau}^{(n)}_t} = z_{\widetilde{\tau}^{(n)}_t(\mathbf{X}_t^{(N)})}\big(X^{(n)}_{\widetilde{\tau}^{(n)}_t(\mathbf{X}_t^{(N)})}\big)$. The dependence on $\mathbf{X}_t^{(N)}$ should be clear from the context and therefore we shall simply write $Z^{(n)}_{\widetilde{\tau}^{(n)}_t}$. In summary, combining both the approximation architecture and the Monte Carlo sampling, at each $t\in[T-1]$ we take $f_t\in \mathscr{H}_t$, depending on $\mathbf{X}_t^{(N)}$, satisfying
\begin{align}
    \label{eq:minimizer}
    \frac{1}{N}\sum_{n=1}^N \left(Z^{(n)}_{\widetilde{\tau}_{t+1}^{(n)}} - f_t(X_t^{(n)}) \right)^2 \leq \epsilon + \operatorname*{\inf}_{g\in\mathscr{H}_t}\frac{1}{N}\sum_{n=1}^N \left(Z^{(n)}_{\widetilde{\tau}_{t+1}^{(n)}} - g(X_t^{(n)}) \right)^2
\end{align}
for some given $\epsilon \geq 0$. It might be the case that an exact minimizer of the above optimization problem does not exist (the infimum does not belong to $\mathscr{H}_t$) or is hard to compute, meaning that an $\epsilon$-approximation could be used. Given the choice of $f_t\in \mathscr{H}_t$, it is then used in Eq.~\eqref{eq:newdynamic} to obtain $\widetilde{\tau}_t^{(n)}$, and so on recursively. At the end of the recursion we can take $f_0 = \frac{1}{N}\sum_{n=1}^N Z^{(n)}_{\widetilde{\tau}_1^{(n)}}$ as an exact minimizer, since $X_0$ is constant, and obtain the approximation $\widetilde{\mathcal{U}}_0$ to $\mathcal{U}_0$ as
\begin{align}
    \label{eq:approximate_result}
    \widetilde{\mathcal{U}}_0 = \max\left\{Z_0,\frac{1}{N}\sum_{n=1}^N Z^{(n)}_{\widetilde{\tau}_1^{(n)}}\right\}.
\end{align}

In this paper we shall be particularly interested in finite-dimensional linear approximation architectures, for which an exact minimizer exists in Eq.~\eqref{eq:minimizer}. Consider then a set $\{e_{t,k}:E\to\mathbb{R}\}_{k=1}^m$ of $m$ linearly independent measurable real functions and take the vector space generated by them as our approximation architecture $\mathscr{H}_t$, $t\in[T-1]$. Therefore, the infimum in Eq.~\eqref{eq:minimizer} is attained by projecting the continuation values onto $\mathscr{H}_t$ as $\alpha_t\cdot \vec{e}_t(X_t)$, where $\vec{e}_t(\cdot) := (e_{t,1}(\cdot),\dots,e_{t,m}(\cdot))^{\top}$ and the $m$-dimensional vector $\alpha_t$, the projection coefficients, is the least-square estimator given by~\cite{clement2001analysis}
\begin{align*}
    \alpha_t = \arg\operatorname*{\min}_{a\in\mathbb{R}^m}\mathbb{E}\big[(Z_{\widetilde{\tau}_{t+1}} - a\cdot \vec{e}_t(X_t))^2\big].
\end{align*}
Given the assumption that $\{e_{t,k}\}_{k=1}^m$ are linearly independent for each $t\in[T-1]$, the vector $\alpha_t\in\mathbb{R}^m$ has the explicit expression
\begin{align}
    \label{eq:coefficient_vector}
    \alpha_t =A_t^{-1}b_t \text{ where }  b_t=\mathbb{E}[Z_{\widetilde{\tau}_{t+1}}\vec{e}_t(X_t)]
\end{align}
and the $m\times m$ matrix $A_t$ has coefficients
\begin{align}
    \label{eq:projection_matrix}
    (A_t)_{k,l} = \mathbb{E}[e_{t,k}(X_t)e_{t,l}(X_t)].
\end{align}
Often it is hard to compute $\alpha_t$ and $A_t$ exactly. As previously mentioned, the {\rm LSM} algorithm approximates these by Monte Carlo sampling,
\begin{align}
    \label{eq:approximate_coefficient_vector}
    \widetilde{\alpha}_t = \widetilde{A}_t^{-1}\frac{1}{N}\sum_{n=1}^N Z^{(n)}_{\widetilde{\tau}^{(n)}_{t+1}}\vec{e}_t(X_t^{(n)})
\end{align}
and
\begin{align}
    \label{eq:matrices_A}
    (\widetilde{A}_t)_{k,l} = \frac{1}{N}\sum_{n=1}^N e_{t,k}(X^{(n)}_t)e_{t,l}(X^{(n)}_t).
\end{align}
More generally, though, any good approximation $\widetilde{\alpha}_t$ and $\widetilde{A}_t$ to $\alpha_t$ and $A_t$, respectively, is valid, and we shall not restrict the notation $\widetilde{\alpha}_t$ and $\widetilde{A}_t$ to only mean the above sampled quantities.

We have introduced the quantities that are important for the 
{\rm LSM} algorithm. To present the algorithm, we first specify the input model, starting with sampling access to the Markov chain.
\begin{definition}[Sampling access to Markov chain]
    Given a Markov chain $(X_t)_{t=1}^T$ on a probability space $(\Omega, \mathbb{P})$ and with state space $E\subseteq\mathbb{R}^d$, we define sampling access as the ability to draw a sample $\omega\in\Omega$ according to $\mathbb{P}$ and observe the value $X_t(\omega)$ for some $t\in[T]$. One sample costs time $\mathcal{T}_{\rm samp}$.
\end{definition}
Furthermore, query access to a function is defined as being able to evaluate a function for a given input.
\begin{definition}[Query access to function]
    Let $E\subseteq\mathbb{R}^d$ and $h:E\to\mathbb{R}$ be a function. We define query access as the ability to observe the value $h(x)$ for any given $x\in E$. One query costs time $\mathcal{T}_h$.
\end{definition}
Here, we assume that the functions of the approximation architecture and functions for the payoff take time $\mathcal{T}_{e}$ and $\mathcal{T}_{z}$, respectively, to evaluate. Both sampling and function access have natural quantum extensions, as will be defined in Section~\ref{sec:quantum}.

We are now in the position to present the classical {\rm LSM} algorithm in Algorithm~\ref{alg:classicalalgo}. Since we focus on the case where the approximation architectures $\mathscr{H}_t$ are finite-dimensional and linear, we write Algorithm~\ref{alg:classicalalgo} for this particular case.

\begin{algorithm}[t]
	\caption{Classical {\rm LSM} algorithm for optimal stopping problem}
	\label{alg:classicalalgo}
	\begin{algorithmic}[1]
		\Require Integer $N\in\mathbb{N}$. Sampling access to Markov chain $(X_t)_{t=0}^T$ defined on a sample space $\Omega$ and with state space $E\subseteq\mathbb{R}^d$. Query access to $\{z_t:E\to \mathbb{R}\}_{t=0}^T$ and $\{e_{t,k}:E\to\mathbb{R}\}_{t\in[T-1],k\in[m]}$, where $\{e_{t,k}\}_{k=1}^m$ are linearly independent for each $t\in[T-1]$. Let $\vec{e}_t(\cdot) := (e_{t,1}(\cdot),\dots,e_{t,m}(\cdot))^\top$.

        \State Sample $N$ independent paths $(X^{(1)}_t,\dots,X^{(N)}_t)_{t=0}^T$.
        \State Query associated payoffs $(Z^{(1)}_t,\dots,Z^{(N)}_t)_{t=0}^T$ and values $(e_{t,k}(X_t^{(1)}),\dots,e_{t,k}(X_t^{(N)}))_{t\in[T-1],k\in[m]}$.
		\State Compute the matrices $\{\widetilde{A}_t\}_{t=1}^{T-1}$ with entries as in Eq.~\eqref{eq:matrices_A}.
	 \State Compute the inverses $\{\widetilde{A}_t^{-1}\}_{t=1}^{T-1}$.
	 \State Set $\widetilde{\tau}^{(n)}_T = T$ for $n\in[N]$.
\For{$t=T-1$ to $1$}
    \State Calculate the vector $\widetilde{\alpha}_t = \widetilde{A}_t^{-1}\frac{1}{N}\sum_{n=1}^N Z^{(n)}_{\widetilde{\tau}^{(n)}_{t+1}}\vec{e}_t(X_t^{(n)})$.
    \State Calculate $\widetilde{\tau}^{(n)}_t = t\mathbf{1}\{Z_t^{(n)} \geq \widetilde{\alpha}_t\cdot \vec{e}_t(X_t^{(n)})\} + \widetilde{\tau}_{t+1}^{(n)}\mathbf{1}\{Z_t^{(n)} < \widetilde{\alpha}_t\cdot \vec{e}_t(X_t^{(n)})\}$, $n\in[N]$.
\EndFor
\State Output $\widetilde{\mathcal{U}}_0 := \max\Big\{Z_0,\frac{1}{N}\sum_{n=1}^N Z^{(n)}_{\widetilde{\tau}_1^{(n)}}\Big\}$. 
 \end{algorithmic}
\end{algorithm} 

\subsection{Error analysis of least squares Monte Carlo algorithm}

As previously mentioned in the introduction, there has been a considerable amount of work in understanding the convergence and expected sample error of the {\rm LSM} algorithm, starting with Cl\'ement et al.~\cite{clement2001analysis} two decades ago until the more recent results of Zanger~\cite{zanger2020general}. We first review a general result of Zanger~\cite{zanger2020general} for any uniformly bounded approximation architecture with finite VC dimension.

\begin{theorem}[{\cite[Theorem~3]{zanger2020general}}]
    Consider the {\rm LSM} algorithm (Algorithm~{\rm \ref{alg:classicalalgo}}) with $N$ random sampled paths, $T$ time steps, and uniformly bounded approximation architectures $\mathscr{H}_t$, $t\in\{0,\dots,T-1\}$, with finite VC dimension. Assume that an exact minimizer is always chosen in Eq.~\eqref{eq:minimizer}. Then the expected error from the {\rm LSM} algorithm is
    \begin{align*}
        \mathbb{E}[|\widetilde{\mathcal{U}}_0 - \mathcal{U}_0 |]  = O\left(6^T\left(\sqrt{\frac{\log{N}}{N}} + \max_{0\leq t<T}\inf_{f\in\mathscr{H}_t}\|f - \mathbb{E}[Z_{\tau_{t+1}}|X_t]\|_{L^2(\rho_t)} \right) \right).
    \end{align*}
\end{theorem}
The term $\inf_{f\in\mathscr{H}_t}\|f - \mathbb{E}[Z_{\tau_{t+1}}|X_t]\|_{L^2(\rho_t)}$ is known in learning theory as the \emph{approximation error} of the approximation architecture, and is a deterministic quantity that conveys the minimum error that is possible from using $\mathscr{H}_t$.

Often in the study of convergence of the {\rm LSM} algorithm, previous works~\cite{clement2001analysis,egloff2005monte,zanger2020general} do not take into consideration the time complexity of finding an exact (or $\epsilon$-approximation) minimizer in Eq.~\eqref{eq:minimizer}, but instead consider only the number of sampled paths $N$ as a measure of complexity. It is thus not immediate to compare, e.g.\ the above theorem, with our quantum {\rm LSM} algorithm in the next section, since sampled paths can be reused in different parts of an algorithm, but calls to a quantum oracle generally cannot. Nonetheless, the time complexity of Algorithm~\ref{alg:classicalalgo} can be precisely analyzed, due to the closed form for the exact minimizer, and it is given in the next theorem. This result will serve as our proxy when comparing our quantum algorithm to its classical counterpart. We shall postpone the proof until the end of next section, since it is very similar to Theorems~\ref{thr:quantum_result1} and~\ref{thr:quantum_result2} on our quantum algorithm.
\begin{theorem}
    
    Within the setting of Algorithm~{\rm \ref{alg:classicalalgo}}, consider $N$ independent sample paths with sample cost $\mathcal{T}_{\rm samp}$ and let the linearly independent functions $\{e_{t,k}:E\to\mathbb{R}\}_{k=1}^m$ for each $t\in[T-1]$ be such that $L := \max_{t\in[T-1],k\in[m]} \|e_{t,k}\|_{L^2(\rho_t)}$ and have query cost $\mathcal{T}_{e}$. Also consider $\{z_t:E\to\mathbb{R}\}_{t=0}^T$ with $R := \max_{t\in[T]}\|z_t\|_u < \infty$ and query cost $\mathcal{T}_{z}$. Moreover, let $\sigma_{\min} := \min_{t\in[T-1]} \sigma_{\min}(A_t) > 0$ and assume that $\sqrt{m}RL/\sigma_{\min} \geq 1$.
    Then, for any $\epsilon\in(0,\sigma_{\min}/2]$, Algorithm~{\rm \ref{alg:classicalalgo}} runs for $O(T m^2 N + Tm^{\omega_\ast} + TN(T+ \mathcal T_{\rm samp} + \mathcal T_{z} + m \mathcal T_{e}))$ time and returns an estimate $\widetilde{\mathcal{U}}_0$ such that
    \begin{align*}
        \operatorname{Pr}\left[\big|\widetilde{\mathcal{U}}_0 - \mathbb{E}[Z_{\tau_0}]\big| \geq 5^T\left(\frac{4\epsilon mRL^2}{\sigma_{\min}^2} + \operatorname*{\max}_{0 < t < T} \operatorname*{\min}_{a\in\mathbb{R}^m}\|a\cdot \vec{e}_t(X_t) - \mathbb{E}[Z_{\tau_{t+1}}|X_k]\|_{L^2(\rho_t)}\right) \right] \leq 6m^2e^{-2N\epsilon^2/m^2}.
    \end{align*}
\end{theorem}

If sampling and query costs are constant, then the complexity is $O(Tm^2N + Tm^{\omega_\ast} + T^2N)$. The factor $Tm^2$ comes from computing $\{A_t\}_{t=1}^{T-1}$ and accounts for runtime instead of only number of samples. The term $O(T^2 N)$ comes from the mapping $\widetilde{\tau}^{(n)}_{t+1}\mapsto Z^{(n)}_{\widetilde{\tau}^{(n)}_{t+1}}$ for all $n\in[N]$ and $t\in[T-1]$.

\section{Quantum algorithm}
\label{sec:quantum}

In this section we shall present our quantum algorithm, which is based on the classical {\rm LSM} algorithm (Algorithm~\ref{alg:classicalalgo}). 
Before we discuss it, we review our computational model, input assumptions, and the quantum algorithm for Monte Carlo used in this work.
In what follows, for simplicity, we suppose that $\ket{\bf 0}$ describes a register with sufficiently many qubits initialized in the all-$0$ state.

\subsection{Computational model}

In this subsection, we address our quantum computational model. We work in the standard circuit model of quantum computation~\cite{nielsen2002quantum}. Aside from these standard assumptions, we take the following additional assumptions on the computational model. 

\paragraph*{Arithmetic model.}
In our work, we perform the arithmetic computations on the quantum computer by using a fixed point representation for real numbers. We assume that we can have enough qubits for storing these numbers, represented as bit-strings using the following definition. We also assume to work with enough precision so that numerical errors in the computation are negligible, and will not impact the final output of our algorithm. 
\begin{definition}[Fixed-point encoding of real numbers~\cite{rebentrost2021quantum}] \label{defEncoding}
Let $c_1,c_2$ be positive integers, and $a\in\{0,1\}^{c_1}$, $b \in \{0,1\}^{c_2}$, and $s \in \{0,1\}$ be bit-strings. Define the rational number as:
\begin{equation}
    \mathcal{Q}(a,b,s):= 
    (-1)^s
    \left(2^{c_1-1}a_{c_1}+ \dots + 2a_2 + a_1 + \frac{1}{2}b_1 + \dots + \frac{1}{2^{c_2}}b_{c_2} \right) \in [-R,R],
\end{equation}
where $R := 2^{c_1}-2^{-c_2}$. If $c_1,c_2$ are clear from the context, we can use the shorthand notation for a number $z:=(a,b,s)$ and write  $\mathcal{Q}(z)$ instead of $\mathcal{Q}(a,b,s)$. Given an $n$-dimensional vector $v \in (\{0,1\}^{c_1} \times \{0,1\}^{c_2} \times \{0,1\})^n$
the notation $\mathcal{Q}(v)$ means an $n$-dimensional vector whose $j$-th component is $\mathcal{Q}(v_j)$, for $j \in[n]$. 
\end{definition}

The choice of values for $c_1$ and $c_2$ depends on the choice of input functions used when running the algorithm. For the purposes of optimizing the quantum circuit, these constants can be changed dynamically in various steps of the computation. While analyzing how error propagates and accumulates throughout the operations in the quantum circuit is essential to ensure a correct estimation of the final result, this analysis can only be done for a given choice of input functions. We avoid the analysis of such details by using the quantum arithmetic model as in Definition~\ref{defQArith}. 
A standard result is that any Boolean function can be reversibly computed. Any reversible computation can be realized with a circuit involving negation and three-bit Toffoli gates. Such a circuit can be turned into a quantum circuit with single-qubit NOT gates and three-qubit Toffoli gates. Since most circuits for arithmetic operations operate with a number of gates of $O(\text{poly}(c_1,c_2))$ this implies a number of quantum gates of $O(\text{poly}(c_1,c_2))$ for the corresponding quantum circuit.
\begin{definition}[Quantum arithmetic model]
\label{defQArith}
Given $c_1, c_2 \in \mathbb{N}$ specifying fixed-point precision numbers as in Definition~{\rm \ref{defEncoding}}, we say we use a quantum arithmetic model of computation if the four arithmetic operations can be performed in constant time in a quantum computer. 
\end{definition}
In our computational model we do not include the cost for performing operations described in our arithmetic model. For instance, a central computational step of the quantum algorithm is the circuit computing the stopping times $\widetilde{\tau}_t(x)$, but as the circuit depth depends polynomially on $c_1$ and $c_2$, we do not take into account this cost when stating our runtime.
For an example of a resource estimation for a financial problem that takes into account the cost of arithmetic operations in fixed-point precision, we refer to~\cite{chakrabarti2021threshold}.

\paragraph*{Quantum input access.}
We assume that we have quantum oracles for certain input functions. 
The classical algorithm assumes access to two different kinds of oracles. The first is an oracle that allows us to obtain samples from the Markov chain $(X_t)_{t=0}^T$.
The second kind of oracle is evaluating the functions $\{z_t\}_{t=0}^T$ and $\{e_{t,k} \}_{t\in[T-1],k\in[m]}$. We assume access to the quantum versions of these oracles (formalized below). 
The first kind of quantum oracle prepares a quantum state that is in a superposition over the different outcomes of the Markov chain, weighted by amplitudes which are square roots of their classical probabilities. A measurement in the computational basis of such a state obtains a single sample with the corresponding probability and hence directly recovers a single use of the classical sampling access. 
The second kind of quantum oracle evaluates a given function in superposition over its inputs. While the functions $\{ e_{t,k} \}_{t\in[T-1],k\in[m]}$ are usually chosen to be low-degree polynomials (and thus admit efficient classical and quantum circuits with gate complexity proportional to the degree of the polynomial), the functions $\{ z_t \}_{t=0}^T$ might be arbitrarily complex. Usually the complexity of these functions is not discussed in classical literature, and we use placeholders for their evaluation cost.
\begin{definition}[Quantum sampling access to a Markov chain]
    \label{ass:markov-quantum}
    Let $(X_t)_{t=1}^T$ be a Markov chain defined on a filtered probability space $(\Omega, \mathcal{F}, (\mathcal{F}_{t})_{t=1}^T, \PP)$, for a finite $\Omega$, assuming values in a finite state space $E\subseteq\mathbb{R}^d$. Given $x\in E^T$, let $p(x) := \mathbb{P}[X_1=x_1]\prod_{t=1}^{T-1}\mathbb{P}[X_{t+1}=x_{t+1}|X_t=x_t]$ be the probability that $X_1=x_1,\dots,X_T=x_T$. Let $\mathcal{H}$ be a finite-dimensional Hilbert space with basis $\{|x\rangle\}_{x\in E^T}$. We say that we have quantum sampling access to $(X_t)_{t=1}^T$ if we are given an oracle $U$ on $\mathcal{H}$ such that $U|\mathbf{0}\rangle = \sum_{x\in E^T}\sqrt{p(x)}|x\rangle$. 
    One application of $U$ costs $T\cdot\mathcal T_{\rm samp}$ time. If $T=1$, we say that we have quantum sampling access to a random variable $X$ if we are given an oracle $U$ on $\mathcal{H}$ such that $U|\mathbf{0}\rangle = \sum_{x\in E}\sqrt{p(x)}|x\rangle$, where $p(x) := \mathbb{P}[X=x]$ is the probability that $X=x$. 
\end{definition}
We note the alternative definition of the unitary $U$ such that $U|\mathbf{0}\rangle = \sum_{x\in E^T}\sqrt{p(x)}|x\rangle|\psi_x\rangle$ for unknown garbage unit states $|\psi_x\rangle$. Such garbage states do not change our analysis, so we shall ignore them and work with the unitary $U$ from Definition~\ref{ass:markov-quantum}.

Even though we assume the existence of the oracle $U$, constructing such unitary is an important question on its own. A few methods have been proposed in order to tackle such problem, one of the most famous is due to Grover and Rudolph~\cite{grover2002creating} (see~\cite{MarinSanchez2021} for recent improvements on the Grover-Rudolph method), which loads into a quantum computer a discretization of a distribution with density function $p(x)$. More specifically, it creates the quantum state $\sum_{i=0}^{2^n-1}\sqrt{p_i^{(n)}}|i\rangle$ with $p_i^{(n)} = \int_{x_i^{(n)}}^{x_{i+1}^{(n)}}p(x)\text{d}x$ by recursively (on $n$) computing quantities like $f_n(i) = \int_{x_i^{(n)}}^{(x_i^{(n)}+x_{i+1}^{(n)})/2}p(x)\text{d}x \big/ \int_{x_i^{(n)}}^{x_{i+1}^{(n)}}p(x)\text{d}x$ for $i=0,\dots,2^n-1$. It is also possible to perform simple Taylor approximations on $f_n(i)$ when $n$ is sufficiently large (see~\cite[Equation~(35)]{kaneko2020quantum}). We briefly note that the issues about the Grover-Rudolph method recently pointed out by~\cite{herbert2021no} only arise when one needs to \emph{sample} from the distribution $p(x)$ in order to compute $f_n(i)$, which is not the case in many situations, e.g.\ in finance.
\begin{definition}[Quantum access to a function]
\label{defQComputeFunc}
    Let $E\subseteq\mathbb{R}^d$ be a finite set and let $\mathcal{H}$ be a Hilbert space with basis $\{|x\rangle\}_{x\in E}$. Given $h:E\to\{0,1\}^n$, we say that we have $(V_h, \mathcal T_h)$-quantum access to $h$ if we have access to a quantum circuit $V_h$ on $\mathcal{H}\otimes\mathbb{C}^{2^n}$ such that $V_h|x\rangle |b\rangle = |x\rangle|b \oplus h(x)\rangle$ for any bit-string $b\in\{0,1\}^n$. 
    One application of $V_h$ costs time $\mathcal{T}_h$. 
\end{definition}

\begin{center}
\begin{tabular}{ |c|c| } 
 \hline
 Computational Model & Classical / Quantum \\
 \hline
 Arithmetic operation & $O(1)$ \\ 
 Query functions $z$ &$\mathcal{T}_{z}$  \\
 Query functions $e$ &$\mathcal{T}_{e}$ \\
 Sampling a Markov chain & $T\cdot\mathcal{T}_{\rm samp}$ 
\\
\hline
\end{tabular}
\end{center}

The sampling and query costs used in this paper are summarized in the table above. We briefly note that a given cost might represent different values in different contexts, e.g.\ we shall use the notation $\mathcal{T}_{\rm samp}$ in both classical and quantum contexts, and it is implied that classical and quantum sampling costs may not be necessarily the same.

\paragraph*{Access to quantum controlled rotations.}

Controlled rotations are a central step in the
quantum algorithm for Monte Carlo (Theorem~\ref{thr:quantum_monte_carlo}).
The cost of a controlled rotation
depends directly on the number of bits used to specify the angle of rotation~\cite{WCNA2009}.
In our computational model we assume that controlled rotations come with a unit cost.
\begin{definition}[Access to quantum controlled rotations]
\label{defModelRotation}
We say that we have access to quantum controlled rotations if we have a quantum circuit $R$ whose application takes constant time and, for all  rational numbers $x \in[0,1]$ defined by a $(1 + c_2)$-bit-string in our fixed-point arithmetic model~{\rm \ref{defQArith}}, operates as: 
\begin{equation}
R \ket {x} \ket 0 = \ket { x} \left ( \sqrt{1- x} \ket 0 +  \sqrt x \ket 1 \right).
\end{equation}
\end{definition}
We note that again this definition allows us to neglect terms $O(c_1+c_2)$ in the runtime and to neglect complications arising from the arithmetic computation of the $\arcsin$. The access to these rotation unitaries leads to the following fact. 
\begin{fact}[Controlled rotations of a function with an interval]
\label{fact:query_to_rotation}
    Consider a rational number representation from Definition~{\rm \ref{defEncoding}} for some $c_1, c_2 \in \mathbb{N}$.
    Assume access to controlled rotations according to Definition~{\rm \ref{defModelRotation}}.
    Assume $(V_h,\mathcal T_h)$-quantum access to a function $h$ according to Definition~{\rm \ref{defQComputeFunc}}.
    For any two bit-strings $a,b\in\{0,1\}^n$, with $0 \leq \mathcal Q(a) < \mathcal Q(b)$, we can construct a  unitary operator $R^h_{a,b}$ on $\mathcal H \otimes \mathbb{C}^2$,  such that, for all $x\in E$,
    \begin{align*}
        R^h_{a,b}|x,0\rangle = \begin{cases}
        |x\rangle\left(\sqrt{1 - \mathcal{Q}(h(x))/\mathcal{Q}(b)}|0\rangle + \sqrt{\mathcal{Q}(h(x))/\mathcal{Q}(b)}|1\rangle\right) \quad &\text{if}~ \mathcal Q(a)\leq\mathcal Q(h(x))\leq \mathcal Q(b),\\
        |x,0\rangle \quad &\text{otherwise},
        \end{cases}
    \end{align*}
    where an application of $R^h_{a,b}$ costs $O(\mathcal T_h)$ time.
\end{fact}
\begin{proof}
The quantum circuits for the division and for checking the interval run in constant time in the quantum arithmetic model.
The quantum circuits allow us to prepare $\ket x \ket{\mathcal{Q}(h(x))/\mathcal{Q}(b)}$ on the interval using $V_h$ two times, where the second register is of size polynomial in $c_1$ and $c_2$.
Performing a controlled rotation of an ancilla costs
constant time by Definition~\ref{defModelRotation}.
\end{proof} 

\paragraph*{Quantum algorithm for Monte Carlo.} Our quantum algorithm requires the computation of several expectation values. In this work we use the quantum algorithm for Monte Carlo from Montanaro~\cite{montanaro2015quantum}, already adapted to our computational model. 
\begin{theorem}[{Quantum algorithm for Monte Carlo $\mathtt{QMonteCarlo}$, \cite[Theorem~2.5]{montanaro2015quantum}}]
    \label{thr:quantum_monte_carlo}
    Let $X$ be a random variable given via quantum sampling access as in Definition~{\rm \ref{ass:markov-quantum}}.
    Consider a rational number representation from Definition~{\rm \ref{defEncoding}} for some $c_1, c_2 \in \mathbb{N}$.
    Let $E\subseteq\mathbb{R}^d$ be a finite set and let $\mathcal{H}$ be a Hilbert space with basis $\{|x\rangle\}_{x\in E}$.
    Consider a function
    $h:E\to \{0,1\}^n$ via quantum access to the controlled rotations as in Fact~{\rm \ref{fact:query_to_rotation}}, where $n \in \mathbb{N}$ such that $c_1+c_2+1=n$, and each access costs time $\mathcal T_h$.
    Assume that the random variable $\mathcal Q (h(X))$ has finite mean $\mu$, and variance upper-bounded by $\sigma^2$ for some known $\sigma>0$. Given $\delta,\epsilon\in(0,1)$, there is a quantum algorithm, called $\mathtt{QMonteCarlo}(h(X), \epsilon, \delta, \sigma)$, that runs in $O((\sigma/\epsilon)\log(1/\delta)\log^{3/2}(\sigma/\epsilon)\log\log(\sigma/\epsilon)) \times (T\mathcal T_{\rm samp} + \mathcal T_h)$ time and outputs an estimate $\widetilde{\mu}$ such that $\operatorname{Pr}\left[|\widetilde{\mu} - \mu] \geq \epsilon\right] \leq \delta$.
\end{theorem}

The above result will be used to approximate expected values, e.g.\ $\mathbb{E}[e_{t,k}(X_t)e_{t,l}(X_t)]$ and $\mathbb{E}[Z_{\widetilde{\tau}_{t+1}}\vec{e}_t(X_t)]$, and was chosen for its simplicity. It is possible, though, to use other, more complicated quantum subroutines for Monte Carlo, e.g.~\cite{hamoudi2021quantum,cornelissen2021quantum,cornelissen2021bquantum}. Refs.~\cite{cornelissen2021quantum,cornelissen2021bquantum} propose quantum algorithms for multivariate Monte Carlo estimation, which could be particularly suitable in our case, since most of our quantities of interest are vectors and matrices. However, since these more complex and alternative quantum subroutines for Monte Carlo lead to the same time complexities up to polylogarithm factors as Theorem~\ref{thr:quantum_monte_carlo}, we decided to use the above result.

\subsection{Quantum circuits for the stopping times}

Recall that Theorem~\ref{thr:newdynamic} allows us to formulate the stochastic optimal stopping problem with dynamic programming for the optimal stopping times. Having introduced the  quantum computational model, we are now in the position to construct quantum circuits for various computations related to the dynamic programming. 
In particular, we construct a unitary that propagates backwards the optimal stopping time by one time step according to Eq.~(\ref{eq:newdynamic}). In what follows, given a path $x\in E^T$, by $z_{\widetilde{\tau}_t(x)}$ we mean $z_{\widetilde{\tau}_t(x)}(x_{\widetilde{\tau}_t(x)})$, i.e., the associated payoff of the $\widetilde{\tau}_t(x)$-th time step of $x$. While we assume a fixed-point encoding model, we will drop the notation $\mathcal{Q}(\cdot)$ from now on for simplicity.

\begin{lemma}[Quantum circuits for computing the stopping times]\label{lemQStop}
Let $\widetilde{\tau}_{t'}$, $t'\in\{0,\dots,T\}$, be stopping times defined in Eq.~\eqref{eq:newdynamic}. For all $t\in[T]$, given quantum query access to functions $\{z_{t'}:E\to\mathbb{R}\}_{t'=1}^T$ and $\{e_{t',k}:E\to\mathbb{R}\}_{t'\in[T-1],k\in[m]}$ in time $\mathcal{T}_{z}$ and $\mathcal{T}_{e}$, respectively, and to the components of real vectors $\{\widetilde{\alpha}_{t'}\in\mathbb{R}^m\}_{t'=t}^T$ in time $O(1)$, the following statements are true (let $e_{0,k}(x_0) := 1$):
\begin{enumerate}
    \item There is a unitary $W_{t}$ such that, in time $O(\mathcal{T}_{z}+m\mathcal{T}_{e})$,
    \begin{align*}
        \begin{cases}
            W_{t}\ket{x}\ket{\widetilde{\tau}_{t+1}(x)}|\mathbf{0}\rangle^{\otimes 3} = \ket{x}\ket{\widetilde{\tau}_{t+1}(x)}\ket{z_{t}(x_t)}|\widetilde{\alpha}_{t}\cdot \vec{e}_t(x_{t})\rangle|\widetilde{\tau}_{t}(x)\rangle &\quad~\text{if}~t\neq T,\\
            W_{t}\ket{x}|\mathbf{0}\rangle = |x\rangle|T\rangle &\quad~\text{if}~t=T.
        \end{cases} 
    \end{align*}
    \item There is a unitary $V^{(k)}_t$ such that $V^{(k)}_t|x\rangle|\widetilde{\tau}_{t}(x)\rangle|\mathbf{0}\rangle = |x\rangle|\widetilde{\tau}_{t}(x)\rangle|z_{\widetilde{\tau}_{t}(x)}e_{t-1,k}(x_{t-1})\rangle$, for $k\in[m]$, in time $O(T \log(T) \mathcal{T}_{z}+\mathcal{T}_{e})$.
    \item The unitary $C_t^{(k)} := W_T^\dagger \dots W_{t+1}^\dagger W_t^\dagger V_t^{(k)} W_t W_{t+1}\dots W_T$ is such that $C_t^{(k)}|x\rangle|\mathbf{0}\rangle^{\otimes(3(T-t)+2)} = |x\rangle|z_{\widetilde{\tau}_{t}(x)}e_{t-1,k}(x_{t-1})\rangle|\mathbf{0}\rangle^{\otimes(3(T-t)+1)}$, for $k\in[m]$, in time $O(T(\log(T)\mathcal{T}_{z}+m\mathcal{T}_{e}))$.
\end{enumerate}
\end{lemma}
\begin{proof}
    We start with the first statement. The existence of $W_T$ is trivial. Assume $t\in[T-1]$, then, with one query access to function oracle $z_t$ and $m$ query accesses to function oracle $e_{t,k}$ for $k\in[m]$, we can perform
    $$
        \ket{x} \ket{\bf 0}\ket{\bf 0}\mapsto \ket{x}\ket{z_t(x_t)} \ket{\bf 0}\mapsto \ket{x}\ket{z_t(x_t)}\ket{\vec{e}_t(x_t)}.
    $$
    By using access to the $m$ elements of $\widetilde{\alpha}_t$, and $O(m)$ multiplications and additions, we can compute the inner product of $\widetilde{\alpha}_t \cdot \vec{e}_t(x_t)$ in superposition over $x_t$, as
    $$
        \ket{x}\ket{\vec{e}_t(x_t)}\ket{\bf 0} \mapsto \ket{x}\ket{\vec{e}_t(x_t)} \ket{\widetilde{\alpha}_t \cdot \vec{e}_t(x_t)}.
    $$
    Comparing between $z_t(x_t)$ and $\widetilde{\alpha}_t \cdot \vec{e}_t(x_t)$ in constant time, we can compute $\widetilde{\tau}_t(x)$ according to Eq.~\eqref{eq:newdynamic}, and hence obtain
    $$
        \ket{x}\ket{\widetilde{\tau}_{t+1}(x)}\ket{z_t(x_t)}\ket{\widetilde{\alpha}_t \cdot \vec{e}_t(x_t)} \ket{\bf 0}
        \mapsto
        \ket{x}\ket{\widetilde{\tau}_{t+1}(x)}\ket{z_t(x_t)}\ket{\widetilde{\alpha}_t \cdot \vec{e}_t(x_t)}\ket{\widetilde{\tau}_t(x)}.
    $$
    Uncomputing the intermediate steps leads to the desired operation. The total runtime is $O(\mathcal{T}_{z}+m\mathcal{T}_{e}+m+m+1) = O(\mathcal{T}_{z}+m\mathcal{T}_{e})$.
    
Regarding the second statement, 
we require a few circuits which can be constructed once as a pre-processing step.
First, we prepare the input for the payoff functions in an ancillary register, where the input depends on the content of the register $\ket{\widetilde{\tau}_t(x)}$.
For this step, we  prepare a conditional copy quantum circuit $V_{\rm copy}$ which operates as 
$\ket{x}\ket{\widetilde{\tau}_t(x)}\ket{\bf 0} \to \ket{x}\ket{\widetilde{\tau}_t(x)}|x_{\widetilde{\tau}_t(x)}\rangle$, 
where the register $|{x_{\widetilde{\tau}_t(x)}}\rangle$ stores the $\widetilde{\tau}_t(x)$-th step of the path $x$. This circuit operates in time given by the size of the registers of at most $O(T\log (T))$.
Second, from the access to the different payoff functions we construct access to the functions in superposition of the time parameter.
By assumption, we are given quantum circuits $V_{z_{t'}}$ for $t' \in [T]$. From these quantum circuits we construct the controlled circuit $V_{\rm select}:=\sum_{t'=1}^T \ket{t'}\bra{t'} \otimes V_{z_{t'}}$, which consists of the controlled versions of the circuits $V_{z_{t'}}$ and has a runtime of $O(T \log (T) \mathcal T_z)$~\cite{Berry2015trunc}.
Now, given $\ket{x}\ket{\widetilde{\tau}_t(x)}$, with one application of $V_{\rm copy}$ and one application of $V_{\rm select}$, we obtain $z_{\widetilde{\tau}_t(x)}$, i.e, the payoff evaluated at $x_{\widetilde{\tau}_t(x)}$, as
\begin{eqnarray*}
\ket{x}\ket{\widetilde{\tau}_t(x)}\ket{\bf 0}\ket{\bf 0}
&\mapsto&
\ket{x}\ket{\widetilde{\tau}_t(x)}|{x_{\widetilde{\tau}_t(x)}}\rangle\ket{\bf 0}
\mapsto
\ket{x}\ket{\widetilde{\tau}_t(x)}|{x_{\widetilde{\tau}_t(x)}}\rangle|{z_{\widetilde{\tau}_t(x)}}\rangle.
\end{eqnarray*}
Using $V_{\rm copy}$ again we uncompute the third register. 
One query to the function oracle $e_{t-1,k}$
obtains $e_{t-1,k}(x_{t-1})$ and by multiplication we obtain,
\begin{align*}
\ket{x}\ket{\widetilde{\tau}_t(x)}|{z_{\widetilde{\tau}_t(x)}}\rangle\ket{e_{t-1,k}(x_{t-1})} \ket{\bf 0}
&\mapsto
\ket{x}\ket{\widetilde{\tau}_t(x)}|{z_{\widetilde{\tau}_t(x)}}\rangle\ket{e_{t-1,k}(x_{t-1})}|{z_{\widetilde{\tau}_t(x)}e_{t-1,k}(x_{t-1})}\rangle.
\end{align*}
We uncompute the third and fourth registers using the given circuits to obtain the desired operation.
The total runtime is $O(T \log(T) \mathcal{T}_{z}+\mathcal{T}_{e})$.

Finally, for the third statement, it is not hard to see that (let $\widetilde{\tau}_T := T$)
\begin{align*}
C_t^{(k)}|x\rangle|\mathbf{0}\rangle^{\otimes(3(T-t)+2)} &= W_T^\dagger \dots W_t^\dagger V_t^{(k)}|x\rangle|T\rangle|\mathbf{0}\rangle\bigotimes_{j=t}^{T-1} \ket{z_j(x_j)}\ket{\widetilde{\alpha}_j \cdot \vec{e}_j(x_j)}\ket{\widetilde{\tau}_j(x)}\\
&= W_T^\dagger \dots W_t^\dagger |x\rangle|T\rangle|{z_{\widetilde{\tau}_t(x)}e_{t-1,k}(x_{t-1})}\rangle\bigotimes_{j=t}^{T-1} \ket{z_j(x_j)}\ket{\widetilde{\alpha}_j \cdot \vec{e}_j(x_j)}\ket{\widetilde{\tau}_j(x)}\\
&= |x\rangle|{z_{\widetilde{\tau}_t(x)}e_{t-1,k}(x_{t-1})}\rangle|\mathbf{0}\rangle^{\otimes(3(T-t)+1)}.
\end{align*}
From the two previous statements, the runtime of $C_t^{(k)}$ is $2(T-t+1)O(\mathcal{T}_{z}+m\mathcal{T}_{e})+O(T\log(T)\mathcal{T}_{z}+\mathcal{T}_{e})=O(T(\log(T)\mathcal{T}_{z}+m\mathcal{T}_{e}))$.
\end{proof}

\subsection{Quantum least squares Monte Carlo algorithm}

We present our quantum {\rm LSM} algorithm in Algorithm~\ref{alg:quantumlongstaff}. As previously mentioned, it follows the classical version in Algorithm~\ref{alg:classicalalgo}. It computes the expectations $\mathbb{E}[e_{t,k}(X_t)e_{t,l}(X_t)]$ (for the matrices $\{A_t\}_{t=1}^{T-1}$), $\mathbb{E}[Z_{\widetilde{\tau}_{t+1}}\vec{e}_t(X_t)]$ and $\mathbb{E}[Z_{\widetilde{\tau}_1}]$ using Theorem~\ref{thr:quantum_monte_carlo} instead of drawing random samples.
Recall that by definition $Z_{\widetilde{\tau}_{t+1}} = z_{\widetilde{\tau}_{t+1}(\mathbf{X}_{t+1})}(X_{\widetilde{\tau}_{t+1}(\mathbf{X}_{t+1})})$, i.e., both the optimal stopping time and the payoff depend on the path of the Markov chain. 

\begin{algorithm}[t]
\caption{Quantum {\rm LSM} algorithm for optimal stopping problem}
\label{alg:quantumlongstaff}
\begin{algorithmic}[1]
	\Require Parameters $\delta\in(0,1)$, $\epsilon>0$. Quantum sampling access to Markov chain $(X_t)_{t=1}^T$ defined on a finite sample space $\Omega$ and with finite state space $E\subseteq\mathbb{R}^d$. Quantum query access to $\{z_t:E\to\R\}_{t=1}^T$ and linearly independent functions $\{e_{t,k}:E \to \R\}_{k=1}^m$ for $t\in[T-1]$. Let $L := \max_{t\in[T-1],k\in[m]} \|e_{t,k}\|_{L^2(\rho_t)}$ and $R := \max_{t\in[T]}\|z_t\|_u$.
	\State $\delta_A := \delta/(4Tm^2)$, $\delta_b := \delta/(4Tm)$, $\epsilon_A := \epsilon/m$ and $\epsilon_b := \epsilon/\sqrt{m}$.
    \State Construct quantum access and controlled rotation access to $e_{t,k} e_{t,l}$, $\forall k,l \in[m],t\in[T-1]$, with quantum query access to $e_{t,j}$, quantum circuits for multiplication and Fact~\ref{fact:query_to_rotation}.
    \State Compute $\{\widetilde{A}_t\}_{t=1}^{T-1}$ by calling $\mathtt{QMonteCarlo}(e_{t,k}(X_t)e_{t,l}(X_t)$, $\epsilon_A$, $\delta_A$, $L^2$) for $k,l\in[m]$.
    \State Compute the inverses $\{\widetilde{A}_t^{-1}\}_{t=1}^{T-1}$.
    \State Prepare unitary $W_T$ s.t.\ $W_T|x\rangle|\mathbf{0}\rangle = |x\rangle|\widetilde{\tau}_T(x)\rangle$, where $\widetilde{\tau}_T(x) = T$ for all $x\in E^T$.
    \For{$t=T$ to $2$}
    \If{$t\neq T$}
        \State \parbox[t]{\dimexpr\linewidth-\algorithmicindent-\algorithmicindent}{
            Prepare unitary $W_t$ s.t.\ $W_{t}\ket{x}\ket{\widetilde{\tau}_{t+1}(x)}|\mathbf{0}\rangle^{\otimes 3} = \ket{x}\ket{\widetilde{\tau}_{t+1}(x)}\ket{z_{t}(x_t)}|\widetilde{\alpha}_{t}\cdot \vec{e}_t(x_{t})\rangle|\widetilde{\tau}_{t}(x)\rangle$ for any $\widetilde{\tau}_{t+1}(x)\in[T]$ (Lemma \ref{lemQStop}).}
        \EndIf
        
        \State \parbox[t]{\dimexpr\linewidth-\algorithmicindent}{
        Prepare unitaries $\{V^{(k)}_t\}_{k=1}^m$ s.t.\ $V^{(k)}_t|x\rangle|\widetilde{\tau}_{t}(x)\rangle|\mathbf{0}\rangle = |x\rangle|\widetilde{\tau}_{t}(x)\rangle|z_{\widetilde{\tau}_{t}(x)}e_{t-1,k}(x_{t-1})\rangle$~(Lemma~\ref{lemQStop}).}
        \State Prepare unitary $W_T^\dagger \dots W_{t+1}^\dagger W_t^\dagger V_t^{(k)} W_t W_{t+1}\dots W_T$ for $k\in[m]$ (Lemma~\ref{lemQStop}).
        \State \parbox[t]{\dimexpr\linewidth-\algorithmicindent}{Construct quantum access to the controlled rotations of the functions $z_{\widetilde{\tau}_{t}(x)}e_{t-1,k}(x_{t-1})$ (Fact~\ref{fact:query_to_rotation}).}
        \State Execute $\mathtt{QMonteCarlo}(Z_{\widetilde{\tau}_{t}} e_{t-1,k}(X_{t-1})$, $\epsilon_b$, $\delta_b$, $RL$), for all $k \in [m]$, to compute $\widetilde{b}_{t-1}$. 
        \State Compute the vector $\widetilde{\alpha}_{t-1} = \widetilde{A}_{t-1}^{-1}\widetilde{b}_{t-1}$ classically.
    \EndFor
    \State Prepare unitary $W_{1}$ s.t.\ $W_{1}\ket{x}\ket{\widetilde{\tau}_{2}(x)}|\mathbf{0}\rangle^{\otimes 3} = \ket{x}\ket{\widetilde{\tau}_{2}(x)}\ket{z_{1}(x)}|\widetilde{\alpha}_{1}\cdot \vec{e}_1(x_{1})\rangle|\widetilde{\tau}_{1}(x)\rangle$ for any $\widetilde{\tau}_{2}(x)\in[T]$ (Lemma~\ref{lemQStop}).
    \State Prepare unitary $V_1$ s.t.\ $V_1|x\rangle|\widetilde{\tau}_{1}(x)\rangle|\mathbf{0}\rangle = |x\rangle|\widetilde{\tau}_{1}(x)\rangle|z_{\widetilde{\tau}_{1}(x)}\rangle$ (Lemma~\ref{lemQStop}). 
    \State Prepare unitary $W_T^\dagger \dots W_{2}^\dagger W_1^\dagger V_1 W_1 W_{2}\dots W_T$ (Lemma~\ref{lemQStop}).
    \State Construct quantum access to the controlled rotations of the function $z_{\widetilde{\tau}_{1}}$ (Fact~\ref{fact:query_to_rotation}).
    \State Execute $\mathtt{QMonteCarlo}(Z_{\widetilde{\tau}_{1}}$, $\epsilon$, $\frac{\delta}{2}$, $R$) to compute $\widetilde{Z}_{\widetilde{\tau}_1}$. 
    \State Output $\widetilde{\mathcal{U}}_0 := \max\big\{Z_0,\widetilde{Z}_{\widetilde{\tau}_1}\big\}$.
\end{algorithmic}
\end{algorithm} 

We explain the three essential steps of the algorithm.
First, we compute $\mathbb{E}[e_{t,k}(X_t)e_{t,l}(X_t)]$ to obtain the matrices $\{A_t\}_{t=1}^{T-1}$. For this we provide quantum access to the functions $e_{t,k}(x_t)e_{t,l}(x_t)$
from the quantum access to all $e_{t,k}(x_t)$ and the arithmetic quantum circuit for multiplication. We take the number of bits such that the numerical error is much smaller than the error obtained from estimating the random variable. By our arithmetic model, the computation is performed in constant time. By availability of the quantum controlled rotations with Fact~\ref{fact:query_to_rotation}, we then satisfy the hypothesis of Theorem~\ref{thr:quantum_monte_carlo}.
Next, for any given time step $t\in\{2,\dots,T\}$, we solve the dynamic recursion $\widetilde{\tau}_{t}(x) = t \mathbf{1}\{z_{t}(x_t) \geq \widetilde{\alpha}_{t}\cdot \vec{e}_t(x_{t})\} +  \widetilde{\tau}_{t+1}(x)\mathbf{1}\{z_{t}(x_t) < \widetilde{\alpha}_{t}\cdot \vec{e}_t(x_{t})\}$ from time $T$ to time $t$ in superposition via unitaries $C_t^{(k)}$ built by a separate subroutine explained in Lemma~\ref{lemQStop}. These unitaries provide $(C_t^{(k)}, \mathcal T_{ze})$-quantum access to the functions $z_{\widetilde{\tau}_{t}(x)}e_{t-1,k}(x_{t-1})$, 
where the cost $\mathcal T_{ze}$ is stated in Lemma~\ref{lemQStop}. With the corresponding quantum controlled rotations, this function access can be used with Theorem~\ref{thr:quantum_monte_carlo} to estimate the corresponding vector of expected values $b_{t-1} := \mathbb{E}[Z_{\widetilde{\tau}_{t}}\vec{e}_{t-1}(X_{t-1})]$ that is then used to classically compute the approximation $\widetilde{\alpha}_{t-1}$ for the next time step $t-1$. 
Finally, for $t=1$, with similar steps, we provide 
$(C_1, \mathcal{T}_{z})$-quantum access to the function $z_{\widetilde{\tau}_{1}(x)}$,
where the cost $\mathcal T_{z}$ follows again from Lemma~\ref{lemQStop}.
This last estimation then provides the approximation $\widetilde{\mathcal U}_0$ to the true initial value of the Snell envelope. We note that the procedure of approximating a matrix $A_t$ and a vector $b_t$ entrywise via quantum algorithms for Monte Carlo followed by the classical computation of $\widetilde{\alpha}_t = \widetilde{A}_t^{-1}\widetilde{b}_t$ was already used in~\cite{PhysRevA.104.022430}.

In order to analyze the approximation error and the complexity of Algorithm~\ref{alg:quantumlongstaff}, we will need the following result from~\cite{zanger2009convergence,zanger2020general} (already modified to our notation) that bounds the error between the exact continuation values $\mathbb{E}[Z_{\tau_{t+1}}|X_t]$ and their approximation $\widetilde{\alpha}_t\cdot \vec{e}_t(X_t)$ in terms of the error between the continuation values $\mathbb{E}[Z_{\widetilde{\tau}_{k+1}}|X_k]$ evaluated on the approximated stopping times $\widetilde{\tau}_{k+1}$ and $\widetilde{\alpha}_k\cdot \vec{e}_t(X_k)$ for $k\in\{t,\dots,T-1\}$. Recall the image probability measures $\rho_t$ in $E\subseteq\mathbb{R}^d$ induced by each element $X_t$, $t\in\{0,\dots,T\}$.
\begin{lemma}[{\cite[Lemma~2.2]{zanger2009convergence}}]
    \label{lem:zanger}
    For each $t\in\{0,\dots,T-1\}$, we have
    \begin{align*}
        \|\widetilde{\alpha}_t\cdot \vec{e}_t(X_t) - \mathbb{E}[Z_{\tau_{t+1}}|X_t]\|_{L^2(\rho_t)} &\leq 2\sum_{k=t}^{T-1}\|\widetilde{\alpha}_k\cdot \vec{e}_k(X_k) - \mathbb{E}[Z_{\widetilde{\tau}_{k+1}}|X_k]\|_{L^2(\rho_k)},\\
        \|\mathbb{E}[Z_{\widetilde{\tau}_{t+1}}|X_t] - \mathbb{E}[Z_{\tau_{t+1}}|X_t]\|_{L^2(\rho_t)} &\leq 2\sum_{k=t+1}^{T-1}\|\widetilde{\alpha}_k\cdot \vec{e}_k(X_k) - \mathbb{E}[Z_{\widetilde{\tau}_{k+1}}|X_k]\|_{L^2(\rho_k)},
    \end{align*}
    where $\widetilde{\alpha}_0\cdot \vec{e}_0(X_0) := \widetilde{Z}_{\widetilde{\tau}_1}$ approximates $\mathbb{E}[Z_{\widetilde{\tau}_1}]$.
\end{lemma}
We will also need the following technical result on the sensitivity of square systems.
\begin{theorem}[{\cite[Theorem~2.6.2]{golub2013matrix}}]
    \label{thr:matrix_error}
    Let $Ax = b$ and $\widetilde{A}\widetilde{x} = \widetilde{b}$, where $A,\widetilde{A}\in\mathbb{R}^{d\times d}$ and $b,\widetilde{b}\in\mathbb{R}^d$, with $b\neq 0$. Suppose that $\|A-\widetilde{A}\|_2 \leq \epsilon_A$ and $\|b - \widetilde{b}\|_2 \leq \epsilon_b$. If $\epsilon_A \leq \sigma_{\min}(A)/2$, where $\sigma_{\min}(A)$ is the minimum singular value of $A$, then
    \begin{align*}
        \|x - \widetilde{x}\|_2 \leq \frac{2}{\sigma_{\min}(A)}\left(\frac{\epsilon_A\|b\|_2}{\sigma_{\min}(A)} + \epsilon_b\right).
    \end{align*}
\end{theorem}
We are now able to state a central theorem for our quantum {\rm LSM} algorithm. Recall the definition $\|f\|_u = \sup\{|f(x)|:x\in E\}$ of the uniform norm.
\begin{theorem}
    \label{thr:quantum_result1}
    Within the setting of Algorithm~{\rm \ref{alg:quantumlongstaff}} with input parameters $\delta$ and $\epsilon$, let $T\mathcal{T}_{\rm samp}$ be the sampling cost of the Markov chain and consider a set of linearly independent functions $\{e_{t,k}:E\to\mathbb{R}\}_{k=1}^m$ for each $t\in[T-1]$ with $L := \max_{t\in[T-1],k\in[m]} \|e_{t,k}\|_{L^2(\rho_t)}$ and query cost $\mathcal{T}_{e}$. Also consider $\{z_t:E\to\mathbb{R}\}_{t=0}^T$ with $R := \max_{t\in[T]}\|z_t\|_u < \infty$ and query cost $\mathcal{T}_{z}$. Moreover, let $\sigma_{\min} := \min_{t\in[T-1]} \sigma_{\min}(A_t) > 0$. Assume that $\sqrt{m}RL/\sigma_{\min} \geq 1$ and define $\mathcal{T}_{\rm total} := \mathcal{T}_{\rm samp} + \mathcal{T}_{z} + \mathcal{T}_{e}$. Then, for any $\delta\in(0,1)$ and $\epsilon\in(0,\sigma_{\min}/2]$, Algorithm~{\rm \ref{alg:quantumlongstaff}} outputs $\widetilde{\mathcal{U}}_0$ such that
    \begin{align}
        \label{eq:final_bound}
        \operatorname{Pr}\left[|\widetilde{\mathcal{U}}_0 - \mathbb{E}[Z_{\tau_0}]| \geq \frac{8T\epsilon mRL^2}{\sigma_{\min}^2} + 2\sum_{t=1}^{T-1} \operatorname*{\min}_{a\in\mathbb{R}^m}\|a\cdot \vec{e}_t(X_t) - \mathbb{E}[Z_{\widetilde{\tau}_{t+1}}|X_t]\|_{L^2(\rho_t)} \right] \leq \delta
    \end{align}
    in time
    \begin{align*}
        O\left(\frac{T^2m^3}{\epsilon}\mathcal{T}_{\rm total}L(L+R)\log(T)\log(Tm/\delta) \log^{3/2}(mRL/\epsilon)\log\log(mRL/\epsilon) \right).
    \end{align*}
\end{theorem}
\begin{proof}
    Set $b_t := \mathbb{E}[Z_{\widetilde{\tau}_{t+1}}\vec{e}_t(X_t)]$. Recall that $\epsilon_A := \epsilon/m$, $\epsilon_b := \epsilon/\sqrt{m}$, $\delta_A := \delta/(4Tm^2)$ and $\delta_b := \delta/(4Tm)$. We start by computing the complexity of the algorithm. We first note the bounds $\|A_t\|_2 \leq m\max_{k,l}|\mathbb{E}[e_{t,k}(X_t)e_{t,l}(X_t)]| \leq mL^2$ and $\|b_t\|_2 \leq R\|\mathbb{E}[\vec{e}_t(X_t)]\|_2 \leq \sqrt{m}RL$ for all $t\in[T-1]$. The computation of all the entries of the matrices $\{A_t\}_{t=1}^{T-1}$ requires time
    \begin{align*}
        O\left(\frac{Tm^2}{\epsilon_A}L^2(T\mathcal{T}_{\rm samp} + \mathcal{T}_{e})\log(1/\delta_A) \log^{3/2}\left(\frac{L^2}{\epsilon_A} \right)\log\log\left(\frac{L^2}{\epsilon_A}\right) \right),
    \end{align*}
    by calling $\mathtt{QMonteCarlo}(e_{t,k}(X_t)e_{t,l}(X_t), \epsilon_A, \delta_A, L^2)$ from Theorem~\ref{thr:quantum_monte_carlo} for $t\in[T-1]$ and $k,l\in[m]$. The computation of all $b_t$ uses time
    \begin{align*}
        O\left(\frac{Tm}{\epsilon_b}RL(T\mathcal{T}_{\rm samp} + T(\log(T)\mathcal{T}_{z} + m\mathcal{T}_{e}))\log(1/\delta_b) \log^{3/2}\left(\frac{RL}{\epsilon_b} \right)\log\log\left(\frac{RL}{\epsilon_b}\right) \right),
    \end{align*}
    by calling $\mathtt{QMonteCarlo}(Z_{\widetilde{\tau}_{t}} e_{t-1,k}(X_{t-1}), \epsilon_b, \delta_b, RL)$ from Theorem~\ref{thr:quantum_monte_carlo} for $t\in\{2,\dots,T\}$ and $k\in[m]$. Note that the term $T(\log(T)\mathcal{T}_{z} + m\mathcal{T}_{e})$ comes from using the unitaries $C_t^{(k)}$ in $\mathtt{QMonteCarlo}$, each with cost $O(T(\log(T)\mathcal{T}_{z} + m\mathcal{T}_{e}))$ according to Lemma~\ref{lemQStop}. Computing $\mathbb{E}[Z_{\widetilde{\tau}_1}]$ requires time
    \begin{align*}
        O\left(\frac{R}{\epsilon}(T\mathcal{T}_{\rm samp} + T(\log(T)\mathcal{T}_{z} + m\mathcal{T}_{e}))\log(1/\delta) \log^{3/2}\left(\frac{R}{\epsilon} \right)\log\log\left(\frac{R}{\epsilon}\right) \right),
    \end{align*}
    by calling $\mathtt{QMonteCarlo}(Z_{\widetilde{\tau}_{1}}, \epsilon, \frac{\delta}{2}, R)$ from Theorem~\ref{thr:quantum_monte_carlo} and where the term $T(\log(T)\mathcal{T}_{z} + m\mathcal{T}_{e})$ again comes from the unitaries $C_t^{(k)}$ in $\mathtt{QMonteCarlo}$. The classical computation of $\{\widetilde{A}^{-1}_t\}_{t=1}^{T-1}$ and $\{\widetilde{\alpha}_t = \widetilde{A}^{-1}_t\widetilde{b}_t\}_{t=1}^{T-1}$ requires time $O(Tm^{\omega_\ast})$, where $2\leq\omega_\ast<3$. Hence, by keeping the largest terms of each complexity, the final complexity is upper-bounded by
    \begin{align*}
        O\left(\frac{T^2m^3}{\epsilon}\mathcal{T}_{\rm total}L(L+R)\log(T)\log(Tm/\delta) \log^{3/2}(mRL/\epsilon)\log\log(mRL/\epsilon) \right).
    \end{align*}
    
    We now move to the error analysis. Fix $t\in[T-1]$. We start by bounding the error $\|\widetilde{\alpha}_t - \alpha_t\|_2$ between $\alpha_t = A_t^{-1}b_t$ and $\widetilde{\alpha}_t = \widetilde{A}_t^{-1}\widetilde{b}_t$. By using $\mathtt{QMonteCarlo}$ from Theorem~\ref{thr:quantum_monte_carlo} we approximate each entry of $A_t$ and $b_t$ as $|(A_t)_{jl} - (\widetilde{A}_t)_{jl}| \leq \epsilon/m$ and $|(b_t)_j - (\widetilde{b}_t)_j| \leq \epsilon/\sqrt{m}$ for all $j,l\in[m]$. All approximations hold with probability at least $1-\delta/2T$ by the union bound. This means that, with probability at least $1-\delta/2T$, 
    \begin{align*}
        \|A_t - \widetilde{A}_t\|_2 \leq \sqrt{\sum_{j,l=1}^m |(A_t)_{jl} - (\widetilde{A}_t)_{jl}|^2} \leq \epsilon, \qquad
        \|b_t - \widetilde{b}_t\|_2 = \sqrt{\sum_{j=1}^m |(b_t)_j - (\widetilde{b}_t)_j|^2} \leq  \epsilon.
    \end{align*}
    According to Theorem~\ref{thr:matrix_error}, we obtain
    \begin{align*}
        \|\widetilde{\alpha}_t - \alpha_t\|_2 \leq \frac{2\epsilon}{\sigma_{\min}(A_t)}\left(1 + \frac{\|b_t\|_2}{\sigma_{\min}(A_t)} \right) \leq \frac{4\epsilon\sqrt{m}RL}{\sigma_{\min}^2}
    \end{align*} 
    with probability at least $1-\delta/2T$, using that $\sqrt{m}RL/\sigma_{\min} \geq 1$. This, in turn, implies that
    \begin{align*}
        \|\widetilde{\alpha}_t\cdot \vec{e}_t(X_t) - \alpha_t\cdot \vec{e}_t(X_t)\|_{L^2(\rho_t)} \leq \|\widetilde{\alpha}_t - \alpha_t\|_2 \big\|\sqrt{\vec{e}_t(X_t)\cdot\vec{e}_t(X_t)}\big\|_{L^2(\rho_t)} \leq \frac{4\epsilon mRL^2}{\sigma_{\min}^2},
    \end{align*}
    using that $\big\|\sqrt{\vec{e}_t(X_t)\cdot\vec{e}_t(X_t)}\big\|_{L^2(\rho_t)} = \sqrt{\sum_{x\in E}\rho_t(x)\sum_{k=1}^m e_{t,k}^2(x)} \leq \sqrt{m}L$.
    Next, we bound 
    \begin{align}
        \|\widetilde{\alpha}_t\cdot \vec{e}_t(X_t) - \mathbb{E}[Z_{\widetilde{\tau}_{t+1}}|X_t]\|_{L^2(\rho_t)} &\leq \|\widetilde{\alpha}_t\cdot \vec{e}_t(X_t) - \alpha_t\cdot \vec{e}_t(X_t)\|_{L^2(\rho_t)} + \|\alpha_t\cdot \vec{e}_t(X_t) - \mathbb{E}[Z_{\widetilde{\tau}_{t+1}}|X_t]\|_{L^2(\rho_t)}\nonumber\\
        &\leq \frac{4\epsilon mRL^2}{\sigma_{\min}^2} + \operatorname*{\min}_{a\in\mathbb{R}^m}\|a\cdot \vec{e}_t(X_t) - \mathbb{E}[Z_{\widetilde{\tau}_{t+1}}|X_t]\|_{L^2(\rho_t)}, \label{eq:bound1}
    \end{align}
    using that $\alpha_t = \arg\operatorname*{\min}_{a\in\mathbb{R}^m}\mathbb{E}\big[(Z_{\widetilde{\tau}_{t+1}} - a\cdot \vec{e}_t(X_t))^2\big]$ minimizes the least square estimator.
    
    Finally, by the union bound, with probability at least $1-\delta$, Eq.~\eqref{eq:bound1} holds for all $t\in[T-1]$, together with $|\widetilde{Z}_{\widetilde{\tau}_1} - \mathbb{E}[Z_{\widetilde{\tau}_1}]| \leq \epsilon$. Lemma~\ref{lem:zanger} then leads to
    \begin{align*}
        |\widetilde{Z}_{\widetilde{\tau}_1} - \mathbb{E}[Z_{\tau_{1}}]| &\leq \frac{8T\epsilon mRL^2}{\sigma_{\min}^2} + 2\sum_{t=1}^{T-1} \operatorname*{\min}_{a\in\mathbb{R}^m}\|a\cdot \vec{e}_t(X_t) - \mathbb{E}[Z_{\widetilde{\tau}_{t+1}}|X_t]\|_{L^2(\rho_t)},
    \end{align*}
    which implies Eq.~\eqref{eq:final_bound} by using $|{\max}\{a_0,a_1\}-\max\{a_0,a_2\}|\leq |a_1-a_2|$ with $a_0,a_1,a_2\in\mathbb{R}$ on the definition of $\widetilde{\mathcal{U}}_0$ in Eq.~\eqref{eq:approximate_result}.
\end{proof}
\begin{remark}
    As previously mentioned, it is possible to use the multivariable quantum algorithm for Monte Carlo from~{\rm \cite{cornelissen2021quantum,cornelissen2021bquantum}} within Algorithm~{\rm \ref{alg:quantumlongstaff}}. This leads to the same complexity of the previous theorem, $\widetilde{O}\big(\frac{T^2m^3}{\epsilon}\mathcal{T}_{\rm total}L(L+R)\big)$, up to $\operatorname{polylog}$ factors.
\end{remark}

As previously mentioned, the quantities $\operatorname*{\min}_{a\in\mathbb{R}^m}\|a\cdot \vec{e}_t(X_t) - \mathbb{E}[Z_{\widetilde{\tau}_{t+1}}|X_t]\|_{L^2(\rho_t)}$ appearing in the theorem above are sometimes known as \emph{approximation errors} of the approximation architecture $\mathscr{H}_t$, here chosen to be finite-dimensional, linear and generated by a set of linearly independent functions. We could rewrite these terms alternatively as $\inf_{f\in\mathscr{H}_t} \|f(X_t) - \mathbb{E}[Z_{\widetilde{\tau}_{t+1}}|X_t]\|_{L^2(\rho_t)}$.

Note that the approximation errors $\operatorname*{\min}_{a\in\mathbb{R}^m}\|a\cdot \vec{e}_t(X_t) - \mathbb{E}[Z_{\widetilde{\tau}_{t+1}}|X_t]\|_{L^2(\rho_t)}$ appearing in Theorem~\ref{thr:quantum_result1} depend on the approximated stopping times $\widetilde{\tau}_{t+1}$, which in turn depend on $\mathscr{H}_{t'}$ for $t'>t$. It is possible to restate Theorem~\ref{thr:quantum_result1} in terms of $\operatorname*{\min}_{a\in\mathbb{R}^m}\|a\cdot \vec{e}_t(X_t) - \mathbb{E}[Z_{\tau_{t+1}}|X_t]\|_{L^2(\rho_t)}$, which we do in the next theorem by following a similar approach to~\cite[Theorem~6.1]{zanger2009convergence}. The downside is that the time dependence now becomes exponential.
\begin{theorem}
    \label{thr:quantum_result2}
    Within the setting of Algorithm~{\rm \ref{alg:quantumlongstaff}} with input parameters $\delta$ and $\epsilon$, let $T\mathcal{T}_{\rm samp}$ be the sampling cost of the Markov chain and consider a set of linearly independent functions $\{e_{t,k}:E\to\mathbb{R}\}_{k=1}^m$ for each $t\in[T-1]$ with $L := \max_{t\in[T-1],k\in[m]} \|e_{t,k}\|_{L^2(\rho_t)}$ and query cost $\mathcal{T}_{e}$. Also consider $\{z_t:E\to\mathbb{R}\}_{t=0}^T$ with $R := \max_{t\in[T]}\|z_t\|_u < \infty$ and query cost $\mathcal{T}_{z}$. Moreover, let $\sigma_{\min} := \min_{t\in[T-1]} \sigma_{\min}(A_t) > 0$. Assume that $\sqrt{m}RL/\sigma_{\min} \geq 1$ and define $\mathcal{T}_{\rm total} := \mathcal{T}_{\rm samp} + \mathcal{T}_{z} + \mathcal{T}_{e}$. Then, for any $\delta\in(0,1)$ and $\epsilon\in(0,\sigma_{\min}/2]$, Algorithm~{\rm \ref{alg:quantumlongstaff}} outputs $\widetilde{\mathcal{U}}_0$ such that
    \begin{align*}
        \operatorname{Pr}\left[\big|\widetilde{\mathcal{U}}_0 - \mathbb{E}[Z_{\tau_0}]\big| \geq 5^T\left(\frac{4\epsilon mRL^2}{\sigma_{\min}^2} + \operatorname*{\max}_{0 < t < T} \operatorname*{\min}_{a\in\mathbb{R}^m}\|a\cdot \vec{e}_t(X_t) - \mathbb{E}[Z_{\tau_{t+1}}|X_t]\|_{L^2(\rho_t)}\right) \right] \leq \delta
    \end{align*}
    in time
    \begin{align*}
        O\left(\frac{T^2m^3}{\epsilon}\mathcal{T}_{\rm total}L(L+R)\log(T)\log(Tm^2/\delta) \log^{3/2}(mRL/\epsilon)\log\log(mRL/\epsilon) \right).
    \end{align*}
\end{theorem}
\begin{proof}
    The proof follows the same steps of the proof of Theorem~\ref{thr:quantum_result1}, with the further observation in Eq.~\eqref{eq:bound1} that if
    \begin{align}
        \label{eq:assumption}
        \|\widetilde{\alpha}_{\ell}\cdot \vec{e}_{\ell}(X_{\ell}) - \mathbb{E}[Z_{\widetilde{\tau}_{\ell+1}}|X_{\ell}]\|_{L^2(\rho_{\ell})} \leq \epsilon_0 + \operatorname*{\min}_{a\in\mathbb{R}^m}\|a\cdot \vec{e}_{\ell}(X_{\ell}) - \mathbb{E}[Z_{\widetilde{\tau}_{\ell+1}}|X_{\ell}]\|_{L^2(\rho_{\ell})},
    \end{align}
    for all $\ell\in\{t,\dots,T-1\}$, where we defined $\epsilon_0 := \frac{4\epsilon mRL^2}{\sigma_{\min}^2}$, then
    \begin{align}
        \label{eq:induction_step}
        2(T-\ell)\epsilon_0 + 2\sum_{k=\ell}^{T-1}\operatorname*{\min}_{a\in\mathbb{R}^m}\|a\cdot \vec{e}_k(X_{k}) - \mathbb{E}[Z_{\widetilde{\tau}_{k+1}}|X_{k}]\|_{L^2(\rho_{k})} \leq 5^{T-\ell}(\epsilon_0 + M^{\ast}_{\ell}),
    \end{align}
    for all $\ell\in\{t,\dots,T-1\}$, where $M^{\ast}_{\ell} := \max_{k=\ell,\dots,T-1}\big({\operatorname*{\min}}_{a\in\mathbb{R}^m}\|a\cdot \vec{e}_k(X_{k}) - \mathbb{E}[Z_{\tau_{k+1}}|X_{k}]\|_{L^2(\rho_{k})}\big)$. We prove this bound using backward induction as follows. Eq.~\eqref{eq:induction_step} clearly holds for $\ell=T-1$. Assume it holds for $\ell=t+1$, we shall prove it is also true for $\ell=t$. First notice that, by the triangle inequality followed by Lemma~\ref{lem:zanger} and then Eq.~\eqref{eq:assumption},
    \begin{align*}
        &\operatorname*{\min}_{a\in\mathbb{R}^m}\|a\cdot \vec{e}_t(X_{t}) - \mathbb{E}[Z_{\widetilde{\tau}_{t+1}}|X_{t}]\|_{L^2(\rho_{t})}\\
        &\leq \operatorname*{\min}_{a\in\mathbb{R}^m}\|a\cdot \vec{e}_t(X_{t}) - \mathbb{E}[Z_{\tau_{t+1}}|X_{t}]\|_{L^2(\rho_{t})} + \|\mathbb{E}[Z_{\tau_{t+1}}|X_{t}] - \mathbb{E}[Z_{\widetilde{\tau}_{t+1}}|X_{t}]\|_{L^2(\rho_{t})}\\
        &\leq \operatorname*{\min}_{a\in\mathbb{R}^m}\|a\cdot \vec{e}_t(X_{t}) - \mathbb{E}[Z_{\tau_{t+1}}|X_{t}]\|_{L^2(\rho_{t})} + 2\sum_{k=t+1}^{T-1}\|\widetilde{\alpha}_k\cdot \vec{e}_k(X_{k}) - \mathbb{E}[Z_{\widetilde{\tau}_{k+1}}|X_{k}]\|_{L^2(\rho_{k})}\\
        &\leq \operatorname*{\min}_{a\in\mathbb{R}^m}\|a\cdot \vec{e}_t(X_{t}) - \mathbb{E}[Z_{\tau_{t+1}}|X_{t}]\|_{L^2(\rho_{t})} + 2\sum_{k=t+1}^{T-1}\big(\epsilon_0 + \operatorname*{\min}_{a\in\mathbb{R}^m}\|a\cdot \vec{e}_k(X_{k}) - \mathbb{E}[Z_{\widetilde{\tau}_{k+1}}|X_{k}]\|_{L^2(\rho_{k})}\big)\\
        &\leq M^{\ast}_{t} + 2(T-t-1)\epsilon_0 + 2\sum_{k=t+1}^{T-1} \operatorname*{\min}_{a\in\mathbb{R}^m}\|a\cdot \vec{e}_k(X_{k}) - \mathbb{E}[Z_{\widetilde{\tau}_{k+1}}|X_{k}]\|_{L^2(\rho_{k})}.
    \end{align*}
    Using the above inequality followed by the induction hypothesis,
    \begin{align*}
        2(T-t)&\epsilon_0 + 2\sum_{k=t}^{T-1}\operatorname*{\min}_{a\in\mathbb{R}^m}\|a\cdot \vec{e}_k(X_{k}) - \mathbb{E}[Z_{\widetilde{\tau}_{k+1}}|X_{k}]\|_{L^2(\rho_{k})}\\
        &\leq 2\epsilon_0 + 2M^{\ast}_{t} + 6(T-t-1)\epsilon_0 +  6\sum_{k=t+1}^{T-1}\operatorname*{\min}_{a\in\mathbb{R}^m}\|a\cdot \vec{e}_k(X_{k}) - \mathbb{E}[Z_{\widetilde{\tau}_{k+1}}|X_{k}]\|_{L^2(\rho_{k})}\\
        &\leq 2\epsilon_0 + 2M^{\ast}_{t} + 3\cdot 5^{T-t-1}(\epsilon_0 + M^{\ast}_{t+1})\\
        &\leq 5^{T-t}(\epsilon_0 + M^{\ast}_{t}),
    \end{align*}
    proving the induction statement. The theorem follows by taking $\ell=0$ in Eq.~\eqref{eq:induction_step} and using $|{\max}\{a_0,a_1\}-\max\{a_0,a_2\}|\leq |a_1-a_2|$, $a_0,a_1,a_2\in\mathbb{R}$, on the definition of $
    \widetilde{\mathcal{U}}_0 $ in Eq.~\eqref{eq:approximate_result}.
\end{proof}

Given the above theorems, we prove a classical analogue.
\begin{theorem}
    \label{thr:classical_error}
    Within the setting of Algorithm~{\rm \ref{alg:classicalalgo}}, consider $N$ independent sample paths with sample cost $\mathcal{T}_{\rm samp}$ and let the linearly independent functions $\{e_{t,k}:E\to\mathbb{R}\}_{k=1}^m$ for each $t\in[T-1]$ be such that $L := \max_{t\in[T-1],k\in[m]} \|e_{t,k}\|_{L^2(\rho_t)}$ and have query cost $\mathcal{T}_{e}$. Also consider $\{z_t:E\to\mathbb{R}\}_{t=0}^T$ with $R := \max_{t\in[T]}\|z_t\|_u < \infty$ and query cost $\mathcal{T}_{z}$. Moreover, let $\sigma_{\min} := \min_{t\in[T-1]} \sigma_{\min}(A_t) > 0$ and assume that $\sqrt{m}RL/\sigma_{\min} \geq 1$.
    Then, for any $\epsilon\in(0,\sigma_{\min}/2]$, Algorithm~{\rm \ref{alg:classicalalgo}} returns an estimate $\widetilde{\mathcal{U}}_0$ such that
    \begin{align*}
    \begin{multlined}[b][\textwidth]
        \operatorname{Pr}\left[\big|\widetilde{\mathcal{U}}_0 - \mathbb{E}[Z_{\tau_0}]\big| \geq 5^T\left(\frac{4\epsilon mRL^2}{\sigma_{\min}^2} + \operatorname*{\max}_{0< t < T} \operatorname*{\min}_{a\in\mathbb{R}^m}\|a\cdot \vec{e}_t(X_t) - \mathbb{E}[Z_{\tau_{t+1}}|X_k]\|_{L^2(\rho_t)}\right) \right]\\
        \leq 6m^2\exp\big(-2N\epsilon^2/(m^2L^2(L^2+R^2))\big)
    \end{multlined}
    \end{align*}
    in time
    \begin{align*}
        O(T m^2 N + Tm^{\omega_\ast} + TN(T + \mathcal T_{\rm samp} + \mathcal T_{z} + m \mathcal T_{e})).
    \end{align*}
\end{theorem}
\begin{proof}
    The error analysis is very similar to that of Theorems~\ref{thr:quantum_result1} and~\ref{thr:quantum_result2}, therefore we shall just point out the required changes. Each entry of $A_t$ and $b_t$ is approximated using a Hoeffding bound, i.e., $\operatorname{Pr}[|(A_t)_{jk} - (\widetilde{A}_t)_{jk}| \geq \epsilon/m] \leq 2e^{-2N\epsilon^2/(m^2L^4)}$ and $\operatorname{Pr}[|(b_t)_{j} - (\widetilde{b}_t)_{j}| \geq \epsilon/\sqrt{m}] \leq 2e^{-2N\epsilon^2/(mR^2L^2)}$, where we have used that $|(A_t)_{jk}| \leq \sqrt{\mathbb{E}[e_{t,j}^2(X_t)]\mathbb{E}[e_{t,k}^2(X_t)]} \leq L^2$ and $|(b_t)_j| \leq R\sqrt{\mathbb{E}[e_{t,j}^2(X_t)]} \leq RL$, for all $j,k\in[m]$ and $t\in[T-1]$. Moreover, $\operatorname{Pr}[|\widetilde{Z}_{\widetilde{\tau}_1} - \mathbb{E}[Z_{\widetilde{\tau}_1}]| \geq \epsilon] \leq 2e^{-2N\epsilon^2/R^2}$. Therefore, by the union bound, all approximations hold with probability at least $1-2m^2e^{-2N\epsilon^2/(m^2L^4)} - 2me^{-2N\epsilon^2/(mR^2L^2)} -2e^{-2N\epsilon^2/R^2} \geq 1 - 6m^2e^{-2N\epsilon^2/(m^2L^2(L^2+R^2))}$.
    
    Regarding the time complexity, the most expensive computational steps are calculating the matrices $\widetilde{A}_t$, which requires $O(Tm^2N)$ time, inverting them and computing the vectors $\widetilde{\alpha}_t$, which requires $O(Tm^{\omega_\ast})$ time, and computing the mapping $\widetilde{\tau}^{(n)}_{t+1} \mapsto Z^{(n)}_{\widetilde{\tau}^{(n)}_{t+1}}$ for all $n\in[N]$ and $t\in[T-1]$, which takes $O(T^2 N)$ time. Sampling $(X_t^{(1)},\dots,X_t^{(N)})_{t=0}^T$ and querying $(Z_t^{(1)},\dots,Z_t^{(N)})_{t=0}^T$ and $(e_{t,k}(X_t^{(1)}),\dots,e_{t,k}(X_t^{(N)}))_{t\in[T-1],k\in[m]}$ require $O(NT(\mathcal{T}_{\rm samp} + \mathcal{T}_{z} + m\mathcal{T}_{e}))$ time. All the other steps, computing $\frac{1}{N}\sum_{n=1}^N Z^{(n)}_{\widetilde{\tau}^{(n)}_{t+1}}\vec{e}(X^{(n)}_t)$, $\widetilde{\mathcal{U}}_0$ and $\widetilde{\alpha}_t\cdot \vec{e}_t(X_t^{(n)})$ require $O(TmN)$ or $O(Tm)$ time.
\end{proof}
If $T_{\rm samp},T_z,T_e=O(1)$, then the complexity is $O(Tm^2N + Tm^{\omega_\ast} + T^2N)$. The factor $Tm^2$ comes from computing $\{A_t\}_{t=1}^{T-1}$ and accounts for runtime instead of only number of samples.

We summarize and compare the results from Theorems~\ref{thr:classical_error} and~\ref{thr:quantum_result2} into a single corollary.
\begin{corollary}
    \label{cor:comparison}
    Within the setting of Algorithm~{\rm \ref{alg:classicalalgo}} with input parameters $\delta,N$ and Algorithm~{\rm \ref{alg:quantumlongstaff}} with input parameters $\delta,\epsilon_0$, let $T\mathcal{T}_{\rm samp}$ be the sampling cost of the Markov chain and consider a set of linearly independent functions $\{e_{t,k}:E\to\mathbb{R}\}_{k=1}^m$ for $t\in[T-1]$ with query cost $\mathcal{T}_{e}$ and $L := \max_{t\in[T-1],k\in[m]} \|e_{t,k}\|_{L^2(\rho_t)}$. Also consider $\{z_t:E\to\mathbb{R}\}_{t=0}^T$ with query cost $\mathcal{T}_{z}$ and $R := \max_{t\in[T]}\|z_t\|_u < \infty$. Let $\sigma_{\min} \leq \min_{t\in[T-1]} \sigma_{\min}(A_t)$ be known. Assume that $\sqrt{m}RL/\sigma_{\min} \geq 1$ and define $\mathcal{T}_{\rm total} := \mathcal{T}_{\rm samp} + \mathcal{T}_{z} + \mathcal{T}_{e}$. For any $\delta\in(0,1)$ and $\epsilon\in\big(0,\frac{2mRL^2}{\sigma_{\min}}\big]$, if $\epsilon_0 = \frac{\epsilon\sigma_{\min}^2}{4mRL^2}$ and $N = \big\lceil\frac{m^2L^2(L^2+R^2)}{2\epsilon_0^{2}}\log(6m^2/\delta)\big\rceil$, then Algorithms~{\rm \ref{alg:classicalalgo}} and~{\rm \ref{alg:quantumlongstaff}} output $\widetilde{\mathcal{U}}_0$ such that
    \begin{align*}
        \operatorname{Pr}\left[|\widetilde{\mathcal{U}}_0 - \mathbb{E}[Z_{\tau_0}]| \geq 5^T\left(\epsilon + \operatorname*{\max}_{0< t < T} \operatorname*{\min}_{a\in\mathbb{R}^m}\|a\cdot \vec{e}_t(X_t) - \mathbb{E}[Z_{\tau_{t+1}}|X_t]\|_{L^2(\rho_t)}\right) \right] \leq \delta
    \end{align*}
    in time, respectively,
    \begin{align*}
        O\left(\frac{Tm^4}{\epsilon^2} \frac{R^2L^6(L^2+R^2)}{\sigma_{\min}^4}(m^2 + T + m\mathcal{T}_{\rm total})\log(m/\delta)\right)
    \end{align*}
   and (up to $\log\log$ factors)
    \begin{align*}
        O\left(\frac{T^2m^4}{\epsilon}\frac{RL^3(L+R)}{\sigma_{\min}^2}\mathcal{T}_{\rm total}\log(T)\log\left(\frac{Tm}{\delta}\right) \log^{3/2}\left(\frac{m RL}{\epsilon \sigma_{\min}}\right) \right).
    \end{align*}
\end{corollary}
\begin{proof}
    The results concerning the quantum Algorithm~\ref{alg:quantumlongstaff} were already proven in Theorem~\ref{thr:quantum_result2}, we just use that $\epsilon_0 = \frac{\epsilon\sigma_{\min}^2}{4mRL^2}$. It is left to prove the results about the classical Algorithm~\ref{alg:classicalalgo}. Indeed, by setting $N = \big\lceil\frac{m^2L^2(L^2 + R^2)}{2\epsilon_0^{2}}\log(6m^2/\delta)\big\rceil$ into Theorem~\ref{thr:classical_error} with $\epsilon_0 = \frac{\epsilon\sigma_{\min}^2}{4mRL^2}$, we obtain
    \begin{align*}
        \operatorname{Pr}\left[|\widetilde{\mathcal{U}}_0 - \mathbb{E}[Z_{\tau_0}]| \geq 5^T\left(\epsilon + \operatorname*{\max}_{0< t < T} \operatorname*{\min}_{a\in\mathbb{R}^m}\|a\cdot \vec{e}_t(X_t) - \mathbb{E}[Z_{\tau_{t+1}}|X_t]\|_{L^2(\rho_t)}\right) \right] \leq \delta
    \end{align*}
    in time
    \begin{align*}
        O(Tm^2 N + Tm^{\omega_\ast} + &TN(T + \mathcal{T}_{\rm samp} + \mathcal{T}_{z} + m\mathcal{T}_{e})) \\
        &= O\left(\frac{Tm^4}{\epsilon^2}(m^2 + T + \mathcal{T}_{\rm samp} + \mathcal{T}_{z} + m\mathcal{T}_{e})\frac{R^2L^6(L^2+R^2)}{\sigma_{\min}^4}\log\Big(\frac{m}{\delta}\Big)\right) \\
        &= O\left(\frac{Tm^4}{\epsilon^2}\frac{R^2L^6(L^2+R^2)}{\sigma_{\min}^4}(m^2 + T+m\mathcal{T}_{\rm total})\log\Big(\frac{m}{\delta}\Big)\right). \qedhere
    \end{align*}
\end{proof}
\begin{remark}
    We note that in the previous theorem it is necessary to know a lower bound on $\min_{t\in[T-1]}\sigma_{\min}(A_t)$ since it must be inputted into quantum subroutines for Monte Carlo to bound the variance. It is possible, though, to get around this detail by using the more complicated algorithm from~{\rm \cite{hamoudi2021quantum}} which does not require a bound on the variance. 
\end{remark}

The accuracy in Corollary~\ref{cor:comparison} is made up of two parts: the ``sampling'' or stochastic error $\epsilon$, and the approximation architecture error $\max_{0< t < T}\min_{a\in\mathbb{R}^m}\|a\cdot \vec{e}_t(X_t) - \mathbb{E}[Z_{\tau_{t+1}}|X_t]\|_{L^2(\rho_t)}$ (ignoring the time dependence). If the sample and query complexities $\mathcal{T}_{\rm samp}, \mathcal{T}_{z}$ and $\mathcal{T}_{e}$ are constant, then we see that the classical and quantum algorithms obtain the same additive accuracy in approximating $\mathcal{U}_0$ in time $\widetilde{O}(T m^4 (T+m^2)/\epsilon^2)$ and $\widetilde{O}(T^2 m^4/\epsilon)$, respectively, thus constituting a quadratic improvement in $\epsilon$. The sampling error can be made arbitrary small by increasing the number of path samples or calls to quantum inputs. The approximation architecture error, on the other hand, does not depend on $\epsilon$, and is intrinsic to projecting the continuation values onto a lower-dimensional space generated by predetermined measurable functions.

\section{Linear finite-dimensional spaces of polynomials}

In this section we shall focus specifically on linear finite-dimensional spaces of polynomials, in order to bound the approximation errors $\operatorname*{\min}_{a\in\mathbb{R}^m}\|a\cdot \vec{e}_t(X_t) - \mathbb{E}[Z_{\tau_{t+1}}|X_t]\|_{L^2(\rho_t)}$ from Corollary~\ref{cor:comparison} and thus obtain overall error estimates. The bounds will come from assuming smoothness properties on the continuation values, which is not an unusual assumption in finance: the continuation values $\mathbb{E}[Z_{\tau_{t+1}}|X_t]$ can even be $C^\infty$-smooth in some models, e.g.\ Black-Scholes~\cite{gerhold2011longstaff}.

Consider an open, connected and bounded set $D\subset\mathbb{R}^d$, together with its closure $\overline{D}$. Let $\mathcal{R}_q$ be the space of all polynomials of degree less than $q$ with real coefficients, i.e., the space of all linear combinations of terms $x_1^{a_1}x_2^{a_2}\dots x_d^{a_d}$ with integers $a_i\geq 0$ such that $\sum_{i=1}^d a_i < q$. Let $\mathcal{R}_q(D)$ be the space of all polynomials in $\mathcal{R}_q$ that are restrictions to $D$, i.e., with domain in $D$. Moreover, for any non-negative integer $n$, let $C^n(\overline{D})$ be the space of all continuous, real-valued functions $f$ on $\overline{D}$ whose continuous derivatives $\partial^a f$ exist on $D$ for all multi-indices $a$ such that $\|a\|_1 = \sum_{i=1}^d a_i \leq n$ and that have continuous extensions onto $\overline{D}$. Define the norm on this space as
\begin{align*}
    \|f\|_{C^n(\overline{D})} := \sum_{\|a\|_1\leq n}\operatorname*{\sup}_{x\in D}|\partial^a f(x)|.
\end{align*}
If $f\in C^n(\overline{D})$, then the next Jackson-type inequality allows one to bound the error committed by approximating $f$ by a polynomial in $\mathcal{R}_q(D)$.
\begin{lemma}[{\cite[Chapter~7, Theorem~6.2]{devore1993constructive}}]
    \label{lem:differentiable_polynomials}
    Let $D\subset\mathbb{R}^d$ be the interior of some closed cube. If $q>n\geq 1$, then, for any $f\in C^n(\overline{D})$,
    \begin{align*}
        \operatorname*{\inf}_{P\in\mathcal{R}_q}\operatorname*{\sup}_{x\in D}|f(x) - P(x)| \leq C_D q^{-n}\|f\|_{C^n(\overline{D})},
    \end{align*}
    where the constant $C_D$ depends on $D$ but not on $q$ or $f$.
\end{lemma}

Another smoothness condition we can consider is that of Lipschitz continuous functions. Let $Q_d(\lambda)$ be the closed cube of side length $0 < 2\lambda < \infty$ in a $d$-dimensional Euclidean space centered at the origin. Let $\operatorname{Lip}(C_L,Q_d(\lambda))$ be the space of Lipschitz continuous functions $f$ on $Q_d(\lambda)$ such that
\begin{align*}
    |f(x) - f(x')| \leq C_L\|x-x'\|_\infty, \quad \forall x,x'\in Q_d(\lambda),
\end{align*}
with Lipschitz constant $C_L$. The Lipschitz condition implies the following error bound for approximating functions by polynomials in $\mathcal{R}_q(Q_d(\lambda))$.
\begin{lemma}[{\cite[Chapter~9, Theorem~10]{feinerman1974polynomial}}]
    \label{lem:lipschitz_polynomials}
    For all $C_L>0$, for all $f\in \operatorname{Lip}(C_L,Q_d(\lambda))$,
    \begin{align*}
        \operatorname*{\inf}_{P\in\mathcal{R}_q}\operatorname*{\sup}_{x\in Q_d(\lambda)}|f(x) - P(x)| \leq \frac{88 \lambda C_L d}{d + q}.
    \end{align*}
\end{lemma}

Given the above theorems, we obtain a corollary of Theorem~\ref{thr:quantum_result2} where the continuation values are assumed to be functions either in $C^n(Q_d(\lambda))$ or in $\operatorname{Lip}(C_L,Q_d(\lambda))$.
\begin{theorem}
    \label{thr:quantum_result3}
    Let $Q_d(\lambda)$ be the closed cube with side length $2\lambda$ in $\mathbb{R}^d$ centered at the origin. Within the setting of Algorithm~{\rm \ref{alg:quantumlongstaff}} with parameters $\delta$ and $\epsilon_0$, let $T\mathcal{T}_{\rm samp}$ be the sampling cost of the Markov chain and consider the approximation architecture $\mathcal{R}_q(Q_d(\lambda))$ for all $t\in\{0,\dots,T-1\}$ with $L := \max_{t\in[T]}\sup_{P\in\mathcal{R}_q} \|P\|_{L^2(Q_d(\lambda),\rho_t)}$ and query cost $\mathcal{T}_{e}$. Also consider $\{z_t:Q_d(\lambda)\to\mathbb{R}\}_{t=0}^T$ with query cost $\mathcal{T}_{z}$ and $R := \max_{t\in[T]}\|z_t\|_u < \infty$. Moreover, let $\sigma_{\min} \leq \min_{t\in[T-1]} \sigma_{\min}(A_t)$ be known. Assume that $\sqrt{m}RL/\sigma_{\min} \geq 1$ and define $\mathcal{T}_{\rm total} := \mathcal{T}_{\rm samp} + \mathcal{T}_{z} + \mathcal{T}_{e}$.
    
    \textbf{\emph{(a).}} Suppose there is a positive integer $n$, $n<q$, such that $\mathbb{E}[Z_{\tau_{t+1}}|X_t]\in C^n(Q_d(\lambda))$ for all $t\in\{0,\dots,T-1\}$. Let $C \geq C_D\max_{0\leq t<T}\|\mathbb{E}[Z_{\tau_{t+1}}|X_t]\|_{C^n(Q_d(\lambda))}$ be known. Then, for $\delta\in(0,1)$, $\epsilon\in\big(0,5^T\big(\frac{RL^2(e^{n}C)^{d/n}}{\sigma_{\min}}\big)^{n/(n+d)}\big]$, and $q = \left \lceil (5^T C/\epsilon)^{1/n} \right \rceil$, if $\epsilon_0 = \frac{\sigma_{\min}^2}{4 (e^{n}C)^{d/n} RL^2}\big(\frac{\epsilon}{5^T}\big)^{1+d/n}$, Algorithm~{\rm \ref{alg:quantumlongstaff}} outputs $\widetilde{\mathcal{U}}_0$ such that
    \begin{align*}
        \operatorname{Pr}\left[|\widetilde{\mathcal{U}}_0 - \mathbb{E}[Z_{\tau_0}]| \leq 2\epsilon \right] \geq 1 - \delta
    \end{align*}
    in time
    \begin{align*}
        \widetilde{O}\left(T^2 (e^{n} C)^{4d/n}\mathcal{T}_{\rm total} \frac{RL^3(R+L)}{\sigma_{\min}^2}\left(\frac{5^T}{\epsilon}\right)^{1+4d/n} \right),
    \end{align*}
    where $\widetilde{O}(\cdot)$ hides $\operatorname{polylog}$ terms in $5^T,R,L,\sigma_{\min},C,e^d,\epsilon,\delta$.
    
    \textbf{\emph{(b).}} Suppose that $\mathbb{E}[Z_{\tau_{t+1}}|X_t]\in \operatorname{Lip}(C_L,Q_d(\lambda))$ for all $t\in\{0,\dots,T-1\}$, and suppose that $C_L$ is known. Then, for $\delta\in(0,1)$, $\epsilon\in\big(0,5^T\big(\frac{RL^2(88 e\lambda C_Ld)^{d}}{\sigma_{\min}}\big)^{1/(1+d)}\big]$, and $q = \lceil 88\cdot 5^T \lambda C_L d/\epsilon\rceil$, if $\epsilon_0 = \frac{\sigma_{\min}^2}{4 (88 e \lambda C_Ld)^{d} RL^2}\big(\frac{\epsilon}{5^T}\big)^{1+d}$, Algorithm~{\rm \ref{alg:quantumlongstaff}} outputs $\widetilde{\mathcal{U}}_0$ such that
    \begin{align*}
        \operatorname{Pr}\left[|\widetilde{\mathcal{U}}_0 - \mathbb{E}[Z_{\tau_0}]| \leq 2\epsilon \right] \geq 1 - \delta
    \end{align*}
    in time
    \begin{align*}
        \widetilde{O}\left(T^2\big(\lambda C_L d\big)^{4d}\mathcal{T}_{\rm total} \frac{RL^3(R+L)}{\sigma_{\min}^2}\left(\frac{5^T}{\epsilon}\right)^{1+4d}\right),
    \end{align*}
    where $\widetilde{O}(\cdot)$ hides $\operatorname{polylog}$ terms in $5^T,R,L,\sigma_{\min},\lambda C_L,e^d,\epsilon$.
\end{theorem}
\begin{proof}
    We start with the case when $\mathbb{E}[Z_{\tau_{t+1}}|X_t]\in C^n(Q_d(\lambda))$ for all $t\in\{0,\dots,T-1\}$. Since the uniform norm upper-bounds the $L^2$-norm over a probability space, and using Lemma~\ref{lem:differentiable_polynomials}, then
    \begin{align*}
        \operatorname*{\inf}_{P\in\mathcal{R}_q(Q_d(\lambda))}\|P - \mathbb{E}[Z_{\tau_{t+1}}|X_t]\|_{L^2(\rho_t)} \leq \operatorname*{\inf}_{P\in\mathcal{R}_q(Q_d(\lambda))}\|P - \mathbb{E}[Z_{\tau_{t+1}}|X_t]\|_{u} \leq Cq^{-n}.
    \end{align*}
    Thus the result follows from Theorem~\ref{thr:quantum_result2} by noting that $\mathcal{R}_q$ has dimension $m := \binom{q-1+d}{d} \leq e^d (q-1)^d \leq e^d\big(\frac{5^T C}{\epsilon}\big)^{d/n}$ by setting $q=\lceil(5^{T} C/\epsilon)^{1/n}\rceil$, and so we must take $m$ linearly independent polynomials in Algorithm~\ref{alg:quantumlongstaff}. By using the accuracy $\epsilon_0 = \frac{\sigma_{\min}^2}{4 (e^{n}C)^{d/n} RL^2}\big(\frac{\epsilon}{5^T}\big)^{1+d/n}$ in Theorem~\ref{thr:quantum_result2} (note that $\epsilon_0\in(0,\sigma_{\min}/2]$), we obtain the final error
    \begin{align*}
        5^T\left(\frac{4\epsilon_0mRL^2}{\sigma_{\min}^2} + \frac{C}{q^n} \right) \leq 5^T\left(\frac{4\epsilon_0e^dRL^2}{\sigma_{\min}^2}\left(\frac{5^T C}{\epsilon}\right)^{d/n} + \frac{C}{q^n} \right) \leq 2\epsilon.
    \end{align*}
    This in turn implies the complexity
    \begin{align*}
        \widetilde{O}\left(\frac{T^2m^3}{\epsilon_0}L(L+R)\mathcal{T}_{\rm total}\right) =
        \widetilde{O}\left(T^2 (e^{n} C)^{4d/n}\mathcal{T}_{\rm total} \frac{RL^3(R+L)}{\sigma_{\min}^2}\left(\frac{5^T}{\epsilon}\right)^{1+4d/n}\right).
    \end{align*}
    
    We now prove the case when $\mathbb{E}[Z_{\tau_{t+1}}|X_t]\in \operatorname{Lip}(C_L,Q_d(\lambda))$ for all $t\in\{0,\dots,T-1\}$. By the same token as before, but now using Lemma~\ref{lem:lipschitz_polynomials}, we obtain
    \begin{align*}
        \operatorname*{\inf}_{P\in\mathcal{R}_q(Q_d(\lambda))}\|P - \mathbb{E}[Z_{\tau_{t+1}}|X_t]\|_{L^2(\rho_t)} \leq \frac{88 \lambda C_L d}{q}.
    \end{align*}
    Once again, we must pick $m\leq e^d (q-1)^d$ linearly independent polynomials in Algorithm~\ref{alg:quantumlongstaff}. By setting $q = \lceil 88\cdot 5^T \lambda C_L d/\epsilon\rceil$ and using the accuracy $\epsilon_0 = \frac{\sigma_{\min}^2}{4 (88 e \lambda C_Ld)^{d} RL^2}\big(\frac{\epsilon}{5^T}\big)^{1+d}$ in Theorem~\ref{thr:quantum_result2} (note that $\epsilon_0\in(0,\sigma_{\min}/2]$), we obtain the final error
    \begin{align*}
        5^T\left(\frac{4\epsilon_0mRL^2}{\sigma_{\min}^2} + \frac{88 \lambda C_L d}{q} \right) \leq 5^T\left(\frac{4\epsilon_0e^dRL^2}{\sigma_{\min}^2}\left(\frac{88\cdot 5^T\lambda C_L d}{\epsilon} \right)^d + \frac{88 \lambda C_L d}{q} \right) \leq 2\epsilon.
    \end{align*}
    This in turn implies the complexity
    \[
        \widetilde{O}\left(\frac{T^2m^3}{\epsilon_0}L(L+R)\mathcal{T}_{\rm total}\right) =
        \widetilde{O}\left(T^2\big(\lambda C_L d\big)^{4d}\mathcal{T}_{\rm total} \frac{RL^3(R+L)}{\sigma_{\min}^2}\left(\frac{5^T}{\epsilon}\right)^{1+4d}\right). \qedhere
    \]
\end{proof}

Similar results can be obtained for the classical Algorithm~\ref{alg:classicalalgo} by using Corollary~\ref{cor:comparison}. In the case of $\mathbb{E}[Z_{\tau_{t+1}}|X_t]\in C^n(Q_d(\lambda))$, pick $m\leq e^d(q-1)^d$, $q = \lceil(5^T C/\epsilon)^{1/n}\rceil$, and final accuracy $\epsilon_f = \frac{\epsilon}{5^T}$ in Corollary~\ref{cor:comparison}, which implies time 
\begin{align*}
    \widetilde{O}\left(\frac{Tm^4}{\epsilon_f^2}(m^2 + T + m\mathcal{T}_{\rm total}) \frac{R^2L^4}{\sigma_{\min}^4}\right) = \widetilde{O}\left(T^2 (e^{n} C)^{6d/n}\frac{R^2L^4}{\sigma_{\min}^4}\left(\frac{5^T}{\epsilon}\right)^{2+6d/n}\mathcal{T}_{\rm total}\right)
\end{align*}
up to the logarithm factor $\log(m/\delta)$.

In the case of $\mathbb{E}[Z_{\tau_{t+1}}|X_t]\in \operatorname{Lip}(C_L,Q_d(\lambda))$, pick $m\leq e^d(q-1)^d$, $q = \lceil 88\cdot 5^T \lambda C_L d/\epsilon\rceil$, and final accuracy $\epsilon_f = \frac{\epsilon}{5^T}$ in Corollary~\ref{cor:comparison}, which implies time 
\begin{align*}
    \widetilde{O}\left(\frac{Tm^4}{\epsilon_f^2}(m^2 + T + \mathcal{T}_{\rm total})\frac{R^2L^6(L^2+R^2)}{\sigma_{\min}^4}\right) = 
    \widetilde{O}\left(T^2 (\lambda C_L d)^{6d}\frac{R^2L^6(L^2+R^2)}{\sigma_{\min}^4}\left(\frac{5^T}{\epsilon}\right)^{2+6d}\mathcal{T}_{\rm total} \right)
\end{align*}
up to the logarithm factor $\log(m/\delta)$.

By assuming sufficiently smooth conditions on the continuation values, we obtain the final classical and quantum complexities $\widetilde{O}((5^T/\epsilon)^{2+6d/n})$ and $\widetilde{O}((5^T/\epsilon)^{1+4d/n})$, respectively. If the continuation values are $n$-differentiable functions for $n = \Theta(\log(5^T/\epsilon)/\log\log(5^T/\epsilon))$ --- note that $n\leq q = O((5^T/\epsilon)^{1/n})$ ---, then we obtain the simpler runtimes $\widetilde{O}((5^T/\epsilon)^{2})$ and $\widetilde{O}(5^T/\epsilon)$, thus constituting a quadratic improvement from classical to quantum complexity. The exponential time dependence comes from the error precision required to obtain a final additive accuracy of $\epsilon$, which is exactly the quantity that is sped up by using quantum algorithm for Monte Carlo, hence the quadratic improvement on time. Finally, we note that, given the condition $n\leq q$ and the value $q=O((5^T/\epsilon)^{1/n})$, it is not possible to take $n\to\infty$. The value $n = \Theta(\log(5^T/\epsilon)/\log\log(5^T/\epsilon))$ implies $n\log{n} = \Theta\left(\frac{\log(5^T/\epsilon)}{\log\log(5^T/\epsilon)}\log\left(\frac{\log(5^T/\epsilon)}{\log\log(5^T/\epsilon)}\right)\right) = O(\log(5^T/\epsilon))$, and thus does not violate $n\leq q$.

Theorem~\ref{thr:quantum_result3} applies to any set of polynomials, e.g.\ Chebyshev polynomials. We can thus compare the above result with the one from Miyamoto~\cite{miyamoto2021bermudan}, which approximates the continuation values by a Chebyshev interpolation through Chebyshev nodes and obtains a complexity of $O(\epsilon^{-1}\log^d(1/\epsilon)\text{poly}\log\log(1/\epsilon))$. Our approach, in contrast, is to project $\mathbb{E}[Z_{\tau_{t+1}}|X_t]$ onto a set of polynomials and is, for this reason, much more general. Moreover, our final result is a \emph{time} complexity, while the result from~\cite{miyamoto2021bermudan} is a \emph{query} complexity on the number of unitaries called by all quantum routines for Monte Carlo (thus being a weaker result). Finally, Miyamoto~\cite{miyamoto2021bermudan} assumes that the continuation values are analytical functions, i.e., are in $C^\infty$, while we only need to assume $\mathbb{E}[Z_{\tau_{t+1}}|X_t]\in C^n(Q_d(\lambda))$ for $n = \Theta(\log(5^T/\epsilon)/\log\log(5^T/\epsilon))$ in order to recover $\widetilde{O}(\epsilon^{-1})$ up to polylog factors. One downside of our approach, though, is the presence of the (non-zero) minimum singular value $\sigma_{\min}$ which depends on the probability measures $\rho_t$. Chebyshev interpolation bypasses such direct probability dependence by shifting it entirely to the approximation of the continuation values on Chebyshev nodes via (quantum) Monte Carlo. 

We can take advantage of our freedom in choosing the approximation architecture in order to simplify the {\rm LSM} algorithm, specially the matrices $A_t$, depending on the underlying probability model. In this section we analyze the common cases of Brownian motion and geometric Brownian motion. 

\subsection{Brownian motion}
\label{sec:brownian}

Consider a Markov chain $(X_t)_{t=0}^T$ where each component of $X_t :\Omega \to \mathbb{R}^d$ follows an independent Brownian motion. A Brownian motion can be described by a stochastic process $W=(W_t)_{t=0}^T$ (called Wiener process) such that:
\begin{enumerate}
\item $W_0=0$;
\item $W_t$ is almost surely continuous in $t$;
\item $W_{t+s} - W_s$ for $t\geq 0$ is independent of $W_s$, for $s\leq t$;
\item $W_{t+u} - W_t \sim \mathcal{N}(0,u)$ for $u\geq 0$, where $\mathcal{N}(\mu,\sigma^2)$ denotes the normal distribution.
\end{enumerate}
It follows from the above definition that the image probability measure $\rho_t$ of $X_t$ reads $\rho_t(x) = \prod_{i=1}^d\frac{1}{\sqrt{2\pi t}}e^{-x_i^2/2t}$ for $t\in[T]$. Given $t\in[T]$, consider the approximation architecture generated by $\{e_{t,\vec{k}}:\mathbb{R}^d\to\mathbb{R}\}_{\vec{k}\in \mathcal{B}}$, where $\mathcal{B} := \{\vec{k}\in\{0,\dots,q-1\}^d: \sum_{i=1}^d k_i < q\}$ for some integer $q\geq 2$ and
\begin{align}
    \label{eq:hermite}
    e_{t,\vec{k}}(x) = \prod_{l=1}^d \frac{1}{\sqrt{k_l!2^{k_l}}}H_{k_l}(x_l/\sqrt{2t}),
\end{align}
with $x\in \mathbb{R}^d$ and the $H_{l}(y)$ are Hermite polynomials defined by
\begin{align*}
    H_{l}(y) := (-1)^{l} e^{y^2}\frac{\text{d}^{l}}{\text{d}y^{l}}e^{-y^2}
\end{align*}
where $l \in \mathbb{N}$ and $H_{0}(x):=1$. The choice for Hermite polynomials comes from the fact that they form an orthogonal basis under the measure $\rho_t(x)$, i.e., 
\begin{align*}
    \int_{-\infty}^\infty H_k(x/\sqrt{2t})H_l(x/\sqrt{2t})\frac{e^{-x^2/2t}}{\sqrt{2\pi t}} \text{d}x = \int_{-\infty}^\infty H_k(x)H_l(x)\frac{e^{-x^2}}{\sqrt{\pi}} \text{d}x = k!2^k\delta_{kl}.
\end{align*}
This means that, for each $t\in[T-1]$, $\{e_{t,\vec{k}}\}_{\vec{k}\in\mathcal{B}}$ are orthonormal under $\rho_t(x)$ and thus the matrices $A_t$ become simply the identity matrix. However, in order to invoke approximation results due to smoothness assumptions, e.g.\ the Jackson-like inequality in Lemma~\ref{lem:differentiable_polynomials}, we must restrict the domain to a close set like $Q_d(\lambda)$. We shall then consider the approximation architecture generated by $\{e^{\rm trunc}_{t,\vec{k}}:\mathbb{R}^d\to \mathbb{R}\}_{\vec{k}\in\mathcal{B}}$ with $e^{\rm trunc}_{t,\vec{k}}$ defined as in Eq.~\eqref{eq:hermite} if $x\in Q_d(\lambda)$ and $e^{\rm trunc}_{t,\vec{k}}(x) = 0$ otherwise. The implied matrix $A^{\rm trunc}_t$ is no longer equal to the identity matrix, but is close to it. Such closeness will be approximated with the aid of the following result whose proof is postponed to Appendix~\ref{sec:hermite_bound}.
\begin{lemma}
    \label{lem:hermite_bound}
    Given $\lambda > 0$ and $k,l\in\mathbb{N}$ with $l\leq k$, then
    \begin{align*}
    	\left|\int_\lambda^\infty H_k(x)H_l(x) e^{-x^2}\text{d}x\right| \leq 2^{2+(k+l)/2}\sqrt{(k+1)!(l+1)!}e^{(\sqrt{2(k+1)}+\sqrt{2(l+1)})\lambda}e^{-\lambda^2}.
    \end{align*}
\end{lemma}
This means that, if we approximate the matrices $A^{\rm trunc}_t$ by the identity matrix, the error incurred will be small for sufficiently large $\lambda$ and, therefore, the complexity of Algorithms~\ref{alg:classicalalgo} and~\ref{alg:quantumlongstaff} can be reduced since the matrices $\widetilde{A}_t$ no longer need to be computed and inverted. The vector $b_t$, $t\in[T-1]$, will still be approximated using quantum algorithm for Monte Carlo, which, in turn, must use a discretization of the Brownian motion. In order to formalize these points in a theorem, we first define the truncation operator $\mathscr{T}_\beta$ such that, for any function $f:\mathbb{R}^d\to\mathbb{R}$,
\begin{align*}
    \mathscr{T}_\beta f(x) = \begin{cases}
           f(x) \quad &\text{if}~|f(x)|\leq \beta,\\
           \beta \operatorname{sgn}(f(x)) \quad &\text{otherwise}.
    \end{cases}
\end{align*}
Define also the quantity
\begin{align}\label{eq:r_p-def1}
    r_p := T\sqrt{\frac{2M_p}{p-2}} + \operatorname*{\max}_{t\in[T-1]}\sum_{i=1}^d \mathbb{E}[|X_{t,i}|] = T\sqrt{\frac{2M_p}{p-2}} + d\sqrt{\frac{2(T-1)}{\pi}}
\end{align}
for $2< p < \infty$, where $M_p := \max\big\{\|z_1\|_{L^p(\rho_1)}^p,\dots,\|z_T\|_{L^p(\rho_T)}^p\big\}$ and we used $\mathbb{E}[|X_{t,i}|] = \sqrt{\frac{2t}{\pi}}$. For simplicity, the next theorem just considers the case when the continuation values are in $C^n(Q_d(\lambda))$. Moreover, we stress that $L^2(\rho_t) = L^2(\mathbb{R}^d, \rho_t)$ considers the whole space $\mathbb{R}^d$.
\begin{theorem}
    \label{thr:brownian}
    Consider a Markov chain $(X_t)_{t=0}^T$ where each component of $X_t:\Omega\to \mathbb{R}^d$ follows an independent Brownian motion. Let $Q_d(\lambda)$ be the closed cube with side length $2\lambda$ in $\mathbb{R}^d$ centered at the origin. Suppose that $M_p := \max\big\{\|z_1\|_{L^p(\rho_1)}^p,\dots,\|z_T\|_{L^p(\rho_T)}^p\big\} < \infty$ for some $2 < p < \infty$ and let $r_p>0$ as in Eq.~\eqref{eq:r_p-def1}. Within the setting of Algorithm~{\rm \ref{alg:quantumlongstaff}} with parameters $\delta$ and $\epsilon_0$, let $T\mathcal{T}_{\rm samp}$ be the sampling cost of the Markov chain and set all matrices $\widetilde{A}_t$, $t\in[T-1]$, to be the identity matrix and use the payoffs $z_0,\mathscr{T}_\beta z_1,\dots,\mathscr{T}_\beta z_T$, where $\beta := \lambda^{2/p}$, with query cost $\mathcal{T}_{z}$ and $R := \max_{t\in[T]}\|\mathscr{T}_\beta z_t\|_u$. Also, for $t\in[T-1]$, consider the approximation architectures $\mathscr{H}_t = \{e^{\rm trunc}_{t,\vec{k}}:\mathbb{R}^d\to\mathbb{R}\}_{\vec{k}\in\mathcal{B}}$ where $\mathcal{B} := \{\vec{k}\in\{0,\dots,q-1\}^d:\sum_{i=1}^d k_i < q\}$ for $q\geq 2$ and $e^{\rm trunc}_{t,\vec{k}}(x)$ is defined as in Eq.~\eqref{eq:hermite} if $x\in Q_d(\lambda)$ and $e^{\rm trunc}_{t,\vec{k}}(x) = 0$ otherwise. Let $\mathcal{T}_{e}$ be the query cost of $\mathscr{H}_t$, $t\in[T]$. Suppose there is $n\in\mathbb{N}$, $n<q$, such that the restriction $\mathbb{E}[Z_{\tau_{t+1}}|X_t]|_{Q_d(\lambda)}$ to $Q_d(\lambda)$ of each $\mathbb{E}[Z_{\tau_{t+1}}|X_t]$, $t\in\{0,\dots,T-1\}$, is in $C^n(Q_d(\lambda))$. Let $C \geq C_{d,\lambda}\max_{0\leq t<T}\|\mathbb{E}[Z_{\tau_{t+1}}|X_t]|_{Q_d(\lambda)}\|_{C^n(Q_d(\lambda))}$ be known, where $C_{d,\lambda}$ is a constant depending only on $Q_d(\lambda)$. Then, for $\delta\in(0,1)$, $\epsilon\in\big(0,5^T(e^{d}C^{d/n}R)^{n/(n+d)}\big]$, $q = \left \lceil (5^T C/\epsilon)^{1/n} \right \rceil$, and 
    \begin{align}
        \label{eq:lamba_bounds}
        \lambda = O\left(\max\left\{\sqrt{T}\left(\frac{5^T C}{\epsilon}\right)^{1/2n},T\ln\left((e^{n} C)^{2d/n} Rd\left(\frac{5^T}{\epsilon}\right)^{1+2d/n}\right), \left(\frac{5^Tr_p}{\epsilon}\right)^{p/(p-2)}\right\}\right),
    \end{align} 
    if $\epsilon_0 = \frac{1}{3(e^{n}C)^{d/n} R}\big(\frac{\epsilon}{5^T}\big)^{1+d/n}$, Algorithm~{\rm \ref{alg:quantumlongstaff}} outputs $\widetilde{\mathcal{U}}_0 := \max\big\{Z_0,\widetilde{(\mathscr{T}_\beta Z)}_{\widetilde{\tau}_1}\big\}$ such that
    \begin{align*}
        \operatorname{Pr}\left[|\widetilde{\mathcal{U}}_0 - \mathbb{E}[Z_{\tau_0}]| \leq 3\epsilon \right] \geq 1 - \delta
    \end{align*}
    in time 
    \begin{align*}
        \widetilde{O}\left(\mathcal{T}_{\rm total}(e^{n} C)^{7d/2n} R^2\left(\frac{5^T}{\epsilon}\right)^{1+7d/2n}\right),
    \end{align*}
    where $\mathcal{T}_{\rm total} := \mathcal{T}_{\rm samp} + \mathcal{T}_{z} + \mathcal{T}_{e}$ and $\widetilde{O}(\cdot)$ hides $\operatorname{polylog}$ terms in $5^T,R,C,e^d,\epsilon,\delta$.
\end{theorem}
\begin{proof}
    First we prove that Algorithm~\ref{alg:quantumlongstaff}, when using $m := |\mathcal{B}| = \binom{q-1+d}{d}$ different Hermite polynomials from Eq.~\eqref{eq:hermite} as approximation architecture and parameters $\epsilon_b = \epsilon_0/\sqrt{m}$ and $\delta_b = \delta/(2Tm)$ (there are no $\epsilon_A$ and $\delta_A$ parameters), outputs $\widetilde{(\mathscr{T}_\beta Z)}_{\widetilde{\tau}_1}$ such that, for any $\epsilon_0\in(0,1/2]$,
    \begin{align*}
        \operatorname{Pr}\left[\big|\widetilde{(\mathscr{T}_\beta Z)}_{\widetilde{\tau}_1} - \mathbb{E}[(\mathscr{T}_\beta Z)_{\tau_1}]\big| \geq 5^T\left(3m\epsilon_0 R + \operatorname*{\max}_{0< t<T} \operatorname*{\min}_{a\in\mathbb{R}^{m}}\|a\cdot \vec{e}_t(X_t) - \mathbb{E}[(\mathscr{T}_\beta Z)_{\tau_{t+1}}|X_t]\|_{L^2(\rho_t)}\right) \right] \leq \delta,
    \end{align*}
    where $\vec{e}_t(\cdot) := (e^{\rm trunc}_{t,\vec{k}}(\cdot))_{\vec{k}\in\mathcal{B}}$, in time
    \begin{align}
        \label{eq:inter_complex}
        O\left(\frac{T^2m^{5/2}}{\epsilon_0}\mathcal{T}_{\rm total}R\log(T)\log(Tm/\delta) \log^{3/2}(mR/\epsilon_0)\log\log(mR/\epsilon_0) \right).
    \end{align}
    Indeed, the proof is very similar to the ones from Theorems~\ref{thr:quantum_result1} and~\ref{thr:quantum_result2}. According to Lemma~\ref{lem:hermite_bound}, approximating the matrices $A^{\rm trunc}_t$ by the identity matrix leads to an error in any diagonal entry $\mathbb{E}[e^{\rm trunc}_{t,\vec{k}}(X_t)e^{\rm trunc}_{t,\vec{k}}(X_t)]$ of at most
    \begin{align}
        1 - \prod_{i=1}^d\left(1- \int_{(-\infty,-\lambda]\cup[\lambda,\infty)} \frac{H_{k_i}(x/\sqrt{2t})^2}{k_i!2^{k_i}}\frac{e^{-x^2/2t}}{\sqrt{2\pi t}}\text{d}x\right) &\leq 2\sum_{i=1}^d \left|\int_{\lambda}^\infty \frac{H_{k_i}(x/\sqrt{2t})^2}{k_i! 2^{k_i}}\frac{e^{-x^2/2t}}{\sqrt{2\pi t}}\text{d}x \right| \label{eq:step1}\\
        &\leq \frac{8e^{-\lambda^2/2t}}{\sqrt{\pi}}\sum_{i=1}^d(k_i+1)e^{2\lambda\sqrt{(k_i+1)/t}}\nonumber\\
        &\leq \frac{8e^{-\lambda^2/2t}}{\sqrt{\pi}}\sum_{i=1}^d e^{4\lambda\sqrt{(k_i+1)/t}}\label{eq:step2}\\
        &\leq \frac{8d}{\sqrt{\pi}} e^{4\lambda\sqrt{q/t}}e^{-\lambda^2/2t}.\nonumber
    \end{align}
    where Eq.~\eqref{eq:step1} follows from convexity and Eq.~\eqref{eq:step2} from $(k_i+1)\leq e^{2\lambda\sqrt{(k_i+1)/t}}$ if $\lambda \geq \sqrt{T}/e$ (which comes from $\lambda \geq \frac{\sqrt{T}}{2}\max_{x\geq 1}\frac{\ln{x}}{\sqrt{x}} = \frac{\sqrt{T}}{e}$).
    If $\lambda \geq \max\big\{16\sqrt{qT},4T\log(5md/\epsilon_0)\big\}$, the above error is at most $\frac{8d}{\sqrt{\pi}}e^{-\lambda^2/4t} \leq \frac{\epsilon_0}{m}$. A similar bound can be obtained for off-diagonal terms $\mathbb{E}[e^{\rm trunc}_{t,\vec{k}}(X_t)e^{\rm trunc}_{t,\vec{l}}(X_t)]$:
    \begin{align*}
        \prod_{i=1}^d\left|\int_{(-\infty,-\lambda]\cup[\lambda,\infty)} \frac{H_{k_i}(x/\sqrt{2t})H_{l_i}(x/\sqrt{2t})}{\sqrt{k_i!l_i!2^{k_i+l_i}}}\frac{e^{-x^2/2t}}{\sqrt{2\pi t}}\text{d}x\right| \leq \left(\frac{8}{\sqrt{\pi}} e^{4\lambda\sqrt{q/t}}e^{-\lambda^2/2t}\right)^d \leq \frac{\epsilon_0}{m}
    \end{align*}
    using the bounds on $\lambda$. The above bounds imply the $L^2$-error between $A^{\rm trunc}_t$ and the identity is $\|A^{\rm trunc}_t-\mathbb{I}\|_2 \leq \epsilon_0$. The $L^2$-error between $b_t$ and $\widetilde{b}_t$ is still $\|b_t - \widetilde{b}_t\|_2\leq \epsilon_0$ with probability at least $1-\delta/2T$. According to Theorem~\ref{thr:matrix_error}, we obtain
    \begin{align*}
        \|\widetilde{\alpha}_t - \alpha_t\|_2 = \|\widetilde{b}_t - (A^{\rm trunc}_t)^{-1} b_t\|_2 \leq \|\widetilde{b}_t - b_t\|_2 + \|b_t - (A^{\rm trunc}_t)^{-1} b_t\|_2 \leq \epsilon_0 + 2\epsilon_0 \|b_t\|_2 \leq 3\epsilon_0\sqrt{m}R,
    \end{align*} 
    where we used that $\sigma_{\min}(\mathbb{I}) = 1$ and that $\|b_t\|_2 \leq R\|\mathbb{E}[\vec{e}_t(X_t)]\|_2 \leq \sqrt{m}R$ since, by the definition of $e^{\rm trunc}_{t,\vec{k}}$, we have $\max_{\vec{k}\in\mathcal{B},t\in[T]}\|e^{\rm trunc}_{t,\vec{k}}\|_{L^2(\rho_t)} \leq 1$. This leads, with probability at least $1-\delta/2T$, to
    \begin{align*}
        \|\widetilde{\alpha}_t\cdot \vec{e}_t(X_t) - \alpha_t\cdot \vec{e}_t(X_t)\|_{L^2(\rho_t)} \leq \|\widetilde{\alpha}_t - \alpha_t\|_2 \big\|\sqrt{\vec{e}_t(X_t)\cdot\vec{e}_t(X_t)}\big\|_{L^2(\rho_t)} \leq 3m\epsilon_0 R,
    \end{align*}
   using $\big\|\sqrt{\vec{e}_t(X_t)\cdot\vec{e}_t(X_t)}\big\|_{L^2(\rho_t)} = \sqrt{\sum_{x\in Q_d(\lambda)}\rho_t(x)\sum_{\vec{k}\in\mathcal{B}} \big(e^{\rm trunc}_{t,\vec{k}}(x)\big)^2} \leq \sqrt{m}$. The rest of the error analysis is the same and we omit it here, the only difference being that the payoffs are now $z_0,\mathscr{T}_\beta z_1,\dots,\mathscr{T}_\beta z_T$. Regarding the complexity, the most expensive step is the computation of all $b_t$, which requires time
    \begin{align*}
        O\left(\frac{Tm}{\epsilon_b}R(T\mathcal{T}_{\rm samp} + T(\log(T)\mathcal{T}_{z} + m\mathcal{T}_{e}))\log\left(\frac{1}{\delta_b}\right) \log^{3/2}\left(\frac{R}{\epsilon_b} \right)\log\log\left(\frac{R}{\epsilon_b}\right) \right),
    \end{align*}
    by calling $\mathtt{QMonteCarlo}(Z_{\widetilde{\tau}_{t}} e_{t-1,\vec{k}}(X_{t-1}), \epsilon_b, \delta_b, R)$ from Theorem~\ref{thr:quantum_monte_carlo} for $t\in\{2,\dots,T\}$ and $\vec{k}\in\mathcal{B}$. The term $T(\log(T)\mathcal{T}_{z} + m\mathcal{T}_{e})$ comes from the unitaries $C_t^{(\vec{k})}$ in $\mathtt{QMonteCarlo}$. Hence, by keeping the largest terms of each complexity and using that $\epsilon_b = \epsilon_0/\sqrt{m}$ and $\delta_b = \delta/(4Tm)$, the final complexity is the one from Eq.~\eqref{eq:inter_complex}.
    
    With this partial result in hands, we again use the Jackson-type inequality from Lemma~\ref{lem:differentiable_polynomials} as in Theorem~\ref{thr:quantum_result3}. In order to do so, we first need the following technical result that limits the continuation values to being a function in the domain $Q_d(\lambda)$. More specifically, by~\cite[Proposition~3.14]{zanger2013quantitative},
    \begin{align}
        \label{eq:truncation_inequality}
        \|\mathbf{1}_{Q_d(\lambda)}\mathbb{E}[Z_{\tau_{t+1}}|X_t] - \mathbb{E}[(\mathscr{T}_\beta Z)_{\tau_{t+1}}|X_t]\|_{L^2(\rho_t)} \leq r_p\lambda^{2/p-1},
    \end{align}
    where $\mathbf{1}_{Q_d(\lambda)}$ is the indicator random variable corresponding to the set $Q_d(\lambda)$. Then, since $\{e^{\rm trunc}_{t,\vec{k}}\}_{\vec{k}\in\mathcal{B}}$ spans $\mathcal{R}_q(Q_d(\lambda))$, the Jackson-type inequality from Lemma~\ref{lem:differentiable_polynomials} reads
    \begin{align*}
        \operatorname*{\min}_{a\in\mathbb{R}^{m}}\|a\cdot \vec{e}_t(X_t) - \mathbf{1}_{Q_d(\lambda)}\mathbb{E}[Z_{\tau_{t+1}}|X_t]\|_{L^2(\rho_t)} = \operatorname*{\inf}_{P\in\mathcal{R}_q}\|(P - \mathbb{E}[Z_{\tau_{t+1}}|X_t])\mathbf{1}_{Q_d(\lambda)}\|_{L^2(\rho_t)} \leq Cq^{-n},
    \end{align*}
    for all $t\in\{0,\dots,T-1\}$. Combining both results, we finally get that
    \begin{align*}
        \big|\widetilde{(\mathscr{T}_\beta Z)}_{\widetilde{\tau}_1} - \mathbb{E}[(\mathscr{T}_\beta Z)_{\tau_1}]\big| \leq 5^T\left(3m\epsilon_0 R + Cq^{-n} + r_p\lambda^{2/p-1}\right),
    \end{align*}
    with probability at least $1-\delta$. Before we plug values for $m$, $q$, $\epsilon_0$ and $\epsilon$, we get rid of the truncation operator $\mathscr{T}_\beta$ in the inequality above by using that~\cite[Proposition~5.2]{egloff2005monte}
    \begin{align*}
        \left \|\mathbb{E}[Z_{\tau_{t+1}}|X_t] - \mathbb{E}[(\mathscr{T}_\beta Z)_{\tau_{t+1}}|X_t]\right\|_{L^2(\rho_t)} \leq r_p\lambda^{1-p/2}
    \end{align*}
    (which is quite similar to Eq.~\eqref{eq:truncation_inequality}). By using that $|{\max}\{a_0,a_1\}-\max\{a_0,a_2\}|\leq |a_1-a_2|$ with $a_0,a_1,a_2\in\mathbb{R}$, then
    \begin{align*}
        |\widetilde{\mathcal{U}}_0 - \mathbb{E}[Z_{\tau_0}]| \leq \big|\widetilde{(\mathscr{T}_\beta Z)}_{\widetilde{\tau}_1} - \mathbb{E}[Z_{\tau_1}]\big| \leq 5^T\left(3m\epsilon_0 R + Cq^{-n} + 2r_p\lambda^{2/p-1}\right),
    \end{align*}
    since $\lambda^{1-p/2} \leq \lambda^{2/p-1}$.
    
    We use that $m \leq e^d (q-1)^d$, $q=\lceil(5^T C/\epsilon)^{1/n}\rceil$, $\lambda \geq \big(\frac{2r_p 5^T}{\epsilon}\big)^{p/(p-2)}$ and take $\epsilon_0 = \frac{1}{3 (e^nC)^{d/n} R}\big(\frac{\epsilon}{5^T}\big)^{1+d/n}$ (note that $\epsilon_0\in(0,1/2]$) in order to obtain the final error
    \begin{align*}
        5^T\left(3m\epsilon_0 R + \frac{C}{q^n} + 2r_p\lambda^{2/p-1} \right) \leq 5^T\left(3\epsilon_0 e^{d} R\left(\frac{5^T C}{\epsilon}\right)^{n/d} + \frac{C}{q^n} + 2r_p\lambda^{2/p-1} \right) \leq 3\epsilon.
    \end{align*}
    The above in turn implies the complexity, up to $\operatorname{polylog}$ factors,
        \[
        \widetilde{O}\left(\frac{T^2 m^{5/2}}{\epsilon_0}\mathcal{T}_{\rm total} R\right) = 
        \widetilde{O}\left(\mathcal{T}_{\rm total} (e^{n} C)^{7d/2n} R^2\left(\frac{5^T}{\epsilon}\right)^{1+7d/2n}\right). \qedhere
    \]
\end{proof}

A similar result can be obtained for the classical {\rm LSM} algorithm. Similarly to Theorems~\ref{thr:classical_error} and~\ref{thr:brownian}, Algorithm~\ref{alg:classicalalgo}, by using the same bounds on $\lambda$ (Eq.~\eqref{eq:lamba_bounds}), outputs $\widetilde{(\mathscr{T}_\beta Z)}_{\widetilde{\tau}_1}$ such that, for any $\epsilon_0\in(0,1/2]$,
\begin{align*}
\begin{multlined}[\textwidth][b]
    \operatorname{Pr}\left[\big|\widetilde{(\mathscr{T}_\beta Z)}_{\widetilde{\tau}_1} - \mathbb{E}[(\mathscr{T}_\beta Z)_{\tau_1}]\big| \geq 5^T\left(3m\epsilon_0 R + \operatorname*{\max}_{0< t<T} \operatorname*{\min}_{a\in\mathbb{R}^{m}}\|a\cdot \vec{e}_t(X_t) - \mathbb{E}[(\mathscr{T}_\beta Z)_{\tau_{t+1}}|X_t]\|_{L^2(\rho_t)}\right) \right] \\
    \leq 4m\exp\big(-2N\epsilon_0^2/(mR^2)\big),
\end{multlined}
\end{align*}
in time $O(TmN + TN(T+ \mathcal{T}_{\rm samp} + \mathcal{T}_{z} + m\mathcal{T}_{e})) = O(TN(T + m\mathcal{T}_{\rm total}))$, where $O(TmN)$ comes from computing $\widetilde{b}_t$, $t\in[T-1]$. The term $4me^{-2N\epsilon_0^2/(mR^2)}$ comes from the Hoeffding bounds $\operatorname{Pr}[|(b_t)_{j} - (\widetilde{b}_t)_{j}| \geq \epsilon/\sqrt{m}] \leq 2e^{-2N\epsilon^2/(mR^2)}$ for all $j\in[m]$, and $\operatorname{Pr}[|\widetilde{Z}_{\widetilde{\tau}_1} - \mathbb{E}[Z_{\widetilde{\tau}_1}]| \geq \epsilon] \leq 2e^{-2N\epsilon^2/R^2}$ (note that there is no Hoeffding bound related to matrices $A_t$). The application of the Jackson-like inequality is the same as in the previous theorem. By taking $N = \lceil \frac{mR^2}{2\epsilon_0^2}\log(4m/\delta)\rceil$ and the same parameters $m\leq e^d(q-1)^d$, $q=\lceil(5^T C/\epsilon)^{1/n}\rceil$, and $\epsilon_0 = \frac{1}{3(e^nC)^{d/n} R}\big(\frac{\epsilon}{5^T}\big)^{1+d/n}$ as before, the final complexity is
\begin{align*}
    O(TN(T+m\mathcal{T}_{\rm total})) = \widetilde{O}\left(\frac{TmR^2}{\epsilon_0^2}(T+m\mathcal{T}_{\rm total})\right) = \widetilde{O}\left(\mathcal{T}_{\rm total}(e^{n} C)^{4d/n}R^4\left(\frac{5^T}{\epsilon}\right)^{2+4d/n}\right),
\end{align*}
up to the $\log(m/\delta)$ factor.

In the particular case when the Markov chain follows a Brownian motion, approximating the matrices by the identity matrix not only leads to a simplification of the {\rm LSM} algorithm, but bounds the minimum singular values that showed up in previous complexities by a constant (to be more specifically, by $\sigma_{\rm min} = 1$). Even more, there is a mild improvement on the final classical and quantum complexities to $\widetilde{O}((5^T/\epsilon)^{2+4d/n})$ and $\widetilde{O}((5^T/\epsilon)^{1+7d/2n})$, respectively. Still, if the continuation values are $n$-differentiable for $n = \Theta(\log(5^T/\epsilon)/\log\log(5^T/\epsilon))$, there is a quantum quadratic improvement in the runtime. Finally, we must mention that, even though we bounded the minimum singular value by a constant, there is still a hidden dependence on $\epsilon$ in the parameter $C \geq C_{d,\lambda}\max_{0\leq t<T}\|\mathbb{E}[Z_{\tau_{t+1}}|X_t]|_{Q_d(\lambda)}\|_{C^n(Q_d(\lambda))}$, since the constant $C_{d,\lambda}$ in the Jackson-like inequality might depend on the size of the cube $\lambda$. This is the reason why we cannot take $\lambda \to \infty$ in the previous theorem.

\subsection{Geometric Brownian motion}
\label{sec:geo_brownian}

We now consider a Markov chain $(X_t)_{t=0}^T$ where each component of $X_t:\Omega\to \mathbb{R}^d$ follows an independent geometric Brownian motion. A geometric Brownian motion is a stochastic process $S=(S_t)_{t=0}^T$ such that $S_t = e^{W_t - t/2}$, where $(W_t)_{t=0}^T$ is a standard Brownian motion.\footnote{More generally, $S_t = S_0 e^{\sigma W_t + t(\mu -\sigma^2/2)}$, but we shall assume $S_0=1$, $\sigma=1$ and $\mu = 0$ here for simplicity.} The probability density function of a geometric Brownian motion is $f_t(s) = \frac{1}{\sqrt{2\pi t} s}\exp\big(-\frac{(\ln{s} + t/2)^2}{2t}\big)$, and thus the image probability measure $\rho_t$ of $X_t$ reads $\rho_t(x) = \prod_{i=1}^d \frac{1}{\sqrt{2\pi t} x_i}\exp\big(-\frac{(\ln{x_i} + t/2)^2}{2t}\big)$ for $t\in[T]$. Given $t\in[T]$, consider the approximation architecture generated by $\{e_{t,\vec{k}}:\mathbb{R}^d\to\mathbb{R}\}_{\vec{k}\in\{0,\dots,q-1\}^d}$ for some integer $q\geq 2$, where $e_{t,\vec{k}}(x) = \prod_{i=1}^d x_i^{k_i}e^{-k_i(k_i-1)\frac{t}{2}}$. Given this choice of functions, we have that
\begin{align*}
    \mathbb{E}[e_{t,\vec{k}}(X_t)e_{t,\vec{l}}(X_t)] &= \prod_{i,j=1}^d e^{(k_i-k_i^2+l_j-l_j^2)\frac{t}{2}} \mathbb{E}[X_{t,i}^{k_i}X_{t,j}^{l_j}] \\
    &= \left(\prod_{i=1}^d e^{(k_i-k_i^2+l_i-l_i^2)\frac{t}{2}} e^{((k_i+l_i)^2 - (k_i+l_i))\frac{t}{2}}\right)\left(\prod_{i\neq j}^d e^{(k_i-k_i^2+l_j-l_j^2)\frac{t}{2}}e^{k_i(k_i-1)\frac{t}{2}}e^{l_j(l_j-1)\frac{t}{2}} \right)\\
    &= \prod_{i=1}^d e^{ k_il_it}
\end{align*}
using that the $n$-th moment of a log-normally distributed variable $Y$ is $\mathbb{E}[Y^n] = e^{n(n-1)\frac{t}{2}}$~\cite[Equation~(14.7)]{johnson1995continuous}.
This means that the matrices $A_t$, $t\in[T-1]$, with entries $(A_t)_{\vec{k},\vec{l}} = \mathbb{E}[e_{t,\vec{k}}(X_t)e_{t,\vec{l}}(X_t)]$ are given by
\begin{align}
    \label{eq:vandermonde}
    \begin{bmatrix}
        1 & 1 & 1 & \cdots & 1\\
        1 & e^{t} & e^{2t} & \cdots & e^{(q-1)t}\\
        1 & e^{2t} & e^{4t} & \cdots & e^{2(q-1)t}\\
        \vdots & \vdots & \vdots & \ddots & \vdots\\
        1 & e^{(q-1)t} & e^{2(q-1)t} & \cdots & e^{(q-1)^2t}
    \end{bmatrix}^{\otimes d},
\end{align}
i.e., they are tensor power of Vandermonde matrices. By using the following result (already generalized for $d>1$), we can bound the minimum singular value of $A_t$.
\begin{lemma}[{\cite[Lemma~4]{glasserman2004number}}]
    \label{lem:singular_bound_geo}
    The minimum singular value $\sigma_{\min}(A_t)$ of the matrix $A_t$ in Eq.~\eqref{eq:vandermonde} is such that
    \begin{align*}
        \frac{1}{\sigma_{\min}(A_t)} \leq e^{2ed/(e^t-1)^2}(q-1)^dq^d\left(\frac{e^t}{e^t - 1}\right)^{(q-1)d} &\leq \left(e^{2e/(e-1)^2}\frac{2e}{e - 1}\right)^{(q-1)d}(q-1)^{2d}\\
        &\leq e^{3(q-1)d}(q-1)^{2d}.
    \end{align*}
\end{lemma}

We shall use the matrices in Eq.~\eqref{eq:vandermonde} within our Algorithm~\ref{alg:quantumlongstaff}, similarly to what we did with the identity matrix in the Brownian motion section. Once again, in order to use the Jackson-like inequality in Lemma~\ref{lem:differentiable_polynomials}, we must restrict the domain to a close set like $Q_d(\lambda)$. We consider the approximation architecture generated by $\{e^{\rm trunc}_{t,\vec{k}}:\mathbb{R}^d\to \mathbb{R}\}_{\vec{k}\in\{0,\dots,q-1\}^d}$ with $e^{\rm trunc}_{t,\vec{k}}(x) = \prod_{i=1}^d x_i^{k_i}e^{-k_i(k_i-1)\frac{t}{2}}$ if $x\in Q_d(\lambda)$ and $e^{\rm trunc}_{t,\vec{k}}(x) = 0$ otherwise. The implied matrix $A^{\rm trunc}_t$ is no longer equal to the matrix in Eq.~\eqref{eq:vandermonde}, but is close to it. Such closeness will be bounded with the aid of the following simple result. The vector $b_t$, $t\in[T-1]$, will still be approximated using quantum algorithm for Monte Carlo, which, in turn, must use a discretization of the geometric Brownian motion.
\begin{lemma}
    \label{lem:geo_bound}
    Given $\lambda > 0$ and $k\in\mathbb{N}$, then
    \begin{align*}
    	\int_\lambda^\infty \frac{x^{k-1}}{\sqrt{2\pi t}}e^{-\frac{(\ln{x}+t/2)^2}{2t}}\text{d}x \leq \frac{1}{2}e^{\frac{t}{2}k(k-1) - \frac{1}{2t}\big({\ln\lambda} - t(k-\frac{1}{2})\big)^2}.
    \end{align*}
\end{lemma}
\begin{proof}
    \begin{align*}
        \int_\lambda^\infty \frac{x^{k-1}}{\sqrt{2\pi t}}e^{-\frac{(\ln{x}+t/2)^2}{2t}}\text{d}x = \frac{e^{k(k-1)\frac{t}{2}}}{\sqrt{\pi}}\int_{\frac{\ln\lambda}{\sqrt{2t}} - \sqrt{\frac{t}{2}}(k-\frac{1}{2})}^\infty e^{-x^2}\text{d}x \leq \frac{1}{2}e^{\frac{t}{2}k(k-1) - \frac{1}{2t}\big({\ln\lambda} - t(k-\frac{1}{2})\big)^2},
    \end{align*}
    using the change of variables $x\to e^{\sqrt{2t}x+t(k-\frac{1}{2})}$ and that $\frac{2}{\sqrt{\pi}}\int_\lambda^\infty e^{-x^2}\text{d}x \leq e^{-\lambda^2}$~\cite{simon1998some,chiani2002improved}.
\end{proof}

Recall the definition of the truncation operator $\mathscr{T}_\beta$ and the quantity
\begin{align}\label{eq:r_p-def2}
    r_p := T\sqrt{\frac{2M_p}{p-2}} + \operatorname*{\max}_{t\in[T-1]}\sum_{i=1}^d \mathbb{E}[|X_{t,i}|] = T\sqrt{\frac{2M_p}{p-2}} + d
\end{align}
for $2< p < \infty$, where $M_p := \max\big\{\|z_1\|_{L^p(\rho_1)}^p,\dots,\|z_T\|_{L^p(\rho_T)}^p\big\}$ and we used that $\mathbb{E}[|X_{t,i}|] = 1$. For simplicity, the next theorem only considers the case when the continuation values are in $C^n(Q_d(\lambda))$. Moreover, we stress that $L^2(\rho_t) = L^2(\mathbb{R}^d, \rho_t)$ considers the whole space $\mathbb{R}^d$.
\begin{theorem}
    \label{thr:geo_brownian}
    Consider a Markov chain $(X_t)_{t=0}^T$ where each component of $X_t:\Omega\to \mathbb{R}^d$ follows an independent geometric Brownian motion. Let $Q_d(\lambda)$ be the closed cube with side length $2\lambda$ in $\mathbb{R}^d$ centered at the origin. Suppose $M_p := \max\big\{\|z_1\|_{L^p(\rho_1)}^p,\dots,\|z_T\|_{L^p(\rho_T)}^p\big\} < \infty$ for some $2 < p < \infty$ and let $r_p>0$ as in Eq.~\eqref{eq:r_p-def2}. Within the setting of Algorithm~{\rm \ref{alg:quantumlongstaff}} with parameters $\delta$ and $\epsilon_0$, let the sampling cost of the Markov chain be $T\mathcal{T}_{\rm samp}$, and set all matrices $\widetilde{A}_t$, $t\in[T-1]$, to be the matrix in Eq.~\eqref{eq:vandermonde} and use the payoffs $z_0,\mathscr{T}_\beta z_1,\dots,\mathscr{T}_\beta z_T$, where $\beta := \lambda^{2/p}$, with query cost $\mathcal{T}_{z}$ and $R := \max_{t\in[T]}\|\mathscr{T}_\beta z_t\|_u$. Also, for $t\in[T-1]$, consider the approximation architectures $\mathscr{H}_t = \{e^{\rm trunc}_{t,\vec{k}}:\mathbb{R}^d\to\mathbb{R}\}_{\vec{k}\in\{0,\dots,q-1\}^d}$ for $q\geq 2$, where $e^{\rm trunc}_{t,\vec{k}}(x) = \prod_{i=1}^d x_i^{k_i}e^{-k_i(k_i-1)\frac{t}{2}}$ if $x\in Q_d(\lambda)$ and $e^{\rm trunc}_{t,\vec{k}}(x) = 0$ otherwise. Let $\mathcal{T}_{e}$ be the query cost of $\mathscr{H}_t$, $t\in[T]$. Suppose there is $n\in\mathbb{N}$, $n<q$, such that the restriction $\mathbb{E}[Z_{\tau_{t+1}}|X_t]|_{Q_d(\lambda)}$ to $Q_d(\lambda)$ of each $\mathbb{E}[Z_{\tau_{t+1}}|X_t]$, $t\in\{0,\dots,T-1\}$, is in $C^n(Q_d(\lambda))$. Let $C \geq C_{d,\lambda}\max_{0\leq t<T}\|\mathbb{E}[Z_{\tau_{t+1}}|X_t]|_{Q_d(\lambda)}\|_{C^n(Q_d(\lambda))}$ be known, where $C_{d,\lambda}$ is a constant depending only on $Q_d(\lambda)$. Then, for $\delta\in(0,1)$, $\epsilon\in\big(0,5^T(2^{d}C^{3d/n}R)^{n/(n+3d)}\big]$, $q = \left \lceil(5^T C/\epsilon)^{1/n} \right \rceil$ and 
    \begin{align}\label{eq:geo_lamba_bounds}
    \resizebox{0.92\hsize}{!}{
        $\ln\lambda = O\left(\max \left\{T\sqrt{d}\left(\frac{5^TC}{\epsilon}\right)^{\frac{1}{n}}, \sqrt{\ln\left(2^{2d}C^{\frac{6d}{n}}e^{27Td(5^TC/\epsilon)^{2/n}}R\left(\frac{5^T}{\epsilon}\right)^{1+\frac{6d}{n}} \right)}, \ln\left(\frac{5^T r_p}{\epsilon}\right)^{\frac{p}{p-2}}\right\}\right),$
        }
    \end{align} 
    if $\epsilon_0 = \frac{\exp(-7Td(5^TC/\epsilon)^{2/n})}{2^{d+2}C^{5d/n}R}\left(\frac{\epsilon}{5^T}\right)^{1+5d/n}$, Algorithm~{\rm \ref{alg:quantumlongstaff}} outputs $\widetilde{\mathcal{U}}_0 := \max\big\{Z_0,\widetilde{(\mathscr{T}_\beta Z)}_{\widetilde{\tau}_1}\big\}$ such that
    \begin{align*}
        \operatorname{Pr}\left[|\widetilde{\mathcal{U}}_0 - \mathbb{E}[Z_{\tau_0}]| \leq 3\epsilon \right] \geq 1 - \delta
    \end{align*}
    in time 
    \begin{align*}
        \widetilde{O}\left(\mathcal{T}_{\rm total}2^{7n/2}C^{15d/2n+3/n}e^{\frac{15}{2}Td(5^TC/\epsilon)^{2/n}} R^2\left(\frac{5^T}{\epsilon}\right)^{1+15d/2n+3/n}\right),
    \end{align*}
    where $\mathcal{T}_{\rm total} := \mathcal{T}_{\rm samp} + \mathcal{T}_{z} + \mathcal{T}_{e}$ and $\widetilde{O}(\cdot)$ hides $\operatorname{polylog}$ terms in $5^T,R,C,e^d,\epsilon,\delta$.
\end{theorem}
\begin{proof}
    Let $\sigma_{\min} \leq \min_{t\in[T-1]}\sigma_{\min}(\widetilde{A}_t)$, where $\widetilde{A}_t$ is the matrix from Eq.~\eqref{eq:vandermonde}, $t\in[T-1]$ and let $L := \max_{t\in[T-1],\vec{k}\in\{0,\dots,q-1\}^d} \|e^{\rm trunc}_{t,\vec{k}}\|_{L^2(\rho_t)} \leq e^{T(q-1)^2d/2}$. First we note that Algorithm~\ref{alg:quantumlongstaff}, when using the approximation architecture $\{e^{\rm trunc}_{t,\vec{k}}:\mathbb{R}^d\to\mathbb{R}\}_{\vec{k}\in\{0,\dots,q-1\}^d}$ and parameters $\epsilon_b = \epsilon_0/\sqrt{m}$ and $\delta_b = \delta/(2Tm)$ with $m:=q^d$ (there are no $\epsilon_A$ and $\delta_A$ parameters), outputs $\widetilde{(\mathscr{T}_\beta Z)}_{\widetilde{\tau}_1}$ such that, for any $\epsilon_0\in(0,\sigma_{\min}/2]$,
    \begin{align*}
        \operatorname{Pr}\left[\big|\widetilde{(\mathscr{T}_\beta Z)}_{\widetilde{\tau}_1} - \mathbb{E}[(\mathscr{T}_\beta Z)_{\tau_1}]\big| \geq 5^T\left(\frac{4\epsilon_0 mRL^2}{\sigma_{\min}^2} + \operatorname*{\max}_{0< t<T} \operatorname*{\min}_{a\in\mathbb{R}^{m}}\|a\cdot \vec{e}_t(X_t) - \mathbb{E}[(\mathscr{T}_\beta Z)_{\tau_{t+1}}|X_t]\|_{L^2(\rho_t)}\right) \right] \leq \delta,
    \end{align*}
    where $\vec{e}_t(\cdot) := (e^{\rm trunc}_{t,\vec{k}}(\cdot))_{\vec{k}\in\{0,\dots,q-1\}^d}$, in time
    \begin{align}
        \label{eq:inter_complex2}
        O\left(\frac{T^2m^{5/2}}{\epsilon_0}\mathcal{T}_{\rm total}RL\log(T)\log(Tm/\delta) \log^{3/2}(mRL/\epsilon_0)\log\log(mRL/\epsilon_0) \right).
    \end{align}
    The proof is exactly the same as the ones from Theorems~\ref{thr:quantum_result1} and~\ref{thr:quantum_result2}, the only difference being that the payoffs are now $z_0,\mathscr{T}_\beta z_1,\dots,\mathscr{T}_\beta z_T$ and we must show that the $L^2$-error between $A^{\rm trunc}_t$ and the matrix $\widetilde{A}_t$ from Eq.~\eqref{eq:vandermonde} is $\|A^{\rm trunc}_t-\widetilde{A}_t\|_2 \leq \epsilon_0$. Indeed, the error commited in any entry $\mathbb{E}[e^{\rm trunc}_{t,\vec{k}}(X_t)e^{\rm trunc}_{t,\vec{k}}(X_t)]$, when approximating the matrix $A^{\rm trunc}_t$ by the one in Eq.~\eqref{eq:vandermonde}, is at most
    \begin{align}
        \prod_{i=1}^d e^{k_il_i t} - \prod_{i=1}^d \Bigg(e^{k_il_i t} - &e^{(k_i-k_i^2+l_i-l_i^2)\frac{t}{2}}\int_\lambda^\infty \frac{x^{k_i+l_i-1}}{\sqrt{2\pi t}}e^{-\frac{(\ln{x}+t/2)^2}{2t}}\text{d}x\Bigg)\nonumber\\
        &\leq \left(\prod_{i=1}^d e^{k_il_it} \right)\sum_{j=1}^d e^{((k_j+l_j) - (k_j+l_j)^2)\frac{t}{2}}\int_\lambda^\infty \frac{x^{k_j+l_j-1}}{\sqrt{2\pi t}}e^{-\frac{(\ln{x}+t/2)^2}{2t}}\text{d}x \label{eq:geo_step1}\\
        &\leq \left(\prod_{i=1}^d e^{k_il_it} \right)\sum_{j=1}^d \frac{1}{2}e^{-\frac{1}{2t}\big({\ln\lambda} - t(k_j+l_j-\frac{1}{2})\big)^2}\label{eq:geo_step2}\\
        &\leq \frac{d}{2}e^{(q-1)^2td - \frac{1}{2t}\big({\ln\lambda} - t(2q-\frac{5}{2})\big)^2},\nonumber
    \end{align}
    where Eq.~\eqref{eq:geo_step1} follows from convexity and Eq.~\eqref{eq:geo_step2} from Lemma~\ref{lem:geo_bound}. Observe now that, if $\ln\lambda \geq t(2q-\frac{5}{2}) + \sqrt{2(q-1)^2t^2d + \ln(md/2\epsilon_0)}$, then the above error is at most $\epsilon_0/m$, which implies $\|A^{\rm trunc}_t-\widetilde{A}_t\|_2 \leq \epsilon_0$. Regarding the complexity, the most expensive steps are inverting $\widetilde{A}_t$, which requires time $O(Tm^{\omega_\ast})$, and computing all $b_t$, which requires time
    \begin{align*}
        O\left(\frac{Tm}{\epsilon_b}RL(T\mathcal{T}_{\rm samp} + T(\log(T)\mathcal{T}_{z} + m\mathcal{T}_{e}))\log\left(\frac{1}{\delta_b}\right) \log^{3/2}\left(\frac{RL}{\epsilon_b} \right)\log\log\left(\frac{RL}{\epsilon_b}\right) \right),
    \end{align*}
    by calling $\mathtt{QMonteCarlo}(Z_{\widetilde{\tau}_{t}} e_{t-1,\vec{k}}(X_{t-1}), \epsilon_b, \delta_b, RL)$ from Theorem~\ref{thr:quantum_monte_carlo} for $t\in\{2,\dots,T\}$ and $\vec{k}\in\{0,\dots,q-1\}^d$. The term $T(\log(T)\mathcal{T}_{z} + m\mathcal{T}_{e})$ comes from the unitaries $C_t^{(\vec{k})}$ in $\mathtt{QMonteCarlo}$. Hence, by keeping the largest terms of each complexity and using that $\epsilon_b = \epsilon_0/\sqrt{m}$ and $\delta_b = \delta/(4Tm)$, the final complexity is the one from Eq.~\eqref{eq:inter_complex2}.
    
    With this partial result in hands, we again use the Jackson-type inequality from Lemma~\ref{lem:differentiable_polynomials} as in Theorem~\ref{thr:quantum_result3}. In order to do so, we first need the following technical result that limits the continuation values to being a function in the domain $Q_d(\lambda)$. More specifically, by~\cite[Proposition~3.14]{zanger2013quantitative},
    \begin{align}
        \label{eq:truncation_inequality_geo}
        \|\mathbf{1}_{Q_d(\lambda)}\mathbb{E}[Z_{\tau_{t+1}}|X_t] - \mathbb{E}[(\mathscr{T}_\beta Z)_{\tau_{t+1}}|X_t]\|_{L^2(\rho_t)} \leq r_p\lambda^{2/p-1},
    \end{align}
    where $\mathbf{1}_{Q_d(\lambda)}$ is the indicator random variable corresponding to the set $Q_d(\lambda)$. Since $\mathcal{R}_q(Q_d(\lambda))$ is spanned by $\{e^{\rm trunc}_{t,\vec{k}}\}_{\vec{k}\in\{0,\dots,q-1\}^d}$, the Jackson-type inequality from Lemma~\ref{lem:differentiable_polynomials} reads
    \begin{align*}
        \operatorname*{\min}_{a\in\mathbb{R}^{m}}\|a\cdot \vec{e}_t(X_t) - \mathbf{1}_{Q_d(\lambda)}\mathbb{E}[Z_{\tau_{t+1}}|X_t]\|_{L^2(\rho_t)} = \operatorname*{\inf}_{P\in\mathcal{R}_q}\|(P - \mathbb{E}[Z_{\tau_{t+1}}|X_t])\mathbf{1}_{Q_d(\lambda)}\|_{L^2(\rho_t)} \leq Cq^{-n},
    \end{align*}
    for all $t\in\{0,\dots,T-1\}$. Combining both results, we finally get that
    \begin{align*}
        \big|\widetilde{(\mathscr{T}_\beta Z)}_{\widetilde{\tau}_1} - \mathbb{E}[(\mathscr{T}_\beta Z)_{\tau_1}]\big| \leq 5^T\left(\frac{4\epsilon_0mRL^2}{\sigma_{\min}^2} + Cq^{-n} + r_p\lambda^{2/p-1}\right),
    \end{align*}
    with probability at least $1-\delta$. Before we plug values for $m$, $q$, $\epsilon_0$ and $\epsilon$, we get rid of the truncation operator $\mathscr{T}_\beta$ in the inequality above by using that~\cite[Proposition~5.2]{egloff2005monte}
    \begin{align*}
        \|\mathbb{E}[Z_{\tau_{t+1}}|X_t] - \mathbb{E}[(\mathscr{T}_\beta Z)_{\tau_{t+1}}|X_t]\|_{L^2(\rho_t)} \leq r_p\lambda^{1-p/2}
    \end{align*}
    (which is quite similar to Eq.~\eqref{eq:truncation_inequality_geo}). By using that $|{\max}\{a_0,a_1\}-\max\{a_0,a_2\}|\leq |a_1-a_2|$ with $a_0,a_1,a_2\in\mathbb{R}$, then
    \begin{align*}
        |\widetilde{\mathcal{U}}_0 - \mathbb{E}[Z_{\tau_0}]| \leq \big|\widetilde{(\mathscr{T}_\beta Z)}_{\widetilde{\tau}_1} - \mathbb{E}[Z_{\tau_1}]\big| \leq 5^T\left(\frac{4\epsilon_0mRL^2}{\sigma_{\min}^2} + Cq^{-n} + 2r_p\lambda^{2/p-1}\right),
    \end{align*}
    since $\lambda^{1-p/2} \leq \lambda^{2/p-1}$.
    
    We now use $q=\lceil(5^T C/\epsilon)^{1/n}\rceil$, $m = q^d \leq 2^d(5^T C/\epsilon)^{d/n}$, and $\lambda \geq \big(\frac{2r_p 5^T}{\epsilon}\big)^{p/(p-2)}$. We also take 
    \begin{align*}
        \epsilon_0 = \frac{e^{-7Td(5^TC/\epsilon)^{2/n}}}{2^{d+2}C^{5d/n}R}\left(\frac{\epsilon}{5^T}\right)^{1+5d/n} \leq \frac{\sigma_{\min}^2}{2^{d+2} C^{d/n} RL^2}\left(\frac{\epsilon}{5^T}\right)^{1+d/n},
    \end{align*} 
    where we used that $\sigma_{\min}^{-2} \leq e^{6(q-1)d}(q-1)^{4d} \leq e^{6d(5^TC/\epsilon)^{1/n}}(5^TC/\epsilon)^{4d/n}$ according to Lemma~\ref{lem:singular_bound_geo} and the bound $L \leq e^{T(q-1)^2d/2} \leq e^{\frac{Td}{2}(5^TC/\epsilon)^{2/n}}$. These bounds lead to the final error
    \begin{align*}
        5^T\left(\frac{4\epsilon_0mRL^2}{\sigma_{\min}^2} + \frac{C}{q^n} + 2r_p\lambda^{2/p-1} \right) \leq 5^T\left(\frac{2^{d+2}\epsilon_0RL^2}{\sigma_{\min}^2}\left(\frac{5^T C}{\epsilon}\right)^{d/n} + \frac{C}{q^n} + 2r_p\lambda^{2/p-1} \right) \leq 3\epsilon.
    \end{align*}
    The above in turn implies the complexity, up to polylog factors,
    \begin{align*}
        \widetilde{O}\left(\frac{T^2m^{5/2}}{\epsilon_0}RL\mathcal{T}_{\rm total}\log^{3/2}\left(\frac{L}{\epsilon_0}\right) \right) = \widetilde{O}\left(\mathcal{T}_{\rm total}2^{7d/2}C^{15d/2n+3/n}e^{\frac{15}{2}Td(5^TC/\epsilon)^{2/n}} R^2\left(\frac{5^T}{\epsilon}\right)^{1+15d/2n + 3/n}\right),
    \end{align*}
    where we again used that $L \leq e^{\frac{Td}{2}(5^TC/\epsilon)^{2/n}}$. Further, regarding the bound on $\lambda$, we note that 
    \begin{align*}
        \ln\lambda &\geq \max\{4(q-1)T\sqrt{d},3\sqrt{\ln(md/2\epsilon_0)}\} \\
        &= O\left(\max \left\{T\sqrt{d}\left(\frac{5^TC}{\epsilon}\right)^{1/n}, \sqrt{\ln\left(2^{2d}C^{6d/n}e^{7Td(5^TC/\epsilon)^{2/n}}R\left(\frac{5^T}{\epsilon}\right)^{1+6d/n} \right)} \right\}\right)
    \end{align*}
    implies $\ln\lambda \geq t(2q-\frac{5}{2}) + \sqrt{2(q-1)^2t^2d + \ln(md/2\epsilon_0)}$. Finally, we note that $\epsilon \leq 5^T(2^{d}C^{3d/n}R)^{n/(n+3d)}$ implies, as required, that
    \begin{align*}
        \epsilon_0 = \frac{e^{-7Td(5^TC/\epsilon)^{2/n}}}{2^{d+2}C^{5d/n}R}\left(\frac{\epsilon}{5^T}\right)^{1+5d/n} &\leq \frac{e^{-3d(5^TC/\epsilon)^{1/n}}}{2^{d+2}C^{5d/n}R}\left(\frac{\epsilon}{5^T}\right)^{1+5d/n}\\
        &\leq \frac{1}{2}e^{-3d(5^TC/\epsilon)^{1/n}} \left(\frac{\epsilon}{5^TC}\right)^{2d/n} \leq \frac{\sigma_{\min}}{2}. \qedhere
    \end{align*}
\end{proof}

A similar result can be obtained for the classical {\rm LSM} algorithm. Similarly to Theorems~\ref{thr:classical_error} and~\ref{thr:geo_brownian}, Algorithm~\ref{alg:classicalalgo}, by using the same bounds on $\lambda$ (Eq.~\eqref{eq:geo_lamba_bounds}), outputs $\widetilde{(\mathscr{T}_\beta Z)}_{\widetilde{\tau}_1}$ such that, for any $\epsilon_0\in(0,\sigma_{\min}/2]$,
\begin{align*}
\begin{multlined}[\textwidth][b]
    \operatorname{Pr}\left[\big|\widetilde{(\mathscr{T}_\beta Z)}_{\widetilde{\tau}_1} - \mathbb{E}[(\mathscr{T}_\beta Z)_{\tau_1}]\big| \geq 5^T\left(\frac{4\epsilon_0mRL^2}{\sigma_{\min}^2} + \operatorname*{\max}_{0< t<T} \operatorname*{\min}_{a\in\mathbb{R}^{m}}\|a\cdot \vec{e}_t(X_t) - \mathbb{E}[(\mathscr{T}_\beta Z)_{\tau_{t+1}}|X_t]\|_{L^2(\rho_t)}\right) \right] \\
    \leq 4m\exp\big(-2N\epsilon_0^2/(mR^2L^2)\big),
\end{multlined}
\end{align*}
in time $O(TmN + Tm^{\omega_\ast} + TN(T + \mathcal{T}_{\rm samp} + \mathcal{T}_{z} + m\mathcal{T}_{e})) = O(TN(T + m\mathcal{T}_{\rm total}) + Tm^{\omega_\ast})$, where $O(TmN)$ comes from computing $\widetilde{b}_t$ and $O(Tm^{\omega_\ast})$ from inverting $\widetilde{A}_t$ and computing $\widetilde{\alpha}_t = \widetilde{A}_t^{-1}\widetilde{b}_t$, $t\in[T-1]$. The application of the Jackson-like inequality is the same as in the previous theorem. By taking $N = \lceil \frac{mR^2L^2}{2\epsilon_0^2}\log(4m/\delta)\rceil$ and the same parameters $q=\lceil(5^T C/\epsilon)^{1/n}\rceil$, $m = q^d \leq 2^d (5^T C/\epsilon)^{d/n}$, and $\epsilon_0 \leq \frac{\sigma_{\min}^2}{2^{d+2} C^{d/n} RL^2}\big(\frac{\epsilon}{5^T}\big)^{1+d/n}$ as before, the final complexity is
\begin{align*}
    O(TN(T + m\mathcal{T}_{\rm total}) + Tm^{\omega_\ast}) &= \widetilde{O}\left(T\left(\frac{mR^2L^2}{\epsilon_0^2}(T+ m\mathcal{T}_{\rm total}) + m^{\omega_\ast}\right) \right) \\
    &= \widetilde{O}\left(\mathcal{T}_{\rm total}2^{4d} C^{4d/n}\frac{R^4L^6}{\sigma_{\min}^4}\left(\frac{5^T}{\epsilon}\right)^{2+4d/n}\right)\\
    &= \widetilde{O}\left(\mathcal{T}_{\rm total}2^{4d}C^{12d/n}e^{15Td(5^TC/\epsilon)^{2/n}} R^4\left(\frac{5^T}{\epsilon}\right)^{2+12d/n}\right),
\end{align*}
up to the $\log(m/\delta)$ factor, using that $\sigma_{\min}^{-4} \leq e^{12(q-1)d}(q-1)^{8d} \leq e^{12d(5^TC/\epsilon)^{1/n}}(5^TC/\epsilon)^{8d/n}$ according to Lemma~\ref{lem:singular_bound_geo} and the bound $L \leq e^{T(q-1)^2d/2} \leq e^{\frac{Td}{2}(5^T C/\epsilon)^{2/n}}$.

Even though the geometric Brownian motion setting does not simplify the {\rm LSM} algorithm as much as the Brownian motion setting from the previous section, the minimum singular value from the matrices $\widetilde{A}_t$ can still be bounded. The final classical and quantum complexities are $\widetilde{O}\big(e^{15Td(5^T C/\epsilon)^{2/n}}(5^T/\epsilon)^{2+12d/n}\big)$ and $\widetilde{O}\big(e^{\frac{15}{2}Td(5^T C/\epsilon)^{2/n}}(5^T/\epsilon)^{1+15d/2n+3/n}\big)$, respectively. We can see that the bound on the minimum singular value in Lemma~\ref{lem:singular_bound_geo}, as well as the $L^2$-norm of the chosen functions $e_{t,\vec{k}}$, are very sensitive to the degree $q$ of the polynomials spanning our approximation architecture: they lead to the exponential factor $e^{O(Td(5^T/\epsilon)^{2/n})}$ on $\epsilon$, which greatly deteriorates the polynomial dependence on $\epsilon$ that was previously observed in previous sections. Nonetheless, if the continuation values are sufficiently smooth, i.e., $n$-differentiable for $n=\Theta(\log(5^T/\epsilon)/\log\log(5^T/\epsilon))$, then such an exponential factor becomes $e^{O(Td\log^c(5^T/\epsilon))}$ for any constant $0<c<1$, and we recover classical and quantum complexities which are close to the expected $\widetilde{O}(1/\epsilon^2)$ and $\widetilde{O}(1/\epsilon)$, respectively.

\section{Conclusions}

In this work, we developed a new quantum algorithm for a stochastic optimal stopping problem (as in Theorem~\ref{thr:newdynamic}) with a quantum advantage in the runtime. This problem cannot be solved accurately by a single application of quantum algorithms for Monte Carlo~\cite{montanaro2015quantum,hamoudi2021quantum,cornelissen2021quantum,cornelissen2021bquantum,an2021quantum}. 
Instead one must compute in superposition (and recursively) the stopping times as in Lemma~\ref{lemQStop}, which is key to obtaining a quantum speedup. 
As the classical {\rm LSM} algorithm can be used to solve a large variety of problems, our quantum {\rm LSM} can also be used for problems in finance including insurance~\cite{krah2018least} and risk management~\cite{Green2015}, and for many optimization problems outside finance, such as quickest detection~\cite{shiryaev2010quickest} and sequential Bayesian hypothesis testing~\cite{daskalakis2017optimal}. Additionally, we believe that there are many other problems in, for example, dynamic programming, stochastic optimal stopping and optimal control where the interplay of function approximation and quantum subroutines for Monte Carlo could be used to design new quantum algorithms.

A few design choices of the quantum algorithm were guided by real problems where the classical algorithm is already used. Even though we took number of expansion functions $m={\rm poly}(5^T/\epsilon)$ in order to bound the approximation error, in practice one typically assumes $m$ to be constant~\cite{longstaff2001valuing}. For big values of $m$, our algorithm could be modified in order to use quantum subroutines for inner product estimation, and reduce the complexity polynomially in $m$, but introducing a further $\epsilon$ dependence. Thus, further analysis is needed to understand the impact of the precision parameters on the runtime of these subroutines. Along these lines, we have chosen to invert the linear systems for finding $\alpha_t$ on a classical computer. A possible modification of our algorithm could output quantum states $|\alpha_t\rangle$ via HHL-like algorithms~\cite{HHL09}. We also discussed how, under the hypothesis that the Markov chain is a Brownian motion or a Geometric Brownian motion, the matrices $A_t$ can be expressed with a closed formula and their minimum singular value be bounded. This idea exhibits some similarity with the idea proposed in~\cite{liu2019american}. There, they leveraged quasi-regression algorithms and a particular choice of expansion functions~\cite{owen2000assessing}, so to pre-compute the matrices $A_t$, and thus skip costly Monte Carlo estimators. Moreover, when considering a Brownian or a Geometric Brownian motion, the chosen functions $\{e_{t,k}\}_{k=1}^m$ had a explicit time dependence on $t$, but it is possible to transform the optimal stopping time problem behind American option pricing in a way that such dependence is suppressed and a single set of approximating functions is employed~\cite{broadie2000american,broadie2000nonparametric}. We believe that such approach could improve our complexities. Finally, in our algorithm, we use quantum subroutines from~\cite{montanaro2015quantum}, but could equivalently use the subroutines from~\cite{hamoudi2021quantum,cornelissen2021quantum,cornelissen2021bquantum}. Our template could be extended to quantum algorithms that are similar in spirit but are solving different problems~\cite{tsitsiklis2001regression}. 

Our final complexities have an exponential dependence on $T$, the number of time steps. We believe that such dependence, present in several past works~\cite{egloff2005monte,zanger2009convergence,zanger2013quantitative,zanger2018convergence,zanger2020general,miyamoto2021bermudan}, is only a consequence of a loose error bound and could possibly be improved. Such hope is backed up by the ubiquitous employment of {\rm LSM} algorithms for pricing American options in every day financial markets. We also note that a more careful error analysis would improve classical results as well, but, regardless, the quantum quadratic improvement would still be present. Finally, notice that it always possible to compensate a reduction on the number of time steps with more accurate approximations for continuation values and similar quantities. 

We stress the importance of fast quantum algorithms for optimal stopping problems. For American option pricing, the value of the payoff function could easily reach a few million dollars, and the additive precision $\epsilon$ could be of the order of $10^{-11}$~\cite{andersen2015high}. Given the level of specialization in classical algorithms and architectures for this specific problem, how and when our algorithm can find application in practice is left for future work.

\section*{Acknowledgements}

Research at CQT is funded by the National Research Foundation, the Prime Minister’s Office, and the Ministry of Education, Singapore under the Research Centres of Excellence programme’s research grant R-710-000-012-135. We thank Rajagopal Raman for pointing out Ref.~\cite{krah2018least} and the importance of the Longstaff-Schwarz algorithm in the insurance industry. We also thank Koichi Miyamoto for Ref.~\cite{PhysRevA.104.022430} and Eric Ghysels for Refs.~\cite{broadie2000american,broadie2000nonparametric}. Special thanks to Armando Bellante for spotting several typos and making us notice the $T^2$ dependence in the classical LSM complexity.

\printbibliography

\appendix

\section{Proof of Lemma~\ref{lem:hermite_bound}}
\label{sec:hermite_bound}

\begin{lemma}
    Given $\lambda > 0$ and $k,l\in\mathbb{N}$ with $l\leq k$, then
    \begin{align*}
        \left|\int_\lambda^\infty H_k(x) H_l(x)e^{-x^2}\text{d}x\right| \leq \begin{dcases}
               k!2^k\frac{\sqrt{2}}{\pi}\operatorname{erfc}(\lambda) + e^{-\lambda^2}H_{k+1}(\lambda)^2 \quad &\text{if}~k=l,\\
               e^{-\lambda^2}|H_{l+1}(\lambda)H_k(\lambda)| \quad &\text{if}~k\neq l,
        \end{dcases}
    \end{align*}
    where $\operatorname{erfc}(\lambda) = \frac{2}{\sqrt{\pi}}\int_\lambda^\infty e^{-x^2}\text{d}x$ is the complementary error function. In particular,
    \begin{align*}
    	\left|\int_\lambda^\infty H_k(x)H_l(x) e^{-x^2}\text{d}x\right| \leq 2^{2+(k+l)/2}\sqrt{(k+1)!(l+1)!}e^{(\sqrt{2(k+1)}+\sqrt{2(l+1)})\lambda}e^{-\lambda^2}.
    \end{align*}
\end{lemma}
\begin{proof}
    In the following, by $H^{(j)}_k(x)$ we denote the $j$-th derivative of $H_k(x)$. By plugging the definition of $H_l(x)$, it is not hard to observe the following using integration by parts:
    \begin{align*}
    	\int_\lambda^\infty H_k(x) H_l(x)e^{-x^2}\text{d}x &= (-1)^k\int_\lambda^\infty H_l(x)\frac{\text{d}^k}{\text{d}x^k}e^{-x^2}\text{d}x\\
    	&= e^{-\lambda^2}\sum_{j=0}^{k-1} H_l^{(j)}(\lambda)H_{k-j-1}(\lambda) + \int_\lambda^\infty H_l^{(k)}(x)e^{-x^2}\text{d}x\\
    	&= e^{-\lambda^2}\sum_{j=0}^{k-1} H_l^{(j)}(\lambda)H_{k-j-1}(\lambda) + \delta_{kl}2^kk!\int_\lambda^\infty e^{-x^2}\text{d}x\\
    	&= e^{-\lambda^2}\sum_{j=0}^{k-1} H_l^{(j)}(\lambda)H_{k-j-1}(\lambda) + \delta_{kl}2^kk!\frac{\sqrt{2}}{\pi}\operatorname{erfc}(\lambda)\\
    	&= e^{-\lambda^2}\sum_{j=0}^{\min(l,k-1)} 2^j\frac{l!}{(l-j)!} H_{l-j}(\lambda)H_{k-j-1}(\lambda) + \delta_{kl}2^kk!\frac{\sqrt{2}}{\pi}\operatorname{erfc}(\lambda),
    \end{align*}
    where we used that $H^{(j)}_l(x) = 2^j \frac{l!}{(l-j)!}H_{l-j}(x)$ for $j\leq l$, which follows from $\frac{\text{d}}{\text{d}x}H_l(x) = 2lH_{l-1}(x)$~\cite[Section~5.6.2]{magnus2013formulas}. If $k=l$, then
    
    \begin{align}
    	\int_\lambda^\infty H_k(x)^2e^{-x^2}\text{d}x &\leq 2^kk!\frac{\sqrt{2}}{\pi}\operatorname{erfc}(\lambda) + e^{-\lambda^2}\sum_{j=0}^{k-1} 2^j\frac{k!}{(k-j)!} |H_{k-j}(\lambda)H_{k-j-1}(\lambda)|\nonumber\\
    	&\leq 2^kk!\frac{\sqrt{2}}{\pi}\operatorname{erfc}(\lambda) + e^{-\lambda^2}\sqrt{\left(\sum_{j=0}^{k-1} \frac{k!2^j}{(k-j)!} H_{k-j}(\lambda)^2\right)\left(\sum_{j=0}^{k-1} \frac{k!2^j}{(k-j)!} H_{k-j-1}(\lambda)^2\right)}\label{eq:hermite1}\\
    	&\leq 2^kk!\frac{\sqrt{2}}{\pi}\operatorname{erfc}(\lambda) + 2^{k} k!e^{-\lambda^2}\sum_{j=0}^{k} \frac{1}{2^j j!} H_{j}(\lambda)^2 \nonumber\\
    	&= 2^kk!\frac{\sqrt{2}}{\pi}\operatorname{erfc}(\lambda) + \frac{1}{2}e^{-\lambda^2}(H_{k+1}(\lambda)^2 - H_k(\lambda)H_{k+2}(\lambda)) \label{eq:hermite2}\\
    	&\leq 2^kk!\frac{\sqrt{2}}{\pi}\operatorname{erfc}(\lambda) + e^{-\lambda^2}H_{k+1}(\lambda)^2,\nonumber
    \end{align}
    where Eq.~\eqref{eq:hermite1} follows from Cauchy-Schwarz inequality and Eq.~\eqref{eq:hermite2} follows from the equality $\frac{1}{2^{k+1}k!}(H_{k+1}(x)^2 - H_k(x)H_{k+2}(x)) = \sum_{j=0}^{k}\frac{1}{2^j j!}H_j(x)^2 > 0$~\cite[Section~5.6.4]{magnus2013formulas}.
    
    If, on the other hand, $k\neq l$, then
    \begin{align*}
    	\left|\int_\lambda^\infty H_k(x)H_l(x)e^{-x^2}\text{d}x\right| &\leq e^{-\lambda^2}\sum_{j=0}^{l} 2^j\frac{l!}{(l-j)!} |H_{l-j}(\lambda)H_{k-j-1}(\lambda)|\\
    	&\leq e^{-\lambda^2}\sqrt{\left(\sum_{j=0}^{l} 2^j\frac{l!}{(l-j)!} H_{l-j}(\lambda)^2\right)\left(\sum_{j=0}^{l} 2^j\frac{l!}{(l-j)!} H_{k-1-j}(\lambda)^2\right)}\\
    	&\leq \sqrt{2^l l!} e^{-\lambda^2}\sqrt{\left(\sum_{j=0}^{l} \frac{1}{2^j j!} H_{j}(\lambda)^2\right)\left(\sum_{j=0}^{k-1}2^j \frac{(k-1)!}{(k-1-j)!} H_{k-1-j}(\lambda)^2\right)}\\
    	&= \sqrt{2^{k+l-1}l!(k-1)!}e^{-\lambda^2}\sqrt{\left(\sum_{j=0}^{l} \frac{1}{2^j j!} H_{j}(\lambda)^2\right)\left(\sum_{j=0}^{k-1} \frac{1}{2^j j!} H_{j}(\lambda)^2\right)}\\
    	&= \frac{1}{2}e^{-\lambda^2}\sqrt{H_{l+1}(\lambda)^2 - H_l(\lambda)H_{l+2}(\lambda)}\sqrt{H_{k}(\lambda)^2 - H_{k-1}(\lambda)H_{k+1}(\lambda)}\\
    	&\leq e^{-\lambda^2}|H_{l+1}(\lambda)H_k(\lambda)|,
    \end{align*}
    where we again used $\frac{1}{2^{k+1}k!}(H_{k+1}(x)^2 - H_k(x)H_{k+2}(x)) = \sum_{j=0}^{k}\frac{1}{2^j j!}H_j(x)^2 \geq 0$~\cite[Section~5.6.4]{magnus2013formulas} and Cauchy-Schwarz inequality. This concludes the first part of the lemma.
    
    For the second part of the lemma, we use that $|H_k(\lambda)| \leq 2^{k/2}\sqrt{k!}e^{\sqrt{2k}|\lambda|}$~\cite{van1990new} and that $\operatorname{erfc}(\lambda) \leq e^{-\lambda^2}$~\cite{simon1998some,chiani2002improved}. Then, if $k = l$,
    \begin{align*}
    	\left|\int_\lambda^\infty H_k(x)^2 e^{-x^2}\text{d}x\right| \leq \frac{\sqrt{2}}{\pi}2^k k! e^{-\lambda^2} + 2^{k+1}(k+1)!e^{2\sqrt{2(k+1)}\lambda}e^{-\lambda^2}
    	\leq 2^{k+2}(k+1)!e^{2\sqrt{2(k+1)}\lambda}e^{-\lambda^2},
    \end{align*}
    and if $l\neq k$, then
    \[
    	\left|\int_\lambda^\infty H_k(x)H_l(x) e^{-x^2}\text{d}x\right| \leq 2^{(k+l+1)/2}\sqrt{k!(l+1)!}e^{(\sqrt{2k} + \sqrt{2(l+1)})\lambda}e^{-\lambda^2}. \qedhere
    \]
\end{proof}

\end{document}